\documentclass[numbook]{svjour3}

\smartqed  % flush right qed marks, e.g. at end of proof

\usepackage[
  paper=a4paper,
  left=25mm,
  right=25mm,
  top=25mm,
  bottom=25mm
]{geometry}

\usepackage[utf8]{inputenc}

\usepackage{bm}
\usepackage{amsmath}
\usepackage{amssymb}
\usepackage{thmtools}

\usepackage{enumerate}

\usepackage{algorithmic}
\usepackage[linesnumbered,titlenumbered,ruled,vlined,resetcount,algosection]{algorithm2e}

\usepackage[usenames,dvipsnames]{xcolor}
\usepackage{graphicx}
\usepackage{float}
\usepackage{color}

\usepackage{longtable}
\usepackage{multirow}
\usepackage{booktabs}
\usepackage{multicol}

\usepackage{xparse}
\usepackage[labelformat=simple]{subcaption}

\usepackage{todonotes}

\usepackage{hyperref}
\hypersetup{colorlinks=true,citecolor=blue, linkcolor=blue, urlcolor=blue}

\usepackage[capitalize,noabbrev]{cleveref}
\usepackage[normalem]{ulem}

\usepackage[style=apa, backend=biber]{biblatex}
\smartqed  % flush right qed marks, e.g. at end of proof

\DeclareMathOperator*{\argmax}{arg\,max}

\newcommand \whole {\mathbb{Z}}
\newcommand \reals {\mathbb{R}}
\newcommand \sgn {\mathrm{sgn}}

\newcommand \bigO[1] {O{\left(#1\right)}}
\newcommand \bigOm[1] {\Omega{\left(#1\right)}}
\newcommand \bigTh[1] {\Theta{\left(#1\right)}}

\NewDocumentCommand \propbip {O{n}} {p_{#1}}
\NewDocumentCommand \approxpropbip {O{n}} {\tilde{p}_{#1}}
\NewDocumentCommand \maximizablebip {O{n}} {b_{#1}}
\NewDocumentCommand \maximizableone {O{n}} {o_{#1}}
\NewDocumentCommand \propone {O{n}} {q_{#1}}
\NewDocumentCommand \approxpropone {O{n}} {\tilde{q}_{#1}}
\NewDocumentCommand \totaltrees {O{n}} {t_{#1}}

\NewDocumentCommand \graph {} { G }
\NewDocumentCommand \bipgraph {} { B }
\NewDocumentCommand \tree {} { T }
\NewDocumentCommand \Root {} { r }
\NewDocumentCommand \rtree { O{\Root} O{\tree} }{ {#2}^{#1} }
\NewDocumentCommand \subtree { m O{\Root} O{\tree} } { {#3}^{#2}_{#1} }
\NewDocumentCommand \freetrees { O{n} } { \mathcal{F}_{#1} }

\NewDocumentCommand \pathgraphsymbol {} { P }
\NewDocumentCommand \pathgraph {} { \pathgraphsymbol }
\NewDocumentCommand \pathgraphclass { O{n} } { \mathcal{\pathgraph}_{#1} }

\NewDocumentCommand \cyclesymbol {} { C }
\NewDocumentCommand \cycle {} { \cyclesymbol }
\NewDocumentCommand \cycleclass { O{n} } { \mathcal{\cycle}_{#1} }

\NewDocumentCommand \regularsymbol {} { R }
\NewDocumentCommand \regular {} { \regularsymbol }
\NewDocumentCommand \kregularclass { O{k} O{n} } { \mathcal{\regular}^{(#1)}_{#2} }

\NewDocumentCommand \spidersymbol {} { A }
\NewDocumentCommand \spider {} { \spidersymbol }
\NewDocumentCommand \spiderclass { O{n} } { \mathcal{\spider}_{#1} }

\NewDocumentCommand \twolinearsymbol {} { L }
\NewDocumentCommand \twolinear {} { \twolinearsymbol }
\NewDocumentCommand \twolinearclass { O{n} } { \mathcal{\twolinear}_{#1} }

\NewDocumentCommand \klinearsymbol {} { L }
\NewDocumentCommand \klinear {} { \klinearsymbol }
\NewDocumentCommand \klinearclass { O{k} O{n} } { \mathcal{\klinear}^{(#1)}_{#2} }

\NewDocumentCommand \bistarsymbol {} { H }
\NewDocumentCommand \bistar {} { \bistarsymbol }
\NewDocumentCommand \bistarclass { O{n} } { \mathcal{\bistar}_{#1} }

\NewDocumentCommand \balancedbistarsymbol {} { O }
\NewDocumentCommand \balancedbistar {} { \balancedbistarsymbol }
\NewDocumentCommand \balancedbistarclass { O{n} } { \mathcal{\balancedbistar}_{#1} }

\NewDocumentCommand \startsymbol {} { S }
\NewDocumentCommand \start {} { \startsymbol }
\NewDocumentCommand \startclass { O{n} } { \mathcal{\start}_{#1} }

\NewDocumentCommand \quasistarsymbol {} { Q }
\NewDocumentCommand \quasistar {} { \quasistarsymbol }
\NewDocumentCommand \quasistarclass { O{n} } { \mathcal{\quasistar}_{#1} }

\NewDocumentCommand \kquasistar { O{k} } { \quasistar^{(#1)} }
\NewDocumentCommand \kquasistarclass { O{n} O{k} } { \quasistarclass[#1]^{(#2)} }

\NewDocumentCommand \degreesymbol {} { d }
\NewDocumentCommand \degreegraph { m m } { \degreesymbol_{#2}(#1) }
\NewDocumentCommand \degree { m } { \degreegraph{#1}{} }
\NewDocumentCommand \degreeG { m } { \degreegraph{#1}{\graph} }

\NewDocumentCommand \maxdegreesymbol {} { \Delta }
\NewDocumentCommand \maxdegreeset { m m } { \maxdegreesymbol_{#2}(#1) }
\NewDocumentCommand \maxdegreesetG { m } { \maxdegreeset{\graph}{#1} }
\NewDocumentCommand \maxdegreesetT { m } { \maxdegreeset{\tree}{#1} }
\NewDocumentCommand \maxdegree { m } { \maxdegreesymbol(#1) }
\NewDocumentCommand \maxdegreeG {} { \maxdegree{\graph} }
\NewDocumentCommand \maxdegreeT {} { \maxdegree{\tree} }

\NewDocumentCommand \neighborssymbol {} { \Gamma }
\NewDocumentCommand \neighbors { m } { \neighborssymbol(#1) }

\NewDocumentCommand \neighborsG { m } { \neighborssymbol_{\graph}(#1) }

\NewDocumentCommand \epotthistlessymbol {} { \psi }
\NewDocumentCommand \epotthistles { m } { \epotthistlessymbol(#1) }
\NewDocumentCommand \epotthistlesT {} { \epotthistles{\tree} }

\NewDocumentCommand \numepotthistlesT {} { |\epotthistlesT| }

\NewDocumentCommand \arr { O{\pi} } {#1}
\NewDocumentCommand \invarr { O{\arr} } {{#1}^{-1}}
\NewDocumentCommand \marr { O{\pi} } {\tilde{#1}}
\NewDocumentCommand \invmarr { O{\marr} } {{#1}^{-1}}
\NewDocumentCommand \oarr { O{\phi} } {#1}
\NewDocumentCommand \oinvarr { O{\oarr} } {{#1}^{-1}}
\NewDocumentCommand \maxarr { O{\arr} } {{#1}^*}
\NewDocumentCommand \biparr { O{\arr} } {{#1}_{B}}
\NewDocumentCommand \maxbiparr { O{\arr} } {\biparr[#1]^*}
\NewDocumentCommand \imaxbiparr { O{\arr} } {{#1}_{+}}
\NewDocumentCommand \nonbiparr { O{\arr} } {{#1}_{\overline{B}}}
\NewDocumentCommand \maxnonbiparr { O{\arr} } {\nonbiparr[#1]^*}
\NewDocumentCommand \diffarr { O{\tree} O{\Delta} } { {#2}(#1) }
\NewDocumentCommand \idiffarr { O{\tree} O{\Lambda} } { {#2}(#1) }

\NewDocumentCommand \lengthedgesymbol {} { \delta }
\NewDocumentCommand \dl { m O{\arr} } { \lengthedgesymbol_{#2}(#1) }
\NewDocumentCommand \incidentDsymbol {} { \gamma }
\NewDocumentCommand \iD { m O{\arr} } { \incidentDsymbol_{#2}(#1) }

\NewDocumentCommand \Dsymbol {} { D }
\NewDocumentCommand \D { O{\arr} O{\graph} } { \Dsymbol_{#1}(#2) }

\NewDocumentCommand \gDmin { m m } { m_{#1}[\Dsymbol(#2)] }
\NewDocumentCommand \DminProj { O{\Rtree} } { \gDmin{\projective}{#1} }
\NewDocumentCommand \DminPlan { O{\tree} } { \gDmin{\planar}{#1} }
\NewDocumentCommand \DminG { O{\graph} } { \gDmin{}{#1} }
\NewDocumentCommand \Dmin { O{\tree} } { \gDmin{}{#1} }
\NewDocumentCommand \gDMax { m m } { M_{#1}[\Dsymbol(#2)] }
\NewDocumentCommand \DMaxProj { O{\Rtree} } { \gDMax{\projective}{#1} }
\NewDocumentCommand \DMaxPlan { O{\tree} } { \gDMax{\planar}{#1} }
\NewDocumentCommand \DMax { m } { \gDMax{}{#1} }
\NewDocumentCommand \DMaxG {} { \DMax{\graph} }
\NewDocumentCommand \DMaxT {} { \DMax{\tree} }
\NewDocumentCommand \DMaxBipartite { m } {\gDMax{\mathrm{bip}}{#1}}
\NewDocumentCommand \DMaxBipartiteG {} {\DMaxBipartite{\graph}}
\NewDocumentCommand \DMaxBipartiteT {} {\DMaxBipartite{\tree}}
\NewDocumentCommand \DMaxNonBipartite { m } {\gDMax{\text{non-bip}}{#1}}
\NewDocumentCommand \DMaxNonBipartiteG {} {\DMaxNonBipartite{\graph}}

\NewDocumentCommand \DMaxOneThistle { m } {\gDMax{\mathrm{1}t}{#1}}
\NewDocumentCommand \DMaxOneThistleG {} {\DMaxOneThistle{\graph}}

\NewDocumentCommand \maxarrset { m } {\mathcal{\Dsymbol}^*(#1)}
\NewDocumentCommand \maxarrsetG {} {\maxarrset{\graph}}

\NewDocumentCommand \biparrset { m O{} } {\mathcal{\bipgraph}^{#2}(#1)}
\NewDocumentCommand \biparrsetG {} {\biparrset{\graph}[]}

\NewDocumentCommand \maxbiparrset { m } {\biparrset{#1}[*]}
\NewDocumentCommand \maxbiparrsetG {} {\maxbiparrset{\graph}}

\NewDocumentCommand \nonbiparrset { m O{} } {\overline{\mathcal{\bipgraph}}^{#2}(#1)}
\NewDocumentCommand \nonbiparrsetG {} {\nonbiparrset{\graph}[]}

\NewDocumentCommand \maxnonbiparrset { m } {\nonbiparrset{#1}[*]}
\NewDocumentCommand \maxnonbiparrsetG {} {\maxnonbiparrset{\graph}}

\NewDocumentCommand \nonbiparrsetonethistle { m O{} } {\overline{\mathcal{\bipgraph}}_{1t}^{#2}(#1)}
\NewDocumentCommand \nonbiparrsetonethistleG {} {\nonbiparrsetonethistle{\graph}[]}

\NewDocumentCommand \maxnonbiparrsetonethistle { m } {\nonbiparrsetonethistle{#1}[*]}
\NewDocumentCommand \maxnonbiparrsetonethistleG {} {\maxnonbiparrsetonethistle{\graph}}

\NewDocumentCommand \NB {} {{\em NB}}
\NewDocumentCommand \OTk {O{k}} {$#1${\em-thistle}}
\NewDocumentCommand \NBOTk {O{k}} {\NB$\times$\OTk[#1]}

\NewDocumentCommand \cutsignaturesymbol {} { C }
\NewDocumentCommand \cutsignature { m O{\arr} } { \cutsignaturesymbol_{#2}(#1) }
\NewDocumentCommand \cutsignatureG { O{\arr} } { \cutsignature{\graph}[#1] }
\NewDocumentCommand \cutsignatureT { O{\arr} } { \cutsignature{\tree}[#1] }
\NewDocumentCommand \cutsymbol {} { c }
\NewDocumentCommand \cut { m O{\arr} } { \cutsymbol_{#2}(#1) }

\NewDocumentCommand \levelsignaturesymbol {} { L }
\NewDocumentCommand \levelsignature { m O{\arr} } { \levelsignaturesymbol_{#2}(#1) }
\NewDocumentCommand \levelsignatureG { O{\arr} } { \levelsignature{\graph}[#1] }
\NewDocumentCommand \levelsignatureT { O{\arr} } { \levelsignature{\tree}[#1] }
\NewDocumentCommand \levelsymbol {} {l}
\NewDocumentCommand \level { m O{\arr} } { \levelsymbol_{#2}{\left(#1\right)} }

\NewDocumentCommand \leftdeg { m O{\arr} } { \mathfrak{l}_{#2}(#1) }
\NewDocumentCommand \rightdeg { m O{\arr} } { \mathfrak{r}_{#2}(#1) }

\NewDocumentCommand \neighborsarrsymbol { } { \Phi }
\NewDocumentCommand \neighborsarr { m m O{\arr} } { \neighborsarrsymbol_{#3}^{#2}(#1) }

\NewDocumentCommand \assignment {} { b }
\NewDocumentCommand \levelfunc { O{\assignment} } {\levelsymbol_{#1}}

\NewDocumentCommand \thistleset { O{\arr} } { \theta(#1) }
\NewDocumentCommand \thistlev {} { t }
\NewDocumentCommand \h {} { h }
\NewDocumentCommand \g {} { g } % it's just another hub...
\NewDocumentCommand \w {} { w_0 } % internal vertex
\NewDocumentCommand \degreehub { m } { \degreesymbol_{#1} }
\NewDocumentCommand \degreeh {} { \degreehub{\h} }
\NewDocumentCommand \degreeg {} { \degreehub{\g} }

\NewDocumentCommand \xyamo { m m } { {#1}_{#2} }
\NewDocumentCommand \nptwo {} { \xyamo{n}{+2} }
\NewDocumentCommand \npone {} { \xyamo{n}{+1} }
\NewDocumentCommand \nzero {} { \xyamo{n}{0} }
\NewDocumentCommand \nmone {} { \xyamo{n}{-1} }
\NewDocumentCommand \nmtwo {} { \xyamo{n}{-2} }
\NewDocumentCommand \hptwo {} { \xyamo{\h}{+2} }
\NewDocumentCommand \hpone {} { \xyamo{\h}{+1} }
\NewDocumentCommand \hmone {} { \xyamo{\h}{-1} }
\NewDocumentCommand \hmtwo {} { \xyamo{\h}{-2} }

\NewDocumentCommand \proofcase {} {\rho}
\NewDocumentCommand \agh {} {a_{\g\h}}
\NewDocumentCommand \awh {} {a_{\w\h}}

\NewDocumentCommand \pathoptimization   {} {\cref{lemma:properties_max_arrs:new:POL}}
\NewDocumentCommand \pathoptimizationO  {} {\cref{lemma:properties_max_arrs:new:POL}(\ref{lemma:properties_max_arrs:new:POL:0})}
\NewDocumentCommand \pathoptimizationOO {} {Lemma \ref{lemma:properties_max_arrs:new:POL}(\ref{lemma:properties_max_arrs:new:POL:00})}

\NewDocumentCommand \alternation   {} {\cref{lemma:properties_max_arrs:new:alternation}}
\NewDocumentCommand \alternationi  {} {\cref{lemma:properties_max_arrs:new:alternation}(\ref{lemma:properties_max_arrs:new:alternation:antenna})}
\NewDocumentCommand \alternationia  {} {\cref{lemma:properties_max_arrs:new:alternation}(\ref{lemma:properties_max_arrs:new:alternation:antenna:nonsingular})}
\NewDocumentCommand \alternationib {} {Lemma \ref{lemma:properties_max_arrs:new:alternation}(\ref{lemma:properties_max_arrs:new:alternation:antenna:singular})}
\NewDocumentCommand \alternationii  {} {\cref{lemma:properties_max_arrs:new:alternation}(\ref{lemma:properties_max_arrs:new:alternation:bridge})}
\NewDocumentCommand \alternationiia  {} {\cref{lemma:properties_max_arrs:new:alternation}(\ref{lemma:properties_max_arrs:new:alternation:bridge:nonthistle})}
\NewDocumentCommand \alternationiib {} {Lemma \ref{lemma:properties_max_arrs:new:alternation}(\ref{lemma:properties_max_arrs:new:alternation:bridge:thistle})}

\NewDocumentCommand \propagation    {} {\cref{lemma:properties_max_arrs:new:propagation}}
\NewDocumentCommand \propagationi   {} {\cref{lemma:properties_max_arrs:new:propagation}(\ref{lemma:properties_max_arrs:new:propagation:antenna})}

\NewDocumentCommand \propagationiic {} {\cref{lemma:properties_max_arrs:new:propagation}(\ref{lemma:properties_max_arrs:new:propagation:bridge:z2})}
\NewDocumentCommand \propagationiid {} {\cref{lemma:properties_max_arrs:new:propagation}(\ref{lemma:properties_max_arrs:new:propagation:bridge:z0})}
\NewDocumentCommand \propagationiie {} {\cref{lemma:properties_max_arrs:new:propagation}(\ref{lemma:properties_max_arrs:new:propagation:bridge:z_neighbors})}
\NewDocumentCommand \propagatefunc {m m} {\lambda(#1,#2)}

\NewDocumentCommand \Nurse {} {\cref{propos:properties_max_arrs:known:non_increasing_levsig,propos:properties_max_arrs:known:no_neighs_in_same_level,propos:properties_max_arrs:known:permutation_of_equal_level}}

\NewDocumentCommand \cir { m } {\tikz[baseline]{\node[anchor=base, draw, circle, inner sep=0, minimum width=1.0em]{{\footnotesize #1}};}}

\NewDocumentCommand \proofcasemone {} {\cir{$a$}}
\NewDocumentCommand \proofcasezero {} {\cir{$b$}}
\NewDocumentCommand \proofcasepone {} {\cir{$c$}}
\begin{document}

\title{Maximum Linear Arrangement: exact algorithms for specific classes of graphs and approximation algorithms for wide classes of graphs}
\author{Llu\'{i}s Alemany-Puig \and Juan Luis Esteban \and Ramon Ferrer-i-Cancho}
\institute{Quantitative, Mathematical and Computational Linguistics Research Group, Department of Computer Science, Universitat Politècnica de Catalunya, Jordi Girona 1-3, 08034, Barcelona, Catalonia, Spain \at \email{lluis.alemany.puig@upc.edu} \and \at \email{esteban@cs.upc.edu} \and \at \email {rferrericancho@cs.upc.edu}}

\titlerunning{Maximum Linear Arrangement: exact and approximate algorithms}
\authorrunning{Alemany-Puig, Esteban, Ferrer-i-Cancho}

\maketitle

\begin{abstract}
Linear arrangements of graphs are a well-known type of graph labeling and are found in many important computational problems. A linear arrangement is usually defined as a permutation of the $n$ vertices of a graph. An intuitive geometric setting is that of vertices lying on consecutive integer positions in the real line, starting at 1; edges are often drawn as semicircles above the real line. A well-known computational problem is the Minimum Linear Arrangement Problem ({\tt minLA}) where the goal is to find an arrangement that minimizes the sum of edge lengths. In this paper we study the Maximum Linear Arrangement problem ({\tt MaxLA}), the counterpart of {\tt minLA}. We devise a new characterization of maximum arrangements of general graphs, and prove that {\tt MaxLA} can be solved for $k$-regular graphs ($k\le2$) in time $\bigO{n}$, and for $k$-linear trees ($k\le2$) in time $\bigO{n}$. We present two constrained variants of {\tt MaxLA} we call {\tt bipartite MaxLA} and {\tt 1-thistle MaxLA}. We prove that the former can be solved in time $\bigO{n}$ for any connected bipartite graph; the latter can be solved by an algorithm that typically runs in time $\bigO{n^3\log n}$ on unlabeled trees. We show that {\tt bipartite MaxLA} is a $3/2$-approximation algorithm for {\tt MaxLA} for trees.
\keywords{Linear arrangements \and Combinatorial optimization \and Graph Theory \and Approximation algorithms}
\end{abstract}

%~ \tableofcontents

%%%%%%%%%%%%%%%%%%%%%%%%%%%%%%%%%%%%%%%%%%%%%%%%%%%%%%%%%%%%%%%%%%%%%%%%%%%%%%%%

%--------------------------------------------%
% <automatic inline of '1-introduction.tex'> %
%--------------------------------------------%
\section{Introduction}
\label{sec:introduction}

There exist many graph labeling problems in the literature \parencite{Gallian2018a}. One notorious example is the {\em Graceful Tree Conjecture}, first posed by \textcite{Rosa1967a} under the term `$\beta$-valuation', and coined later as `graceful' by \textcite{Golomb1972a}. This conjecture claims that all trees are {\em graceful}, that is, every tree admits a {\em graceful labeling}. To see that a tree $\tree=(V,E)$, where $V$ ($n=|V|$) and $E$ are the vertex and edge sets respectively, admits a graceful labeling, one only needs to find a bijection $\arr\,:\,V\rightarrow [n]$ such that every edge $uv\in E$ is uniquely identified by the value
\begin{equation}
\label{eq:introduction:d}
\dl{uv}=|\arr(u) - \arr(v)|.
\end{equation}
\cref{fig:introduction:examples_arrangements}(a) illustrates a caterpillar tree and a graceful labeling.

Another notorious example of a labeling problem is the Minimum Linear Arrangement Problem,\footnote{The Minimum Linear Arrangement Problem is also said to be a graph layout problem \parencite{Diaz1999a,Petit2003a,Petit2011a}.} henceforth abbreviated as {\tt minLA}, where the vertices of a graph $\graph=(V,E)$ are also labeled with a bijection $\arr\,:\,V\rightarrow [n]$, called {\em linear arrangement}, but the goal of the problem is to find the labeling that minimizes the {\em cost} associated to the labeling, $\D$, defined as
\begin{equation}
\label{eq:introduction:D}
\D = \sum_{uv\in E} \dl{uv}.
\end{equation}
More formally, the solution to {\tt minLA} is
\begin{equation*}
\DminG = \min_{\arr}\{\D\},
\end{equation*}
where the minimum is taken over all $n!$ linear arrangements. \cref{fig:introduction:examples_arrangements}(b) illustrates a minimum arrangement of the tree in \cref{fig:introduction:examples_arrangements}(a). Among applications for both {\tt minLA} and {\tt MaxLA} we find the study of cognitive principles in human languages via statistical normalization of $\D$ when $\graph$ is the tree of syntactic relationships among the words of a sentence and $\arr$ is the ordering of the words \parencite{Alemany2026a}. One specific application of {\tt MaxLA} is the placement of obnoxious facilities \parencite{Tamir1991a}.

Both labeling problems share the same geometric setting: place the vertices of the graph along consecutive integer positions of the real line, starting at $1$ and ending at $n$, where every vertex $u$ is placed at position $\arr(u)$, and draw the edges of the graph as semicircles above the line. The value $\lengthedgesymbol$ associated to an edge $uv\in E$ can be interpreted as the length of said edge in the arrangement; this is why the cost of an arrangement is often called the {\em sum of edge lengths} (e.g., \textcite{Ferrer2021a}). \cref{fig:introduction:examples_arrangements}(c) illustrates the graceful labeling in \cref{fig:introduction:examples_arrangements}(a) as a linear arrangement.

There exists sizable literature regarding {\tt minLA}. It is known to be {\bf NP}-Hard for general graphs \parencite{Garey1976a}, but there exist several polynomial-time algorithms for specific classes of graphs. For instance, {\tt minLA} is polynomial-time (in $n$) solvable in trees \parencite{Goldberg1976a,Shiloach1979a,Chung1984a}; the fastest algorithm is, to the best of our knowledge, due to \textcite{Chung1984a} with cost $\bigO{n^\lambda}$, where $\lambda$ is the smallest number such that $\lambda>\log{3}/\log{2}\approx1.58$. See \textcite{Diaz2002a,Petit2011a} for a comprehensive survey of polynomial-time algorithms for {\tt minLA} on several classes of graphs, including, but not limited to, rectangular grids, square grids, outerplanar graphs and chord graphs. There also exist promising linear programming approaches to solve {\tt minLA} \parencite{Andrade2017a}. There are studies concerning upper and lower bounds of {\tt minLA} for general graphs \parencite{Diaz1999a} and also several heuristic methods have been investigated \parencite{Petit2003a}. There exist several variants of {\tt minLA}. For instance, there exists the {\em planar} variant for free trees ({\tt planar minLA}) where the minimum is taken over all {\em planar arrangements}, where edges are not allowed to cross \parencite{Iordanskii1987a,Hochberg2003a}. There also exists the {\em projective} variant for rooted trees ({\tt projective minLA}) where the minimum is taken over all {\em projective arrangements}, planar arrangements where the root cannot be covered \parencite{Gildea2007a,Alemany2022a}. {\tt minLA} also features in language research acting as a universal linguistic principle \parencite{Ferrer2004a,Ferrer2022a}.

Comparatively less research exists on the Maximum Linear Arrangement problem, the maximum variant of {\tt minLA}, henceforth abbreviated as {\tt MaxLA} and formally defined as
\begin{equation}
\label{eq:introduction:MaxLA}
\DMaxG = \max_{\arr}\{\D\},
\end{equation}
where the maximum is taken over all $n!$ linear arrangements. \cref{fig:introduction:examples_arrangements}(d) illustrates a maximum arrangement of the tree in \cref{fig:introduction:examples_arrangements}(a). It is also known to be {\bf NP}-Hard for general graphs \parencite{Even1975a}, but in this case very few polynomial-time algorithms for classes of graphs are known \parencite{Ferrer2013a,Ferrer2021a,Ferrer2022a}. Bounds have been studied via direct construction of solutions \parencite{Ferrer2013a,Ferrer2021a,Ferrer2022a}, and also via the probabilistic method \parencite{Hassin2001a}, and there exist linear programming approaches to solve it for general graphs \parencite{Zundorf2022a}. Characterizations of maximum arrangements of graphs were only recently put forward by \textcite{Nurse2018a,Nurse2019a}. Moreover, they devised a $\bigO{n^d}$-time algorithm for rooted trees with maximum degree bounded by a constant $d$. {\tt MaxLA} for trees has also been studied in the {\em planar} and {\em projective} variants ({\tt projective/planar MaxLA}) \parencite{Alemany2024a}. For example, \cref{fig:introduction:examples_arrangements}(c) illustrates a maximum planar arrangement (in which edges are not allowed to cross) of the tree in \cref{fig:introduction:examples_arrangements}(a).

\begin{figure}
	\centering
	\includegraphics{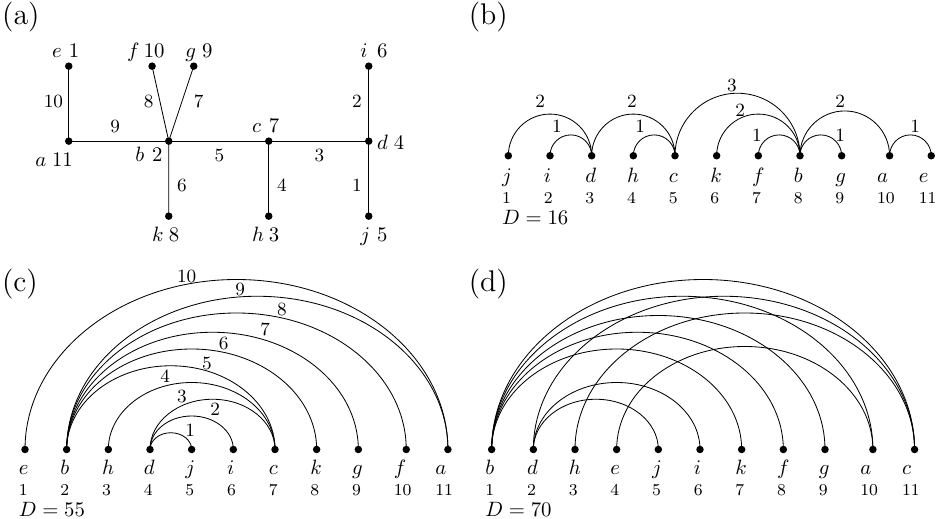}
	\caption{(a) A tree, whose vertices are labeled with letters $a,\dots,h$, and its graceful labeling, in numbers next to the vertices. Edges are labeled with the difference of the values in their endpoints. (b) A minimum arrangement of the tree in (a). (c) Linear arrangement of the graceful labeling in (a); it is also a maximum planar arrangement of the same tree \parencite{Alemany2024a}. (d) A maximum arrangement of the tree in (a).}
	\label{fig:introduction:examples_arrangements}
\end{figure}

Although {\tt MaxLA} has not been studied as extensively as {\tt minLA}, there also exist known solutions on particular classes of trees \parencite{Ferrer2021a,Ferrer2022a}. For path graphs $\pathgraph\in\pathgraphclass$ \parencite{Ferrer2021a},
\begin{equation}
\label{eq:introduction:DMax_path_graphs}
\DMax{\pathgraph} = \left\lfloor \frac{n^2}{2} \right\rfloor - 1.
\end{equation}
A bistar graph $\bistar\in\bistarclass$ ($n\ge2$) is a tree that consists of two star graphs joined at their {\em hubs}, denoted as $\h_1$ and $\h_2$, that is $\h_1\h_2\in E(\bistar)$. For any $\bistar\in\bistarclass$ \parencite{Ferrer2021a},
\begin{equation}
\label{eq:introduction:DMax_bistar_trees}
\DMax{\bistar} = \degreesymbol_1\degreesymbol_2 + {n-1 \choose 2},
\end{equation}
where $\degreesymbol_1$, $\degreesymbol_2$ denote the degrees of $\h_1$ and $\h_2$ respectively. Three types of bistar graphs are relevant. First, balanced bistar graphs $\balancedbistarclass$, trees where $|\degreesymbol_1 - \degreesymbol_2|\le 1$. For any $\balancedbistar\in\balancedbistarclass$ \parencite{Ferrer2021a},
\begin{equation*}
\DMax{\balancedbistar} = \frac{3(n - 1)^2 + 1 - (n\bmod 2)}{4}.
\end{equation*}
Second, star graphs $\start\in\startclass$ (bistar graphs in which $d_1=n-1$ and $d_2=1$) \parencite{Ferrer2021a}
\begin{equation*}
\DMax{\start} = {n \choose 2}
\end{equation*}
and, third, quasi-star graphs (bistar graphs in which $d_1=n-2$ and $d_2=2$). These trees have received special attention in Quantitative Linguistics research \parencite{Ferrer2013a,Ferrer2018a,Ferrer2022a}; their value of {\tt MaxLA} can be easily obtained by instantiating \cref{eq:introduction:DMax_bistar_trees}.

\textcite{Ferrer2022a} also presented a  solution to {\tt MaxLA} for $k$-quasistar graphs $\kquasistar\in\kquasistarclass$, star graphs where $k$ edges have been subdivided once (thus $0\le k\le n-1$ and $n = 2k + l + 1$, where $l\ge0$ is the number of edges of the star that are not subdivided). The quasistar graph above is an instance with $k=1$. It was shown that
\begin{equation*}
\DMax{\kquasistar} = \frac{(n - 1 - k)(n + 3k)}{2}.
\end{equation*}
All of these results were accompanied by a construction of the linear arrangement that yields the corresponding cost.

Finally, thanks to the characterization of the trees that maximize {\tt MaxLA} over all $n$-vertex free trees $\freetrees$, it was found that the maximum {\tt MaxLA} is achieved by a balanced bistar graph $\balancedbistar\in\balancedbistarclass$ \parencite{Ferrer2021a}
\begin{equation}
\label{eq:introduction:maximum_MaxLA}
\max_{\tree\in\freetrees}\{ \DMaxT \} = \DMax{\balancedbistar}.
\end{equation}
Moreover, it was also found that {\tt MaxLA} is minimized by star graphs \parencite{Ferrer2021a}
\begin{equation}
\label{eq:introduction:minimum_MaxLA}
\min_{\tree\in\freetrees}\{ \DMaxT \} = \DMax{\start}.
\end{equation}

In this paper we present a new characterization of maximum arrangements, and the solution to {\tt MaxLA} in $k$-regular graphs $\kregularclass$ and in $k$-linear trees $\klinearclass$ with $k \le2$. The class $k$-linear trees was introduced recently; a tree is a $k$-linear tree if all of its vertices of degree $\ge3$ lie on a single induced path \parencite{Johnson2020a}. For that class, we cover $k\le2$, which comprises path graphs ($k=0$), spider graphs ($k=1$) and $2$-linear trees ($k=2$). We present a constrained variant of {\tt MaxLA} for any given connected bipartite graph $\bipgraph=(V_1\cup V_2;E)$ which we call {\tt bipartite MaxLA} whereby the maximum sum of edge lengths is computed among all {\em bipartite arrangements}, namely arrangements where the first vertices are those from $V_1$ followed then by those from $V_2$ (or vice versa). We devise an algorithm that solves {\tt bipartite MaxLA} in time and space $\bigO{n}$ for any connected bipartite graph. We also prove that {\tt bipartite MaxLA} approximates {\tt MaxLA} by a factor of $3/2$ on trees. Moreover, we present another constrained variant we call {\tt 1-thistle MaxLA}, defined for any graph, whereby the maximum is taken over all arrangements that contain exactly one vertex (which we call {\em thistle}) whose neighbors are arranged to either side in the arrangement and thus makes the arrangement {\em non-bipartite}; we devise an algorithm that runs in time $\bigO{n^2\Delta2^\Delta}$, where $\Delta$ is the maximum degree of the tree, but that typically runs, for unlabeled trees, in time $\bigO{n^3 \cdot \log{n}}$ to solve it. The algorithms devised in this article rely heavily on a variant of bucket sort known as counting sort \parencite[Chapter 8, page 194]{Cormen2001a}. Counting sort is a non-comparison-based sorting algorithm that can sort $n$ elements within the range $[0,k]$ in time and space $\bigO{n + k}$. When $k = \bigO{n}$ as in our setting, counting sort has a time and space complexity of $\bigO{n}$, which is much lower than the lower bound $\bigOm{n\log{n}}$ of comparison-based algorithms \parencite[Chapter 8, page 193]{Cormen2001a}.

Our strategy in tackling {\tt MaxLA} has a parallel in research for {\tt minLA} on trees. The best known algorithm is due to \textcite{Chung1984a} with time complexity $\bigO{n^{1.58}}$. Although it runs in polynomial time in $n$, other researchers tried to approximate {\tt minLA} via a constrained variant, namely {\tt planar minLA} where edges are not allowed to cross; this variant has been solved in time $\bigO{n}$ by different authors \parencite{Iordanskii1987a,Hochberg2003a} and it has been claimed that it yields a $3/2$-approximation \parencite{Hochberg2003a}. Our approach is similar: there exists an algorithm for {\tt MaxLA} on trees due to \textcite{Nurse2018a,Nurse2019a} with time complexity $\bigO{n^{4\Delta}}$. Our goal is to approximate {\tt MaxLA} with an algorithm of lower time complexity by imposing constraints on the placement of the vertices ({\tt bipartite} and {\tt 1-thistle MaxLA} variants). In particular, {\tt bipartite MaxLA} can be solved in time $\bigO{n}$ and yields a $3/2$-approximation.

The implementations of {\tt bipartite MaxLA} and {\tt 1-thistle MaxLA} are all available in the Linear Arrangement Library \parencite{Alemany2021a}.\footnote{\url{https://github.com/LAL-project/linear-arrangement-library}} These algorithms have been tested with an exhaustive enumeration method of all linear arrangements up to $n=14$.

The remainder of the article is structured as follows. In \cref{sec:preliminaries}, we introduce core definitions and notation. In \cref{sec:properties_max_arrs:known}, we review past research on {\tt MaxLA}: we list and prove known characterizations of maximum arrangements. In \cref{sec:properties_max_arrs:new}, we present a new property of maximum linear arrangements of graphs, which is used throughout the remaining sections of the article. In \cref{sec:bip_and_nonbip_arrangements}, we present the division of arrangements into two types: bipartite and non-bipartite arrangements, which allows us to calculate {\tt MaxLA} as the maximum between a maximal bipartite arrangement (the solution to {\tt bipartite MaxLA}) and a maximal non-bipartite arrangement. In \cref{sec:bip_and_nonbip_arrangements:maximal_bip}, we show that the maximal bipartite arrangement of any connected bipartite graph can be calculated in time $\bigO{n}$, in \cref{sec:bip_and_nonbip_arrangements:maximal_nonbip} we devise an algorithm for {\tt 1-thistle MaxLA}, and in \cref{sec:bip_and_nonbip_arrangements:maximal_bip:relationship} we prove that {\tt bipartite MaxLA} approximates {\tt MaxLA} by a factor of $3/2$ on trees. In \cref{sec:max_for_classes}, we characterize the solution to {\tt MaxLA} for $k$-regular graphs and $k$-linear trees ($k\le2$). In this section we directly apply the aforementioned property of maximum linear arrangements. In \cref{sec:conclusions}, we draw some concluding remarks, and, in \cref{sec:future_work}, suggest future lines of research for {\tt MaxLA}.
%---------------------------------------------%
% </automatic inline of '1-introduction.tex'> %
%---------------------------------------------%
%---------------------------------------------%
% <automatic inline of '2-preliminaries.tex'> %
%---------------------------------------------%
\section{Preliminaries}
\label{sec:preliminaries}

Graphs are always assumed to be simple and are denoted as $\graph=(V,E)$ where $V$ and $E$ are, respectively, the set of vertices and edges; unless stated otherwise $n=|V|$. Edges $\{u,v\}\in E$ are denoted as $uv$. Bipartite graphs $\bipgraph=(V,E)=(V_1\cup V_2,E)$ make up a special class of graphs in which the vertex set $V$ can be partitioned into the sets $V_1$ and $V_2$ such that for every edge $uv\in E(\bipgraph)$ we have that $u\in V_1$ and $v\in V_2$ or vice versa. The vertices in $V_1$ (resp. $V_2$) are colored in blue (resp. red) in the figures of this paper.

The set of neighbors of a vertex $u$ in a graph $\graph$ is the set of vertices $v$ adjacent to $u$ in $\graph$, denoted as $\neighbors{u}=\neighborsG{u}$; the degree of a vertex $u$ in $\graph$ is the number of neighbors it has $\degree{u}=\degreeG{u}=|\neighborsG{u}|$. We omit the graph when it is clear from the context. The maximum degree over a set of vertices $S\subseteq V$ is denoted as $\maxdegreeset{\graph}{S} = \max_{v\in S} \{\degreeG{v}\}$; that of a graph is denoted as $\maxdegreeG=\maxdegreesetG{V}$. Recall that graphs can be represented using an $n\times n$ adjacency matrix $A=\{a_{uv}\}$ where $u,v\in V$ and $a_{uv}=1 \leftrightarrow uv\in E$.

We denote {\em free trees} (simply referred to as {\em trees}) as $\tree$. {\em Rooted trees} are denoted as $\rtree$ where $\Root$ is the root of the tree. In free trees, edges are undirected; in rooted trees edges are directed away from the root, and every vertex of the tree has a set of children which are said to be {\em siblings}. We denote the subtree of $\rtree$ rooted at $u\in V$, with $u\neq \Root$, as $\subtree{u}$. We denote the set of $n$-vertex free trees up to isomorphism as $\freetrees$; we denote its size as $\totaltrees=|\freetrees|$. See \textcite{OEIS_UlabFreeTrees} for a list of values of $\totaltrees$.

A {\em linear arrangement} $\arr$ of a graph $\graph=(V,E)$ is a bijection $\arr\::\:V\rightarrow [n]$. We use $\arr(u)$ to denote the position of vertex $u$ in $\arr$, and $\invarr(p)$ to denote the vertex at position $p$ in $\arr$. For simplicity, we also denote the arrangement of the vertices of an $n$-vertex graph $\graph$, where $V=\{u_1,\dots,u_n\}$ as an $n$-tuple $\arr = (u_1, \dots, u_n)$. For example, if $V=\{u,v,w,x\}$ and $\arr=(x,w,u,v)$ then $\arr(x)=1$, $\arr(w)=2$, and so on. We use the notation $(L_1 : L_2 : \cdots : L_{k-1}: L_k)$ from \textcite{Chung1984a} to denote an arrangement as the concatenation of the elements in the lists $L_1$, $L_2$, $\cdots$, $L_{k-1}$, $L_k$ in the order they appear in each list.

For any given graph $\graph$ and an arrangement $\arr$ of its vertices, we denote the length of any edge $uv\in E(\graph)$ as $\dl{uv}$ (\cref{eq:introduction:d}). We define the arrangement's cost as the sum of all edge lengths $\dl{uv}$, denoted as $\D$ (\cref{eq:introduction:D}). We denote the maximum $\D$ over all arrangements of $G$ as $\DMaxG$ (\cref{eq:introduction:MaxLA}). Furthermore, we denote the set of maximum arrangements of $G$ as
\begin{equation*}
\maxarrsetG = \{ \arr \;|\; \D = \DMaxG \}.
\end{equation*}

It will be useful to define {\em directional degrees} \parencite{deMier2007a} of a vertex in $\arr$. We can define the {\em left} (resp. {\em right}) {\em directional degree} of vertex $u$ as the number of neighbors of $u$ to the left (resp. right) of $u$ in $\arr$. More formally,
\begin{equation}
\label{eq:preliminaries:directional_degrees}
\leftdeg{u} = |\{ v\in\neighbors{u} \;|\; \arr(v)<\arr(u) \}|, \qquad
\rightdeg{u} = |\{ v\in\neighbors{u} \;|\; \arr(u)<\arr(v) \}|.
\end{equation}
Furthermore, we denote the set of neighbors of a vertex $u$ of a graph $\graph$ in a given interval of positions $I\subseteq[1,n]$ as
\begin{equation*}
\neighborsarr{u}{I} = \{ v\in\neighbors{u} \;|\; \arr(v)\in I \}.
\end{equation*}

In order to tackle {\tt MaxLA}, \textcite{Nurse2018a,Nurse2019a} identified two key features of arrangements of graphs: the {\em cut signature} and {\em level signature} of an arrangement.

\paragraph{Cut signature}
\label{sec:preliminaries:signature:cuts}

Cuts are a well-known concept in arrangements of graphs.\footnote{See, for instance, \textcite{Yannakakis1985a}.} For any given arrangement $\arr$ of $\graph$, the {\em cut width}, namely the {\em width of a cut}, between the $p$-th and the $(p+1)$-st positions of the arrangement is defined as the number of edges that cross the infinite line that goes through the midpoint between said positions and is perpendicular to the line of the arrangement (\cref{fig:preliminaries:illustration_cuts_levels}(b,c) shows an example of the cut widths of an arrangement). The width of the $p$-th cut, with $1\le p < n$, is defined as
\begin{equation}
\label{eq:preliminaries:cut_values}
\cut{p} = |\{ uv\in E(\graph) \;|\; \arr(u)\le p,\; p<\arr(v) \}|.
\end{equation}
We define $\cut{0}=\cut{n}=0$. For a given arrangement $\arr$ of $\graph$, \textcite{Nurse2018a,Nurse2019a} defined the {\em cut signature} of $\arr$ as the $(n-1)$-tuple
\begin{equation*}
\cutsignatureG = (\cut{1}, \dots, \cut{n-1})
\end{equation*}
where $\cut{p}$ denotes the {\em width of the $p$-th cut}.

Notice that $\D$ can be calculated as the sum of cut width values in the cut signature
\begin{equation}
\label{eq:preliminaries:arr_cost_as_sum_of_cuts}
\D = \sum_{p=1}^{n - 1} \cut{p}.
\end{equation}
A simple argument to prove this equality is that every unit of edge length contributes to a unit of width for all the cuts an edge goes through (\cref{fig:preliminaries:illustration_cuts_levels}(c)).

\paragraph{Level signature}
\label{sec:preliminaries:signature:levels}

For a given arrangement $\arr$ of $\graph$, \textcite{Nurse2018a,Nurse2019a} defined the {\em level value} of vertex $\invarr(i)$ as the difference of consecutive cut widths \parencite{Nurse2018a,Nurse2019a}
\begin{equation}
\label{eq:preliminaries:levels:difference_of_cuts}
\level{u_i} = \cut{i} - \cut{i - 1}.
\end{equation}
While this concept was defined by \textcite{Nurse2018a,Nurse2019a} in the context of arrangements of directed graphs, the definition also applies to undirected graphs. They also defined the {\em level signature} of an arrangement $\arr=(u_1,u_2,\dots,u_{n-1},u_n)$ as the $n$-tuple
\begin{equation*}
\levelsignatureG = ( \level{u_1}, \level{u_2}, \cdots, \level{u_{n-1}}, \level{u_n} ).
\end{equation*}

It is easy to see that an equivalent definition of level can be obtained with the difference of right and left directional degrees (\cref{eq:preliminaries:directional_degrees})
\begin{equation}
\label{eq:preliminaries:levels:difference_of_degrees}
\level{u} = \rightdeg{u} - \leftdeg{u},
\end{equation}
where $\rightdeg{u}$ and $\leftdeg{u}$ denote the right and left directional degrees (\cref{eq:preliminaries:directional_degrees}). \cref{fig:preliminaries:illustration_cuts_levels}(b) shows an example of a level signature. As noted by \textcite{Nurse2018a,Nurse2019a}, the $p$-th cut width value can be written as a function of the previous level values by applying \cref{eq:preliminaries:levels:difference_of_cuts} recursively
\begin{equation}
\label{eq:preliminaries:cut_values:sum_of_levels}
\cut{p}
	= \level{u_p} + \cut{p - 1}
	= \level{u_p} + \level{u_{p-1}} + \cut{p - 2}
	= \cdots
	= \sum_{q=1}^p \level{u_q}.
\end{equation}
Now, notice that the sum of all level values in any given arrangement $\arr$ of any given graph $\graph$ always equals $0$ since
\begin{equation}
\label{eq:preliminaries:sum_levels_equals_0}
\sum_{i=1}^n \level{u_i} = \cut{n} = 0.
\end{equation}
Finally, after straightforward calculations, one can see that the cost of an arrangement can also be calculated using only the level values
\begin{equation}
\label{eq:preliminaries:sum_edge_lengths_as_cutsum_levsum}
\D
	= \sum_{p=1}^{n - 1} \cut{p}
	= \sum_{p=1}^{n - 1} \sum_{q=1}^p \level{u_q}
	= \sum_{p=1}^{n - 1} (n - p)\level{u_p}.
\end{equation}

\begin{figure}
	\centering
	\begin{tabular}{l}
	\includegraphics{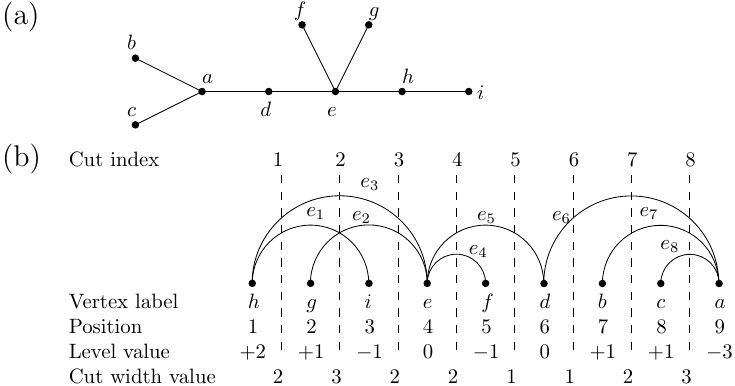} \\
	\\
	\begin{tabular}[t]{c}
	{\Large (c)}
	\end{tabular}
	\begin{tabular}[t]{rccccccccc}
	Edges     & \multicolumn{8}{c}{Cut index}                & Edge length \\
	          & 1   & 2   & 3   & 4   & 5   & 6   & 7   & 8  &             \\
	\midrule
	$e_{1}$   & x   & x   &     &     &     &     &     &    & 2           \\
	$e_{2}$   &     & x   & x   &     &     &     &     &    & 2           \\
	$e_{3}$   & x   & x   & x   &     &     &     &     &    & 3           \\
	$e_{4}$   &     &     &     & x   &     &     &     &    & 1           \\
	$e_{5}$   &     &     &     & x   & x   &     &     &    & 2           \\
	$e_{6}$   &     &     &     &     &     & x   & x   & x  & 3           \\
	$e_{7}$   &     &     &     &     &     &     & x   & x  & 2           \\
	$e_{8}$   &     &     &     &     &     &     &     & x  & 1           \\
	\midrule
	Cut width & 2   & 3   & 2   & 2   & 1   & 1   & 2   & 3  & Cost: 16    \\
	\end{tabular} \\
	\end{tabular}
	\caption{(a) A tree and (b) one of its arrangements. Under the arrangement are indicated, from top to bottom, the vertex labels of the tree, the positions of the vertices (from $1$ to $n$), the values of the levels per vertex, and the values of the cut widths. (c) Table indicating the cuts each vertex goes through, the length of each edge (rightmost column), the width of each cut (last row), and the total cost of the arrangement (bottom-right value).}
	\label{fig:preliminaries:illustration_cuts_levels}
\end{figure}
%----------------------------------------------%
% </automatic inline of '2-preliminaries.tex'> %
%----------------------------------------------%
%------------------------------------------%
% <automatic inline of '3-properties.tex'> %
%------------------------------------------%
\section{Properties of maximum linear arrangements}
\label{sec:properties_max_arrs}

In \cref{sec:properties_max_arrs:known}, we review known properties of maximum arrangements by \textcite{Nurse2018a,Nurse2019a} for any $n$-vertex graph $\graph=(V,E)$, where $V=\{u_1,\dots,u_n\}$. In all the properties presented below, $\maxarr$ is a maximum arrangement taken from the set of all (unconstrained) arrangements. In \cref{sec:properties_max_arrs:new}, we prove a new property of maximum arrangements.

\subsection{Overview of key contributions by Nurse and De Vos}
\label{sec:properties_max_arrs:known}

First, we introduce a new property (\cref{lemma:properties_max_arrs:known:vertex_swap}) that eases the proof of key properties by Nurse and De Vos and is applied later on in the article for other purposes. Second, we present the key properties by \textcite{Nurse2018a,Nurse2019a}. To make the article more streamlined while being instructive, the proofs of the properties in this section are given in \cref{sec:appendix:proof:properties:known}. Notice that the proofs of \cref{propos:properties_max_arrs:known:components_arranged_maximally,propos:properties_max_arrs:known:non_increasing_levsig,propos:properties_max_arrs:known:no_neighs_in_same_level,propos:properties_max_arrs:known:permutation_of_equal_level} are reelaborations with respect to the pioneering research of \textcite{Nurse2018a,Nurse2019a} that are instructive and connected by the new property (\cref{lemma:properties_max_arrs:known:vertex_swap}).

The new property relates two arrangements $\arr$ and $\arr'$, where $\arr'$ results from swapping two arbitrary vertices $v$ and $w$ in $\arr$. This relationship establishes that the level and cut values of $\arr'$ can be written as a function of the level and cut values of $\arr$ and the number of neighbor vertices that $v$ and $w$ have in the interval of positions $[\arr(v),\arr(w)]$. This relationship is independent of the graph structure.

\begin{lemma}
\label{lemma:properties_max_arrs:known:vertex_swap}
Let $\graph=(V,E)$ be any graph where $V=\{u_1,\dots,u_n\}$, and $\arr$ an arrangement of its vertices
\begin{equation*}
\arr=(u_1,\dots,u_{i-1},v,u_{i+1},\dots,u_{j-1},w,u_{j+1},\dots,u_n)
\end{equation*}
where $v=u_i$ and $w=u_j$ for two fixed $i,j\in[n]$, $i\neq j$. Now, let
\begin{equation*}
\arr'=(u_1,\dots,u_{i-1},w,u_{i+1},\dots,u_{j-1},v,u_{j+1},\dots,u_n)
\end{equation*}
be the arrangement where the positions of $v$ and $w$ have been swapped with respect to arrangement $\arr$ and the other vertices are not moved. Due to \cref{eq:preliminaries:levels:difference_of_degrees,eq:preliminaries:cut_values:sum_of_levels}, we find the following relationships between the level values in $\arr$ and $\arr'$:
\begin{alignat}{2}
\label{eq:properties_max_arrs:known:vertex_swap:levels:left}
\level{u_k}[\arr'] - \level{u_k} &= 0,
	& \quad & \forall k\in[1,i) \\ %]
\label{eq:properties_max_arrs:known:vertex_swap:levels:w}
\level{w}[\arr'] - \level{w} &= 2(|\neighborsarr{w}{(i,j)}| + a_{vw}), & & \\
\label{eq:properties_max_arrs:known:vertex_swap:levels:middle}
\level{u_k}[\arr'] - \level{u_k} &= 2(a_{vu_k} - a_{wu_k}),
	& \quad & \forall k\in(i,j) \\
\label{eq:properties_max_arrs:known:vertex_swap:levels:v}
\level{v}[\arr'] - \level{v} &= - 2(|\neighborsarr{v}{(i,j)}| + a_{vw}), & & \\
\label{eq:properties_max_arrs:known:vertex_swap:levels:right}
\level{u_k}[\arr'] - \level{u_k} &= 0,
	& \quad & \forall k\in(j,n]
\end{alignat}
and the following relationships between the cut widths in $\arr$ and $\arr'$:
\begin{alignat}{2}
\label{eq:properties_max_arrs:known:vertex_swap:cuts:left}
\cut{k}[\arr'] - \cut{k} &= 0,	& \quad & \forall k\in[1,i) \\ %]
\label{eq:properties_max_arrs:known:vertex_swap:cuts:middle}
\cut{k}[\arr'] - \cut{k} &=
	\level{w} - \level{v} + 2(
		a_{vw}
		+ |\neighborsarr{w}{(i,j)}|
		+ |\neighborsarr{v}{(i,k]}|
		- |\neighborsarr{w}{(i,k]}|
	),
	& \quad & \forall k\in[i,j) \\ %]
\label{eq:properties_max_arrs:known:vertex_swap:cuts:right}
\cut{k}[\arr'] - \cut{k} &= 0,
	& \quad & \forall k\in[j,n), %]
\end{alignat}
where, recall, $a_{vw}=1\leftrightarrow wv\in E$.
\end{lemma}

%-------------------------------------------------------------------------------
\begin{proposition}[\textcite{Nurse2018a,Nurse2019a}]
\label{propos:properties_max_arrs:known:components_arranged_maximally}
Let $\graph$ be a disconnected graph. The components of $G$, denoted as $H_i$, are arranged maximally in $\maxarr$. That is, $\maxarr$ restricted to $H_i$, denoted as $\maxarr(H_i)$, is a maximum arrangement of $H_i$.
\end{proposition}
%~ \begin{proof}
%~ By contrapositive. Let $\arr$ be an arrangement of $G$ where $\arr(H_i)$, the arrangement $\arr$ restricted to $H_i$, is not maximum. Then we can choose an arrangement of $H_i$ with higher cost than that of $\arr(H_i)$, and rearrange the vertices of $H_i$ within $\arr$ accordingly. The resulting arrangement has a higher cost than $\arr$ thus contradicting its maximality.
%~ \end{proof}
%-------------------------------------------------------------------------------

%~ The proofs of the remaining properties are all by contradiction.

%-------------------------------------------------------------------------------
\begin{proposition}[\textcite{Nurse2018a,Nurse2019a}]
\label{propos:properties_max_arrs:known:non_increasing_levsig}
The level signature of $\maxarr$ is monotonically non-increasing, that is, $\level{u_1}[\maxarr]\ge\dots\ge\level{u_n}[\maxarr]$.
\end{proposition}
%~ \begin{proof}
%~ Consider a maximum arrangement $\arr$ of $\graph$,
%~ \begin{equation*}
%~ \arr = (u_1,\dots,u_{i-1}, v, u_{i+1}, \dots, u_{j-1}, w, u_{j+1},\dots,u_n),
%~ \end{equation*}
%~ where, by way of contradiction, $\level{v} < \level{w}$. First, assume that $i+1=j$. Swap vertices $v$ and $w$ to obtain
%~ \begin{equation*}
%~ \arr' = (u_1,\dots,u_{i-1},w, v,u_{j+1},\dots,u_n).
%~ \end{equation*}
%~ Due to \cref{lemma:properties_max_arrs:known:vertex_swap} we have that
%~ \begin{alignat*}{2}
%~ \cut{k}[\arr'] - \cut{k} &= 0,                               & \quad & \forall           k\in[1,i)\cup[j,n] \\
%~ \cut{i}[\arr'] - \cut{i} &= \level{w} - \level{v} + 2a_{wv}. & \quad &
%~ \end{alignat*}
%~ Since we assumed that $\level{v}<\level{w}$, we have that $\cut{i}[\arr'] - \cut{i}>0$. Therefore $\D[\arr'] > \D$ (\cref{eq:preliminaries:arr_cost_as_sum_of_cuts}).

%~ In case $i+1<j$ then can find a pair of vertices $u_p$ and $u_{p+1}$ with $p\in[i,j]$ such that $\level{u_p} < \level{u_{p+1}}$. We can apply the same reasoning above to vertices $u_p$ and $u_{p+1}$ to reach the same contradiction.
%~ \end{proof}
%-------------------------------------------------------------------------------

%-------------------------------------------------------------------------------
\begin{proposition}[\textcite{Nurse2019a}]
\label{propos:properties_max_arrs:known:no_neighs_in_same_level}
There is no edge $uv\in E(\graph)$ such that $\level{u}[\maxarr]=\level{v}[\maxarr]$.
\end{proposition}
\begin{proposition}[\textcite{Nurse2019a}]
\label{propos:properties_max_arrs:known:permutation_of_equal_level}
Permuting vertices of equal level in a maximum arrangement $\maxarr$ does not change the cost of $\maxarr$.
\end{proposition}
%~ \begin{proof}
%~ Consider a maximum arrangement $\maxarr$ of $\graph$,
%~ \begin{equation*}
%~ \arr = (u_1,\dots,u_{p-1},u_p,\dots,u_{i-1},v,u_{i+1}, \dots, u_{j-1},w,u_{j+1},\dots,u_q,u_{q+1},\dots,u_n),
%~ \end{equation*}
%~ in which $\level{u_{p-1}}>\level{u_p}=\level{u_q}>\level{u_{q+1}}$ for some $1\le p<q\le n$. By \cref{propos:properties_max_arrs:known:non_increasing_levsig}, the vertices $u_p,\dots,u_q$ have the same level value, that is, $\level{u_p}=\cdots=\level{u_q}$. Due to \cref{propos:properties_max_arrs:known:no_neighs_in_same_level}, $a_{vw}=0$, and also $v$ and $w$ do not have neighbors in the interval of positions $(i,j)$ in $\arr$. Let $\arr'$ be the result of swapping $v$ and $w$ in $\arr$. Due to \cref{lemma:properties_max_arrs:known:vertex_swap}, $\cutsignatureG[\arr']=\cutsignatureG$ and therefore, due to \cref{eq:preliminaries:arr_cost_as_sum_of_cuts}, $\D[\arr'] = \D$.
%~ \end{proof}
%-------------------------------------------------------------------------------

Notice that \cref{propos:properties_max_arrs:known:non_increasing_levsig,propos:properties_max_arrs:known:no_neighs_in_same_level} establish necessary conditions for an arrangement of a graph to be maximum, and \cref{propos:properties_max_arrs:known:permutation_of_equal_level} states that all orderings of equal-level vertices are equivalent cost-wise in a maximum arrangement, thus greatly reducing the combinatorial exploration when looking for a maximum arrangement. Finally, notice that the proofs of \cref{propos:properties_max_arrs:known:components_arranged_maximally,propos:properties_max_arrs:known:non_increasing_levsig,propos:properties_max_arrs:known:no_neighs_in_same_level,propos:properties_max_arrs:known:permutation_of_equal_level} (\cref{sec:appendix:proof:properties:known}) are all independent of the graph's structure.

\subsection{A new property}
\label{sec:properties_max_arrs:new}

Here we identify a new property of maximum arrangements of graphs concerning the distribution of the vertices of branchless paths in maximum arrangements. In these paths of $k$ vertices, $k-2$ vertices have degree exactly $2$, and the endpoints have degree different from $2$. More formally, a branchless path $P$ (of $k-1$ edges) with vertices $(u_1,\dots,u_k)$ of a graph $\graph=(V,E)$ is a maximal sequence of vertices $u_1,\dots,u_k\in V$ such that
\begin{enumerate}
\item $\{u_i,u_{i+1}\}\in E$ for all $1\le i\le k-1$,
\item $\degree{u_i}=2$ for all $2\le i\le k-1$, and
\item $\degree{u_1}\neq2$ and $\degree{u_k}\neq2$.
\end{enumerate}
We refer to the vertices $\{u_2,\dots,u_{k-1}\}$ of $P$ as the {\em internal vertices} of the path. The {\em length} of these paths is measured in edges and equal to $k-1$. We say that a branchless path is a {\em bridge} if $\degree{u_1}\ge3$ and $\degree{u_k}\ge3$. We say that a branchless path is an {\em antenna} if $\degree{u_1}=1$ or $\degree{u_k}=1$.

Now follows a characterization of the level values of the vertices of branchless paths in a maximum arrangement of a given graph. In particular, we focus on the internal vertices of paths (those of degree $2$). For this, we define {\em thistle} vertices $v$ as those vertices whose neighbors are arranged to both sides of $v$ in the arrangement. More formally, a vertex $v$ is a {\em thistle vertex in an arrangement} $\arr$ if $|\level{v}|<\degree{v}$. Notice that a vertex $u$ of degree $2$ can only have three possible level values in an arbitrary arrangement: $\pm2$, and $0$ (the only case in which $u$ is a thistle). Leaves, by definition, can only have level values $\pm1$.

% ------------------------------------------------------------------------------
\begin{lemma}[Path Optimization Lemma]
\label{lemma:properties_max_arrs:new:POL}
Let $\maxarr$ be a maximum arrangement of a graph $\graph=(V,E)$. Let $P=(u_1,\dots,u_k)$ be any of the branchless paths in $\graph$. Then,
\begin{enumerate}[(i)]
\item \label{lemma:properties_max_arrs:new:POL:0} When $P$ is an antenna, none of its internal vertices is a thistle (that is, with level $0$) in $\maxarr$. Therefore, we have that $|\level{u}[\maxarr]|=2$ for all internal vertices $u$ of $P$.

\item \label{lemma:properties_max_arrs:new:POL:00} When $P$ is a bridge, then
	\begin{itemize}
	\item There is at most one internal vertex of $P$ that is a thistle (that is, with level $0$) in $\maxarr$,
	\item If one of the internal vertices is a thistle vertex in $\maxarr$, say $u$, then for any other internal vertex $v$ there is a maximum arrangement $\maxarr_2$ with $v$ as a thistle vertex and where $u$ is not a thistle vertex.
	\end{itemize}
	
	Therefore, for a bridge $P=(u_1,\dots,u_{j-1},u_j,u_{j+1},\dots,u_k)$ there are two possible configurations of level values in $\maxarr$. Either
	\begin{itemize}
	\item $|\level{u}[\maxarr]|=2$ for all $u\in\{u_2,\dots,u_{k-1}\}$ and thus there are no thistles, or
	
	\item given any $j\in[2,k-1]$, then $\level{u_j}[\maxarr]=0$ and $|\level{u}[\maxarr]|=2$ for all $u\in\{u_2,\dots,u_{j-1},u_{j+1},\dots,u_{k-1}\}$ and thus there is only one thistle ($u_j$).
	\end{itemize}
\end{enumerate}
\end{lemma}

For this proof we assume an orientation of the edges in $P$ in the direction from $u_1$ to $u_k$. We prove \pathoptimizationOO{} by performing swaps of certain vertices in a starting arrangement $\arr$ until we increase the cost or we reach a situation where \cref{propos:properties_max_arrs:known:no_neighs_in_same_level} is contradicted. All these swaps are illustrated in \cref{fig:properties_max_arrs:new:POL:swaps}. Their cost can be calculated applying \cref{lemma:properties_max_arrs:known:vertex_swap}.

\begin{figure}
	\centering
	\includegraphics{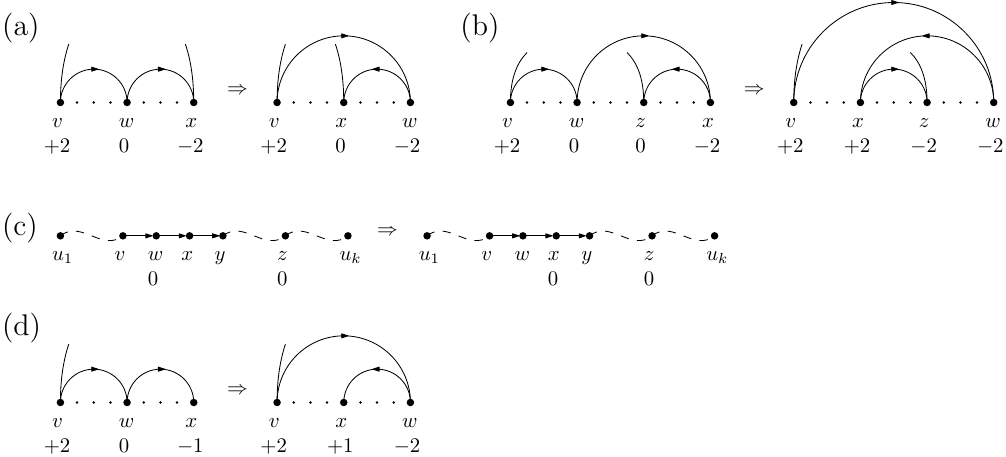}
	\caption{Proof of \pathoptimization. Different types of swaps of the vertices of a branchless path $P$. (a) Equal-cost swap of vertices such that the levels of the swapped vertices remain constant position-wise; $x$ has no neighbors to the right of $w$. The swap is of the same nature if the orientation is reversed. (b) Swap of two vertices that disturbs the level signature, thus increasing the cost of the arrangement; similar to case (a) but $x$ has a neighbor to the right of $w$. (c) Depiction of the level 0 being `pushed' along the path $P$ towards vertex $z$ when the swap in (a) is applied. (d) Increasing-cost swap applied to triples of vertices where the last vertex is a leaf.}
	\label{fig:properties_max_arrs:new:POL:swaps}
\end{figure}

% ------------------------------------------------------------------------------
\begin{proof}[Proof of \pathoptimizationOO]
Let $\arr$ be a maximum arrangement of $\graph$ (thus it satisfies \Nurse). Suppose that at least two internal and non-adjacent vertices of $P$ are placed in $\arr$ in such a way that their level value is $0$. In the orientation of $P$, let $v$ be the first vertex such that $\level{v}=0$, and let $z$ be the next vertex after $v$ in $P$ such that $\level{z}=0$. We do not assume any particular value for the level of the other vertices. Now, traverse the vertices of the path $P$ starting at $u_1$ towards $u_k$ while identifying triples of vertices $(s,w,x)$ (edges $(s,w)$, $(w,x)$ in the orientation of $P$) such that $\level{w}=0$ and $s$ is a vertex $u_i$ with $i\ge1$, and $x$ is a vertex $u_j$ with $j<k$. For every such triple, swap its last two vertices. This is illustrated in Figures \ref{fig:properties_max_arrs:new:POL:swaps}(a,b). Notice that after applying the swaps in Figures \ref{fig:properties_max_arrs:new:POL:swaps}(a,b), either (1) the level values change (\cref{fig:properties_max_arrs:new:POL:swaps}(b)) or (2) the level values remain constant position-wise (\cref{fig:properties_max_arrs:new:POL:swaps}(a)). If (1) the vertices are arranged in such a way that the levels change, then the cost of the arrangement increases (\cref{lemma:properties_max_arrs:known:vertex_swap}) and we are done. When (2) the level values remain constant position-wise, the cost of the arrangement does not change. In this case, notice that after every swap the vertex of $P$ with level 0 in $\arr$ closest to $u_1$ (in $P$) is `pushed' one vertex closer to $z$ (in $P$); this is illustrated in \cref{fig:properties_max_arrs:new:POL:swaps}(c). By repeated application of these swaps we obtain another arrangement, denoted as $\arr'$, where there will be two vertices $u_i$ and $u_{i+1}$ of $P$ ($u_iu_{i+1}\in E$) such that $\level{u_i}[\arr']=\level{u_{i+1}}[\arr']=0$ which contradicts \cref{propos:properties_max_arrs:known:no_neighs_in_same_level}. Since $\arr'$ cannot be maximum and has the same cost as $\arr$, we conclude that $\arr$ cannot be maximum either and thus at most one internal vertex of $P$ has level $0$ in $\arr$.

Now, notice that the above implies that, when there is only one internal vertex of $P$ with level $0$, we can apply the swaps in \cref{fig:properties_max_arrs:new:POL:swaps}(a) to change the thistle vertex of $P$ for another while keeping the cost of the arrangement constant.

Finally, it is easy to see that when there is one thistle in $\maxarr$, say $z$, we have that $\level{z}[\maxarr]=0$ since there is no other possibility for a degree-$2$ vertex to be a thistle, and the remaining internal vertices $u$ of $P$ are arranged such that $|\level{u}[\maxarr]|=2$.
\qed
\end{proof}
% ------------------------------------------------------------------------------

We now prove \pathoptimizationO{} by using the same type of argument as that in the proof for \pathoptimizationOO{} where one performs a series of swaps of vertices until the last swap in the series increases the cost of the arrangement.

% ------------------------------------------------------------------------------
\begin{proof}[Proof of \pathoptimizationO]
Again, let $\arr$ be a maximum arrangement of $\graph$ (thus it satisfies \Nurse). Due to \pathoptimizationOO{}, there is at most one thistle vertex with level 0. Assume that exactly one of the internal vertices of $P$ has level value $0$ in $\arr$. Using a reasoning very similar to the one above, traverse the vertices of $P$ starting at $u_1$ while identifying triples of vertices $(v,w,x)$ such that $\level{w}=0$; apply the equal-cost swaps on the triples $(v,w,x)$ where $v$ is a vertex $u_i$ with $i\ge1$, and $x$ is a vertex $u_j$ with $j\le k$ (\cref{fig:properties_max_arrs:new:POL:swaps}(a)). Eventually, the last such triple will be $(v,w,x)$ with $w=u_{k-1}$ and $x=u_k$, and thus $\degree{x}=1$ and $\level{x}\in\{+1,-1\}$. The swap to be applied in this case increases the cost of the arrangement (\cref{fig:properties_max_arrs:new:POL:swaps}(d)) thus contradicting the assumption that $\arr$ is maximum.

Finally, the internal vertices $u$ of $P$ are arranged such that $|\level{u}[\maxarr]|=2$ since the other possibility (level value equal to $0$) would make them thistles.
\qed
\end{proof}
% ------------------------------------------------------------------------------

As it will be seen in the next sections, \pathoptimization{} is instrumental in proving results on trees, since it tells us the distribution of the vertices of branchless paths of any graph in any maximum arrangement without making assumptions about the overall structure of the graph. We make this distribution more precise in \cref{lemma:properties_max_arrs:new:alternation} where we show that level values alternate sign along branchless paths. From now on, given a bridge path $P=(u_1,\dots,u_{i-1},z,u_{i+1},\dots,u_k)$ of any graph $\graph$, we designate a vertex $z$ to be the only vertex of $P$ allowed to be a thistle in a maximum arrangement; the case where $z$ is not a thistle ($|\level{z}[\maxarr]|=2$) is explained in \cref{lemma:properties_max_arrs:new:alternation}. Recall that by \pathoptimizationOO, we can consider $z$ to be the `only' vertex that can be a thistle in a maximum arrangement since it is interchangeable with other internal vertices of $P$ as thistles while keeping the cost constant.

\begin{lemma}
\label{lemma:properties_max_arrs:new:alternation}
Let $\maxarr$ be a maximum arrangement of a graph $\graph=(V,E)$. Let $P=(u_1,\dots,u_k)$ be any of the branchless paths in $\graph$. 
\begin{enumerate}[(i)]
	\item \label{lemma:properties_max_arrs:new:alternation:antenna} If $P$ is an antenna, where, w.l.o.g., $\degree{u_k}=1$, then the sign of the level value alternates from vertex to vertex, as in, w.l.o.g., $+2,-2,\dots,+2,-2,+1$. In particular,
	\begin{enumerate}[(a)]
		\item \label{lemma:properties_max_arrs:new:alternation:antenna:nonsingular} When $\degree{u_1}\ge3$ we have that $\level{u_j}[\maxarr]\level{u_{j+1}}[\maxarr]<0$, for $1<j< k$.
		
		\item \label{lemma:properties_max_arrs:new:alternation:antenna:singular} When $\degree{u_1}=1$ we have that $\level{u_j}[\maxarr]\level{u_{j+1}}[\maxarr]<0$, for $1\le j< k$.
	\end{enumerate}
	
	\item \label{lemma:properties_max_arrs:new:alternation:bridge} If $P=(u_1,\dots,u_{i-1},z,u_{i+1},\dots,u_k)$ is a bridge, then
	\begin{enumerate}[(a)]
		\item \label{lemma:properties_max_arrs:new:alternation:bridge:nonthistle} In case $|\level{z}[\maxarr]|=2$ then the level values alternate as in, w.l.o.g., $+2,-2,\dots,+2$. More precisely, $\level{u_j}[\maxarr]\level{u_{j+1}}<0$ for all $1 < j < k$.
		
		\item \label{lemma:properties_max_arrs:new:alternation:bridge:thistle} In case $\level{z}[\maxarr]=0$ then the pattern of level values follows, w.l.o.g., $+2,-2,\dots,+2,0,-2,+2,\dots,-2$. More formally,
		\begin{itemize}
			\item $\level{u_{i-1}}[\maxarr]\level{u_{i+1}}[\maxarr]<0$,
			\item $\level{u_j}[\maxarr]\level{u_{j+1}}[\maxarr]<0$ for all $1 < j < i-1$,
			\item $\level{u_j}[\maxarr]\level{u_{j+1}}[\maxarr]<0$ for all $i+1 < j < k-1$.
		\end{itemize}
	\end{enumerate}
\end{enumerate}
\end{lemma}
\begin{proof}
We prove only \alternationia; the case \alternationib{} follows immediately from similar arguments since $u_1$ is a leaf. Now, let $u_j\in\{u_2,\dots,u_k\}$ be a vertex such that $\level{u_j}[\maxarr]>0$. Then, it is easy to see that all its neighbors $v\in\neighbors{u_j}\setminus\{u_1\}=\{u_{j-1},u_{j+1}\}\setminus\{u_1\}$ must be such that $\level{v}[\maxarr]<0$: the situation in which $\level{v}[\maxarr]>0$ is simply not possible and $\level{v}[\maxarr]\neq0$ as per \pathoptimizationO. Moreover, the sign of level value also changes again for subsequent vertices $w\in\{u_{j-2},u_{j+2}\}\setminus\{u_1\}$ in $P$, that is, $\level{w}[\maxarr]>0$, and so on, due to similar arguments. Vertex $u_1$ is not included in case \alternationia{} since we assumed that $\degree{u_1}\ge3$. Therefore, the sign of $\level{u_j}$ for any $u_j$ is different from $\level{u_{j+1}}$ and $\level{u_{j-1}}$ hence the claim.

%~ By \pathoptimizationO, none of the internal vertices of an antenna $P$ can have level value $0$.

The case \alternationii{} can be proven with similar arguments. \alternationiia{} is proven with the same arguments above, not including $u_k$ since that is a vertex of degree $\ge3$ (by definition) and there are no thistles. \alternationiib{} is proven while taking into account that there may be a thistle in $P$, that is, the arguments can be applied to the vertices to each side of the thistle, and never including it.
\qed
\end{proof}

It is easy to see that one can determine the level values of many vertices given the level value of one of the internal vertices of a branchless path in a maximum arrangement $\arr$. We define this as a {\em propagation of level values} which follows almost trivially from \alternation. More precisely, consider any process of construction of a maximum arrangement $\arr$ of a graph $\graph$ and let $P=(u_1,\dots,x,\dots,u_i,\dots,y,\dots,u_k)$ be a branchless path of $\graph$ such that the level values of its vertices are unknown in $\arr$. Under the assumption that $\arr$ is maximum and that the level value of some internal vertex of $P$ can be ascertained say, $u_i$, then the sign of level value of $u_i$ can be propagated in $\arr$ through the vertices of $P$ starting at $u_i$ towards the endpoints as per \cref{lemma:properties_max_arrs:new:propagation}. Sometimes the propagation cannot include the endpoints, and has to stop at, say, $x$, or $y$.

\begin{lemma}
\label{lemma:properties_max_arrs:new:propagation}
Let $\graph$ be any graph and let $P$ be any branchless path of $\graph$. Define $\propagatefunc{i}{j} = (-1)^{|i - j|} \cdot \level{u_i}$, where $u_i$ is a vertex of $P$. Let $\maxarr$ be a maximum arrangement of $\graph$.
\begin{enumerate}[(i)]
\item \label{lemma:properties_max_arrs:new:propagation:antenna} Consider first the case where $P=(u_1,u_2,\dots,u_i,\dots,u_k)$ is an antenna with $\degree{u_k}=1$. Given the value of $\level{u_i}$ for $i\in[2,k]$ then
	\begin{itemize}
	\item $\level{u_j}[\maxarr] = \propagatefunc{i}{j}$ for all $j\in[2,k]$ when $\degree{u_1}\ge3$, or
	\item $\level{u_j}[\maxarr] = \propagatefunc{i}{j}$ for all $j\in[1,k]$ when $\degree{u_1}=1$.
	\end{itemize}

\item \label{lemma:properties_max_arrs:new:propagation:bridge} Consider now the case where $P$ is a bridge whose vertices are $P=(u_1,\dots,u_{i-1},z,u_{i+1},\dots,u_k)$.
	\begin{enumerate}[(a)]
	\item \label{lemma:properties_max_arrs:new:propagation:bridge:left_side} Given $\level{u_s}[\maxarr]$ for a fixed internal vertex $u_s$ to the left of $z$ (with $s\in[2,i-1]$) then $\level{u_j}[\maxarr] = \propagatefunc{s}{j}$ for all $j\in[2,i-1]$.
	
	\item \label{lemma:properties_max_arrs:new:propagation:bridge:right_side} Given $\level{u_s}[\maxarr]$ for a fixed internal vertex $u_s$ to the right of $z$ (with $s\in[i+1,k-1]$) then $\level{u_j}[\maxarr] = \propagatefunc{s}{j}$ for all $j\in[i+1,k-1]$.
	
	\item \label{lemma:properties_max_arrs:new:propagation:bridge:z2} Given $|\level{z}[\maxarr]| = 2$, then $\level{u_j}[\maxarr] = \propagatefunc{i}{j}$ for all $j\in[2,k-1]$, and $-\level{u_{i-1}}[\maxarr] = \level{z}[\maxarr] = -\level{u_{i+1}}[\maxarr]$.
	
	\item \label{lemma:properties_max_arrs:new:propagation:bridge:z0} Given $\level{z}[\maxarr] = 0$, then $\level{u_{i-1}}[\maxarr] = -\level{u_{i+1}}[\maxarr]$.
	
	\item \label{lemma:properties_max_arrs:new:propagation:bridge:z_neighbors} Given $\level{u_{i-1}}[\maxarr]$ and $\level{u_{i+1}}[\maxarr]$, then
		\begin{itemize}
		\item If $\level{u_{i-1}}[\maxarr] = \level{u_{i+1}}[\maxarr]$ then $-\level{u_{i-1}}[\maxarr] = \level{z}[\maxarr] = -\level{u_{i+1}}[\maxarr]$.
		\item If $\level{u_{i-1}}[\maxarr] = -\level{u_{i+1}}[\maxarr]$ then $\level{z}[\maxarr] = 0$.
		\end{itemize}
	\end{enumerate}
\end{enumerate}
\end{lemma}
\begin{proof}
\propagationi{} easily follows from \alternationi. The level values of the internal vertices of a bridge can also be expressed as a function of other internal vertices.
\begin{itemize}
\item \cref{lemma:properties_max_arrs:new:propagation}(\ref{lemma:properties_max_arrs:new:propagation:bridge:left_side},\ref{lemma:properties_max_arrs:new:propagation:bridge:right_side}) follow easily from \alternationii. Notice, however, that a single level value of an internal vertex at either side of $z$ is not enough to determine the level value of $z$.

\item \propagationiic{} follows directly from \alternationiia.

\item \propagationiid{} follows from \alternationiib.

\item To prove \propagationiie, first recall that $|\level{u_{i-1}}[\maxarr]| = |\level{u_{i+1}}[\maxarr]| = 2$ (\alternationii). If $\level{u_{i-1}}[\maxarr] = \level{u_{i+1}}[\maxarr]$ then the only possible values for $z$ is $\pm2$; following \alternationiia, the level value of $z$ has to have opposite sign to that of $\level{u_{i-1}}[\maxarr]$ and $\level{u_{i+1}}[\maxarr]$. If $\level{u_{i-1}}[\maxarr] = -\level{u_{i+1}}[\maxarr]$ then the only possible value for $z$ is $0$.
\end{itemize}
\qed
\end{proof}
%-------------------------------------------%
% </automatic inline of '3-properties.tex'> %
%-------------------------------------------%
%---------------------------------------------%
% <automatic inline of '4-0-bipartition.tex'> %
%---------------------------------------------%
\section{Bipartite and non-bipartite arrangements}
\label{sec:bip_and_nonbip_arrangements}

We start by formally defining the concepts of bipartite and non-bipartite arrangements in \cref{sec:bip_and_nonbip_arrangements:definitions}. In \cref{sec:bip_and_nonbip_arrangements:maximal_bip}, we prove that maximal bipartite arrangements of $n$-vertex trees can be constructed in time and space $\bigO{n}$. In \cref{sec:bip_and_nonbip_arrangements:maximal_bip:relationship} we study the approximation factor of {\tt bipartite MaxLA} to {\tt MaxLA}. Finally, in \cref{sec:bip_and_nonbip_arrangements:maximal_nonbip}, we devise an algorithm to construct a special case of maximal non-bipartite arrangements which will be useful in further sections.

%---------------------------------------------%
% <automatic inline of '4-1-definitions.tex'> %
%---------------------------------------------%
\subsection{Formal definitions}
\label{sec:bip_and_nonbip_arrangements:definitions}

{\em Bipartite} and {\em non-bipartite} arrangements split the whole set of $n!$ arrangements of any bipartite graph and help us tackle {\tt MaxLA} methodically. We say that a linear arrangement $\arr$ of a bipartite graph $\bipgraph=(V_1\cup V_2,E)$ is a {\em bipartite arrangement} if for every pair of vertices $u,v\in V_1$ there is no vertex $w\in V_2$ such that $\arr(u)<\arr(w)<\arr(v)$. Thus
\begin{equation*}
\max_{u\in V_1} \{ \arr(u) \} < \min_{v\in V_2} \{ \arr(v) \}
\qquad
\text{or}
\qquad
\max_{v\in V_2} \{ \arr(v) \} < \min_{u\in V_1} \{ \arr(u) \}.
\end{equation*}
Equivalently, let $V_1=\{u_1,\dots,u_{n_1}\}$ and $V_2=\{v_1,\dots,v_{n_2}\}$. The linear ordering of the vertices in a bipartite arrangement of a bipartite graph $\bipgraph$ is one of the following two
\begin{equation*}
	(u_{i_1}, \cdots, u_{i_p}, v_{j_1}, \cdots, v_{j_q}), \qquad
	(v_{j_1}, \cdots, v_{j_q}, u_{i_1}, \cdots, u_{i_p}),
\end{equation*}
where $\{i_1,\dots,i_p\}$ is a permutation of $\{1,\dots,n_1\}$, and $\{j_1,\dots,j_q\}$ is a permutation of $\{1,\dots,n_2\}$. A {\em non-bipartite arrangement} of a graph is simply an arrangement that is not bipartite.

We denote the set of bipartite arrangements of a graph $\graph$ as $\biparrsetG$. The maximum sum of edge lengths over bipartite arrangements, that is, the solution to {\tt bipartite MaxLA}, is denoted as $\DMaxBipartiteG$, and the set of maximal bipartite arrangements as $\maxbiparrsetG$. Formally,
\begin{equation*}
\DMaxBipartiteG = \max_{\arr \in \biparrsetG} \{ \D[\arr][\graph] \},\quad\text{and}\quad
%~ \maxbiparrsetG = \{ \arr\in\biparrsetG \;|\; \D[\arr][\graph] = \DMaxBipartiteG \}
\maxbiparrsetG = \argmax_{\arr \in \biparrsetG} \{ \D[\arr][\graph] \}.
\end{equation*}
Obviously, if $\graph$ is not bipartite, $\biparrsetG = \emptyset$ and thus $\DMaxBipartiteG=0$. % Notice that $\maxbiparrsetG$ is the red shaded region in \cref{fig:bip_and_nonbip_arrangements:definitions:maximum_maximal}.

Similarly, we denote the set of non-bipartite arrangements of a graph $\graph$ as $\nonbiparrsetG$. The maximum sum of edge lengths over non-bipartite arrangements, that is, the solution to {\tt non-bipartite MaxLA}, is denoted as $\DMaxNonBipartiteG$, and the set of maximal non-bipartite arrangements as $\maxnonbiparrsetG$. Formally,
\begin{equation*}
\DMaxNonBipartiteG = \max_{\arr \in \nonbiparrsetG} \{ \D[\arr][\graph] \},\quad\text{and}\quad
%~ \maxnonbiparrsetG = \{ \arr\in\nonbiparrsetG \;|\; \D[\arr][\graph] = \DMaxNonBipartiteG \}
\maxnonbiparrsetG = \argmax_{\arr \in \nonbiparrsetG} \{ \D[\arr][\graph] \}.
\end{equation*}
Obviously, each division contains arrangements that are maximum among arrangements of the same type. Thus we speak of {\em maximal bipartite arrangement} (resp. {\em maximal non-bipartite arrangement}) to refer to an arrangement that is maximum among bipartite (resp. non-bipartite) arrangements. In order to ease our discourse, we say that a graph $\graph$ is (resp. not) maximizable by a specific arrangement $\arr$ if its set of maximum arrangements $\maxarrsetG$ contains (resp. does not contain) $\arr$. More specifically, we say that {\em a graph $\graph$ is maximizable by a bipartite arrangement} (resp. {\em non-bipartite arrangement}) if $\maxbiparrsetG\subseteq\maxarrsetG$ (resp. $\maxnonbiparrsetG\subseteq\maxarrsetG$). If $\graph$ is maximizable by both bipartite and non-bipartite arrangements then $\maxarrsetG=\maxbiparrsetG\cup\maxnonbiparrsetG$. While an arrangement (of any type) can be maximal (maximum in its own part) it needs not be maximum among all $n!$ arrangements. \cref{fig:bip_and_nonbip_arrangements:definitions:maximum_maximal} illustrates all this terminology. Since the union of all bipartite and non-bipartite arrangements make up the whole space of $n!$ arrangements, it is easy to see that the solution to {\tt MaxLA} is simply the maximum between a maximal bipartite arrangement and a maximal non-bipartite arrangement,
\begin{equation*}
\DMaxG = \max \{ \DMaxBipartiteG, \DMaxNonBipartiteG \}.
\end{equation*}

\begin{figure}
	\centering
	\includegraphics{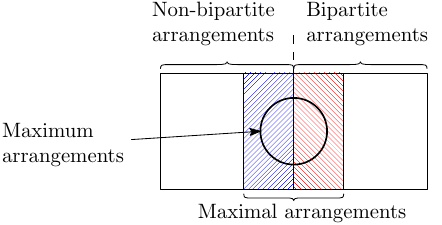}
	\caption{Diagram to illustrate the terminology concerning bipartite and non-bipartite arrangements. This figure shows all $n!$ arrangements of a graph divided into non-bipartite (left half) and bipartite (right half) arrangements. Left and right regions indicate, respectively, {\em maximal} non-bipartite and {\em maximal} bipartite arrangements. The filled circle in the intersection marks {\em maximum} arrangements. Notice that the actual intersection needs not cover both sides of the division.}
	\label{fig:bip_and_nonbip_arrangements:definitions:maximum_maximal}
\end{figure}

For the sake of brevity, we refer to the cost of a maximum arrangement of $\graph$, that is, the value of {\tt MaxLA} of $\graph$, as {\em the maximum cost} of $\graph$; we refer to the cost of a maximal bipartite arrangement (resp. maximal non-bipartite arrangement) of $\graph$ as {\em the maximal bipartite cost} (resp. {\em the maximal non-bipartite cost}) of $\graph$. The motivation for this division lies in the fact, proven later in this section, that any maximal bipartite arrangement of any $n$-vertex tree is constructible in time and space $\bigO{n}$.

It is easy to see that vertices in any bipartite arrangement $\arr$ always have as level value their degree (or minus degree). More formally, for every $u\in V$, $|\level{u}|=\degree{u}$. This basic fact follows from the definition of bipartite arrangements: all neighbors of $u$ are all arranged to either its left or to its right. On the other hand, many non-bipartite arrangements will contain at least one thistle vertex. Notice that there exist non-bipartite arrangements that do not contain thistle vertices.\footnote{Consider, for instance, the $4$-vertex path graph $\{1,2\},\{2,3\},\{3,4\}$ and the arrangement $(2,1,4,3)$. This arrangement is not bipartite but it does not contain any thistle vertex.} However, non-bipartite arrangements with no thistle vertices do not have a monotonic level signature, thus do not satisfy \cref{propos:properties_max_arrs:known:non_increasing_levsig}, and therefore cannot be maximum. Because of this, non-bipartite arrangements with no thistle vertices do not play any role in this article. Henceforth we use $\maxbiparr$ to denote a maximal bipartite arrangement, and $\maxnonbiparr$ to denote a maximal non-bipartite arrangement.

We now define yet another variant of {\tt MaxLA}, a generalization of {\tt bipartite MaxLA}, whereby the maximum is taken over all (non-bipartite) arrangements with exactly one thistle vertex, the set of which is denoted as $\nonbiparrsetonethistleG$. We call this variant {\tt 1-thistle MaxLA} and define it as
\begin{equation*}
\DMaxOneThistleG = \max_{\arr\in\nonbiparrsetonethistleG} \{ \D[\arr][\graph] \},\quad\text{and}\quad
%~ \maxnonbiparrsetonethistleG = \{ \arr\in\nonbiparrsetonethistleG \;|\; \D[\arr][\graph] = \DMaxOneThistleG \}
\maxnonbiparrsetonethistleG = \argmax_{\arr \in \nonbiparrsetonethistleG} \{ \D[\arr][\graph] \}
.
\end{equation*}

But not all vertices of a tree can be thistle vertices in a maximum arrangement. We say that a vertex of a tree has the potential to be a thistle (in a maximum arrangement) if no reason is known to discard it as one; for example, the leaves of a tree cannot be thistles by definition of leaf, and neither can the internal vertices of antennas (\pathoptimizationO). In the case of bridges, every vertex has the potential of being a thistle (\pathoptimizationOO). Even when a tree has many potential thistles, not all of them need to be thistles in the same maximum arrangement, nor all of them are thistles in some maximum arrangement. It is important to bear in mind that the property of being a thistle vertex depends on its position in an arrangement and it is unknown if that can be determined using the graph's structure alone in all cases.
%----------------------------------------------%
% </automatic inline of '4-1-definitions.tex'> %
%----------------------------------------------%
%-------------------------------------------%
% <automatic inline of '4-2-bipartite.tex'> %
%-------------------------------------------%
\subsection{{\tt Bipartite MaxLA}}
\label{sec:bip_and_nonbip_arrangements:maximal_bip}

In this section we explain how to construct maximal bipartite arrangements (\cref{sec:bip_and_nonbip_arrangements:maximal_bip:algorithm}) and how {\tt Bipartite MaxLA} for trees approximates the solution of {\tt MaxLA} for trees (\cref{sec:bip_and_nonbip_arrangements:maximal_bip:relationship}).

\subsubsection{Constructing maximal bipartite arrangements}
\label{sec:bip_and_nonbip_arrangements:maximal_bip:algorithm}

The algorithm to construct a maximal bipartite arrangement (\cref{algo:bip_and_nonbip_arrangements:maximal_bip:algorithm:Bipartite_MaxLA}): one simply has to compute a proper $2$-coloring of the tree, sort the set of vertices by degree for each color using counting sort \parencite{Cormen2001a}, and arrange the vertices to form a bipartite arrangement using said order.

\begin{algorithm}
	\caption{{\tt Bipartite MaxLA} in time and space $\bigO{n}$.}
	\label{algo:bip_and_nonbip_arrangements:maximal_bip:algorithm:Bipartite_MaxLA}
	\DontPrintSemicolon
	
	\KwIn{$\bipgraph$ a connected bipartite graph.}
	\KwOut{A maximal bipartite arrangement of $\bipgraph$.}
	
	\SetKwProg{Fn}{Function}{ is}{end}
	\Fn{\textsc{MaximalBipartiteArrangement}$(\bipgraph)$} {
		$V_1, V_2 \gets$ Proper coloring of the vertices of $\bipgraph$ with two colors. \tcp{Cost: $\bigO{n}$}
		\tcp{Cost: time and space $\bigO{n}$ with counting sort \parencite{Cormen2001a}}
		$V_1^* \gets$ a list, {\em non-increasing} ordering of $V_1$ by vertex degree \;
		$V_2^* \gets$ a list, {\em non-decreasing} ordering of $V_2$ by vertex degree \;
		
		$\maxbiparr \gets (V_1^* : V_2^*)$ \tcp{Construct $\maxbiparr$ by arranging $V_1^*$ and $V_2^*$ accordingly. Cost: $\bigO{n}$}
		\Return $\maxbiparr$
	}
\end{algorithm}

Now we prove that a maximal bipartite arrangement of any $n$-vertex connected bipartite graph, that is, the solution to {\tt bipartite MaxLA}, can be constructed in time and space $\bigO{n}$. This is key to many of the results proven later in this article and follows from \Nurse.

\begin{theorem}
\label{thm:bip_and_nonbip_arrangements:maximal_bip:algorithm}
\cref{algo:bip_and_nonbip_arrangements:maximal_bip:algorithm:Bipartite_MaxLA} solves {\tt bipartite MaxLA} for any connected bipartite graph in time and space $\bigO{n}$.
\end{theorem}
\begin{proof}
Let $\bipgraph=(V_1\cup V_2,E)$ be an $n$-vertex connected bipartite graph, where $V_1$ and $V_2$ are the two partitions of vertices for which every edge $uv\in E$ is such that $u\in V_1$ and $v\in V_2$. W.l.o.g., in any bipartite arrangement $\arr$ of $\bipgraph$ the vertices of $V_1$ are arranged first and the vertices of $V_2$ are arranged second in the arrangement. Therefore, for any $u\in V_1$ we must have $\level{u}=+\degree{u}$, and for $v\in V_2$ we must have $\level{v}=-\degree{v}$. By recalling the proof of \cref{propos:properties_max_arrs:known:non_increasing_levsig}, it is easy to see that vertices in $V_1$ must be placed in the arrangement non-increasingly by degree, and the vertices in $V_2$ must be placed non-decreasingly by degree, in order to construct a maximal bipartite arrangement. By construction of the arrangement, no two adjacent vertices will ever have the same level value. Given these two facts, and by recalling the proof of \cref{propos:properties_max_arrs:known:permutation_of_equal_level}, it is easy to see that the vertices with the same level value can also be ordered arbitrarily while keeping the cost constant. Finally, these conditions are also sufficient since a connected bipartite graph has a unique proper bicoloring.

For the sake of clarity, the steps to construct such arrangements are given in \cref{algo:bip_and_nonbip_arrangements:maximal_bip:algorithm:Bipartite_MaxLA}. The algorithm constructs a maximal bipartite arrangement by first ordering the vertices in $V_1$ and in $V_2$ following the description above; the result is denoted as $V_1^*$ and $V_2^*$ respectively. This sorting step can be done in time and space $\bigO{n}$ using counting sort \parencite{Cormen2001a} since we want to sort $\bigO{n}$ elements where the keys used to sort them (their degrees) range in the interval $[1,n)$. Then, it arranges first the vertices $V_1^*$ in the order they appear in $V_1^*$, then continue arranging the vertices of $V_2^*$ also in the order in which they appear in $V_2^*$. Both steps can be done in time $\bigO{n}$.
\qed
\end{proof}

Notice that when the partition $(V_1,V_2)$ of the bipartite graph is known, \cref{algo:bip_and_nonbip_arrangements:maximal_bip:algorithm:Bipartite_MaxLA} only needs knowledge on vertex degrees, that is, connectivity information (neighborhood of vertices) is no longer required.

\cref{algo:bip_and_nonbip_arrangements:maximal_bip:algorithm:Bipartite_MaxLA} is similar to the randomized algorithm devised by \textcite[Figure 4]{Hassin2001a}, that yields a $2$-approximation algorithm to solve {\tt MaxLA} in any general weighted graph $G=(V,E)$. Their algorithm makes a random partition of $V$ into two sets $V_1$ and $V_2$ and then arranges first the vertices in $V_1$ and then the vertices in $V_2$. The placement of the vertices is done such that the sum of the weight of the edges incident to a vertex is non-increasing among $V_1$ and non-decreasing among $V_2$. In our setting, the edges have all weight equal to $1$ and the sorting is by degree.

It is important to point out that the known solutions of {\tt MaxLA} for particular classes of trees cited in \cref{sec:introduction} are all obtained by constructing maximal bipartite arrangements.

\subsubsection{$3/2$-approximation}
\label{sec:bip_and_nonbip_arrangements:maximal_bip:relationship}

Here we study the approximation factor of {\tt bipartite MaxLA} with respect to {\tt MaxLA} over trees. \textcite[Figure 4]{Hassin2001a} already studied this relationship by means of their $2$-approximation algorithm. They showed that for any given graph $\graph$ the expected cost of the random arrangement produced by their approximation algorithm, denoted as $\phi(\graph)$, ``{\em is asymptotically $\frac{1}{2}opt$ with probability that approaches $1$ as a function of $n$}'' \parencite[Theorem 10, Figure 5]{Hassin2001a}, where $opt$ denotes the optimal value of {\tt MaxLA}. That is,
\begin{equation}
\label{eq:bip_and_nonbip_arrangements:maximal_bip:relationship:Hassin_2_factor}
\frac{\DMaxG}{\phi(\graph)} \le 2.
\end{equation}

This approximation factor can be improved easily in trees as an immediate consequence of previous results by \textcite{Ferrer2021a}.
\begin{corollary}
\label{cor:bip_and_nonbip_arrangements:relationship:our_factor}
Given any tree $\tree$,
\begin{equation}
\label{eq:bip_and_nonbip_arrangements:maximal_bip:relationship:our_factor}
\frac{\DMaxT}{\DMaxBipartiteT} \le \frac{3}{2}.
\end{equation}
\end{corollary}
\begin{proof}
It is easy to see that for any given $\tree\in\freetrees$,
\begin{equation*}
\frac{ \DMaxT                          }{ \DMaxBipartiteT } \le
\frac{ \max_{\tree'} \{\DMax{\tree'}\} }{ \min_{\tree'} \{\DMaxBipartite{\tree'}\} }.
\end{equation*}
The numerator is the maximum {\tt MaxLA} (\cref{eq:introduction:maximum_MaxLA}). The denominator is the minimum {\tt MaxLA}, which is achieved by a star graph whose maximum is a bipartite arrangement (\cref{eq:introduction:minimum_MaxLA}). Then
\begin{equation*}
\frac{ \max_{\tree'} \{\DMax{\tree'}\} }{ \min_{\tree'} \{\DMaxBipartite{\tree'}\} } =
\frac{ \DMax{\balancedbistar} }{ \DMax{\start} } \le
\frac{ 3n^2 - 6n + 4 }{ 2n^2 - 2n }.
\end{equation*}
Finally,
\begin{equation*}
\frac{ \DMaxT }{ \DMaxBipartiteT } \le
\lim_{n\rightarrow\infty} \frac{ 3n^2 - 6n + 4 }{ 2n^2 - 2n } = \frac{3}{2}.
\end{equation*}
\qed
\end{proof}
%--------------------------------------------%
% </automatic inline of '4-2-bipartite.tex'> %
%--------------------------------------------%
%-----------------------------------------------%
% <automatic inline of '4-3-non-bipartite.tex'> %
%-----------------------------------------------%
\subsection{Maximal non-bipartite arrangements of trees}
\label{sec:bip_and_nonbip_arrangements:maximal_nonbip}

It is easy to see that maximal non-bipartite arrangements with $k$ thistle vertices form (possibly empty) sequences of consecutive non-thistle vertices that satisfy \Nurse{} individually. By `satisfy' we mean that these sequences have the same properties as those stated in \Nurse. Notice, however, that a maximal arrangement of $k$ thistles needs not be maximum, and thus whether or not \Nurse{} hold in these arrangements, globally or in individual sequences, is not straightforward. We formalize this claim as follows.
\begin{proposition}
\label{propos:bip_and_nonbip_arrangements:maximal_nonbip:near_Nurse}
Let $\graph$ be any graph and let $\maxnonbiparr$ be any of its maximal (not necessarily maximum) non-bipartite arrangements in which vertices $\thistlev_1,\dots,\thistlev_k\in V$ are the only thistle vertices, such that, w.l.o.g., $\maxnonbiparr(\thistlev_1)<\cdots<\maxnonbiparr(\thistlev_k)$, that is, let
\begin{equation*}
\maxnonbiparr = \left(
	u_1^{1}, \dots, u_{i_1}^{1},
	\thistlev_1,
	u_1^{2}, \dots, u_{i_2}^{2},
	\thistlev_2,
	\dots,
	u_1^{k}, \dots, u_{i_k}^{k},
	\thistlev_k,
	u_1^{k+1}, \dots, u_{i_{k+1}}^{k+1}
\right),
\end{equation*}
where $i_1,\dots,i_{k+1} \ge 0$ are the number of non-thistle vertices in each sequence. Then for each $j\in[1,k+1]$, 
\begin{itemize}
	\item The sequence of level values $\level{u_1^{j}}[\maxnonbiparr]$, $\dots$, $\level{u_{i_j}^{j}}[\maxnonbiparr]$ satisfies \cref{propos:properties_max_arrs:known:non_increasing_levsig}.
	\item The edges with both endpoints within the $j$-th sequence satisfy \cref{propos:properties_max_arrs:known:no_neighs_in_same_level}.
	\item The two previous claims imply that the sequence $\level{u_1^{j}}[\maxnonbiparr]$, $\dots$, $\level{u_{i_j}^{j}}[\maxnonbiparr]$ also satisfies \cref{propos:properties_max_arrs:known:permutation_of_equal_level}.
\end{itemize}
\end{proposition}

The proof of \cref{propos:bip_and_nonbip_arrangements:maximal_nonbip:near_Nurse} is similar to that of \Nurse{}, but applied only to each sequence individually. Notice that applying the swaps in the proofs of \cref{propos:properties_max_arrs:known:non_increasing_levsig,propos:properties_max_arrs:known:no_neighs_in_same_level} to sequences that include a thistle vertex, may move the thistle vertex in such a way that it is no longer a thistle vertex in the new arrangement. For example, notice that \cref{propos:properties_max_arrs:known:no_neighs_in_same_level} is not satisfied in a maximal non-bipartite arrangement of a star graph of $4$ vertices when taking the entire arrangement into consideration: the thistle vertex (the center) and one of its leaves have both level value equal to $1$ in absolute value.

\subsubsection{{\tt 1-thistle MaxLA}}

In this section we devise an algorithm to construct a maximal non-bipartite arrangement of a tree $\tree$ such that the arrangement only has one thistle vertex which is given as part of the input (\cref{algo:bip_and_nonbip_arrangements:NonBipartite_MaxLA_one_thistle}). This algorithm is particularly useful to construct maximum arrangements of trees that are known to have at most one thistle vertex (\cref{sec:max_for_classes:k_linear:two_linear}). We use \cref{algo:bip_and_nonbip_arrangements:NonBipartite_MaxLA_one_thistle} later in this section to tackle {\tt 1-thistle MaxLA} (\cref{algo:bip_and_nonbip_arrangements:one_thistle_MaxLA}).

The central idea behind this algorithm is, given a thistle vertex $\thistlev$, its neighbors being $\neighbors{\thistlev}=\{v_1,\dots,v_d\}$, and given an assignment $\assignment \;:\; \neighbors{\thistlev} \rightarrow \{-1,+1\}$ of the side of the thistle to which each $v_i$ goes ($-1$ means {\tt left}, $+1$ means {\tt right}), to construct a maximal arrangement based on $\assignment$. This is an arrangement in which (1) the level value of $\thistlev$ is determined by the placement of its neighbors according to $\assignment$, and (2) the connected components $\tree_1, \dots, \tree_d$ (with $d=\degree{\thistlev}$) are arranged in maximal bipartite arrangements $\biparr^{(i)}$ (\cref{sec:bip_and_nonbip_arrangements:maximal_bip:algorithm}) such that the level value $\level{v_i}[\biparr^{(i)}]$ is in accordance to the side of $\thistlev$ on which $v_i$ is to be placed (at according to $\assignment$). \cref{algo:bip_and_nonbip_arrangements:NonBipartite_MaxLA_one_thistle} finds a maximal such arrangement for each possible $\assignment$ for a given $\thistlev$, and keeps track of the arrangement with largest cost. The crucial operation is to properly merge the different $\biparr^{(i)}$ and apply the necessary relocation of vertices so that it is the best possible for a given $\assignment$. \cref{fig:bip_and_nonbip_arrangements:NonBipartite_MaxLA_one_thistle:step_by_step} shows an example of the step-by-step construction of a maximal arrangement with $1$ thistle vertex.

\begin{algorithm}
	\caption{{\tt MaxLA} with exactly one given thistle vertex $\thistlev$, in time $\bigO{n\degree{\thistlev}2^{\degree{\thistlev}}}$ and space $\bigO{n}$.}
	\label{algo:bip_and_nonbip_arrangements:NonBipartite_MaxLA_one_thistle}
	\DontPrintSemicolon
	
	\KwIn{$\tree=(V_1\cup V_2, E)$ a free tree, $\thistlev$ a thistle vertex.}
	\KwOut{A maximal non-bipartite arrangement of $\tree$ with vertex $\thistlev$ as thistle vertex.}
	
	\SetKwProg{Fn}{Function}{ is}{end}
	\Fn{\textsc{1-Given-Thistle-MaxLA}$(\tree, \thistlev)$} {
		$\arr\gets \emptyset$ \tcp{empty arrangement of $n$ positions}
		
		\tcp{Generate all assignments of left- and right-side of $\neighbors{\thistlev}$. $\bigO{2^d}$ iterations}
		\For {every assignment $\assignment$} {
			
			$\levelfunc \gets$ computation of the level value to every vertex in the tree \;
			\lIf {$\levelfunc(\thistlev) < 0$} { skip to the next assignment }
			
			$\arr_1 \gets$ Step \ref{step:bip_and_nonbip_arrangements:maximal_nonbip_one_thistle:1}: construct an initial arrangement using $\levelfunc$: \;
			\quad \ $|$ Make it of the form (vertices with positive level : $\thistlev$ : vertices with negative level) \;
			\quad \ $|$ Sort the vertices in $\arr_1$ non-decreasingly by (positive) level value \;
			\quad \ $|$ Sort the vertices in $\arr_1$ non-decreasingly by (negative) level value \;
			
			$\arr_2 \gets$ Step \ref{step:bip_and_nonbip_arrangements:maximal_nonbip_one_thistle:2}: normalize $\arr_1$ by sorting the sequences of vertices of equal level value. \;
			
			$\arr_3 \gets$ Step \ref{step:bip_and_nonbip_arrangements:maximal_nonbip_one_thistle:3}: on $\arr_2$, move vertices $u\in V'$ that are found to the left of $\thistlev$ immediately to the right of $\thistlev$ as long as that increases the cost of the arrangement. \;
			
			$\arr \gets \max(\arr, \arr_3)$ \tcp{Keep the arrangement with maximum cost}
		}
		\Return $\arr$
	}
\end{algorithm}

\begin{figure}
	\centering
	\includegraphics[scale=0.95]{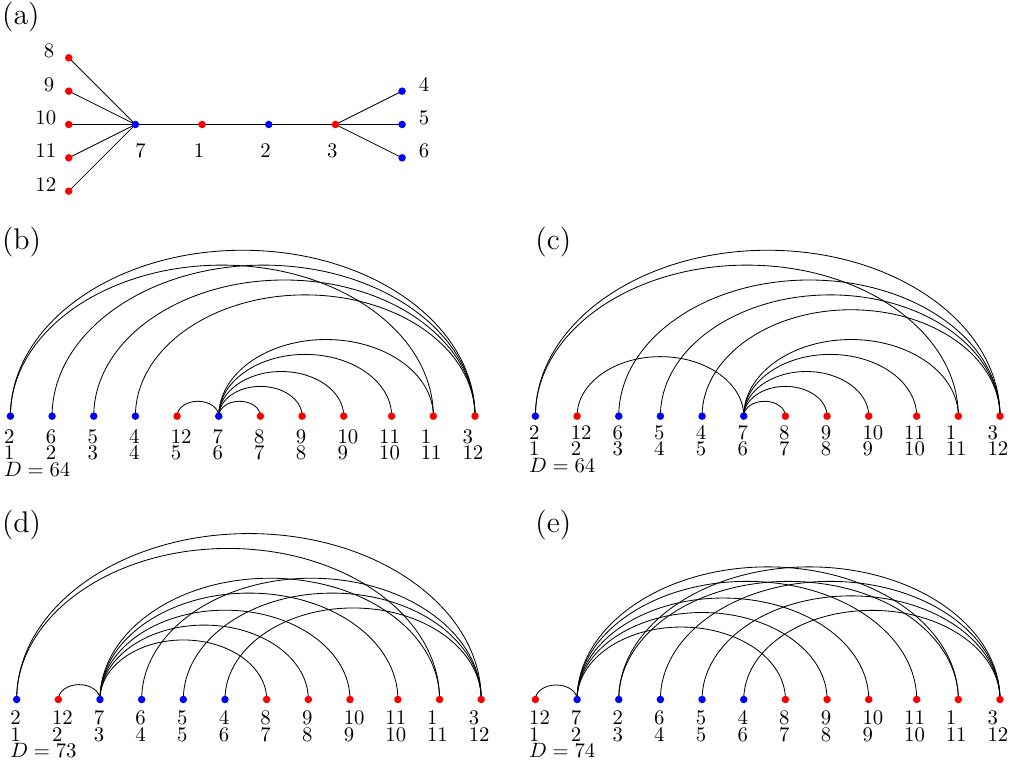}
	\caption{Step-by-step illustration of \cref{algo:bip_and_nonbip_arrangements:NonBipartite_MaxLA_one_thistle} on the tree in (a), for which we choose $\thistlev=7$ and assignment $\assignment$ such that $\assignment(12)=-1$ and $\assignment(8)=\assignment(9)=\assignment(10)=\assignment(11)=1$. (b) The arrangement that results after placing the vertices in maximal bipartite arrangements and placing $\thistlev$ between the vertices with positive level value and those with negative value (Step \ref{step:bip_and_nonbip_arrangements:maximal_nonbip_one_thistle:1}). (c) Normalization of the arrangement (Step \ref{step:bip_and_nonbip_arrangements:maximal_nonbip_one_thistle:2}). (d) Relocating the first block of non-neighbors of the thistle (Step \ref{step:bip_and_nonbip_arrangements:maximal_nonbip_one_thistle:3}). (e) Relocating the second block of non-neighbors of the thistle (Step \ref{step:bip_and_nonbip_arrangements:maximal_nonbip_one_thistle:3}).}
	\label{fig:bip_and_nonbip_arrangements:NonBipartite_MaxLA_one_thistle:step_by_step}
\end{figure}

Prior to detailing and proving the correctness of \cref{algo:bip_and_nonbip_arrangements:NonBipartite_MaxLA_one_thistle}, we first prove several intermediate results. We now prove that, in maximal non-bipartite arrangements of only $1$ thistle vertex $\thistlev$, the vertices that are not the thistle and are not neighbors of any of the thistle must be arranged so that they form a non-increasing level sequence. As before, this is obvious in arrangements with 1 thistle vertex that are {\em maximum}, but not so obvious in {\em maximal} arrangements of $1$ thistle vertex, which need not be maximum.

\begin{lemma}
\label{lemma:bip_and_nonbip_arrangements:maximal_nonbip:one_thistle_and_neighbors}
Let $\graph$ be any graph and let $\maxnonbiparr$ be any of its maximal non-bipartite arrangements as in \cref{propos:bip_and_nonbip_arrangements:maximal_nonbip:near_Nurse} with exactly one thistle vertex. Let $\thistlev\in V$ be the only thistle vertex in $\maxnonbiparr$. The vertices in $V'$
\begin{equation*}
V'	= V \setminus
	\left( \{\thistlev\} \cup \neighbors{\thistlev} \right)
	= \{r_1, \dots, r_\omega\}
\end{equation*}
for which we assume, w.l.o.g., that $\maxnonbiparr(r_1) < \cdots < \maxnonbiparr(r_\omega)$, are arranged in $\maxnonbiparr$ so that
\begin{equation*}
\level{r_1}[\maxnonbiparr] \ge \cdots \ge \level{r_\omega}[\maxnonbiparr].
\end{equation*}
Notice that the vertices in $V'$ need not be consecutive in the arrangement.
\end{lemma}
\begin{proof}
Let $\arr$ be a maximal non-bipartite arrangement of the form described above. First, recall that, by construction of $\arr$, we must have that $|\level{u}|=\degree{u}$ for any $u\in V'$. Let $u_1,\dots,u_x$ be all the vertices to the left of $\thistlev$, and let $v_1,\dots,v_y$ be all the vertices to the right of $\thistlev$ (\cref{fig:bip_and_nonbip_arrangements:one_thistle_and_neighbors:proof}). By \cref{propos:bip_and_nonbip_arrangements:maximal_nonbip:near_Nurse}
\begin{equation*}
\level{u_1} \ge \cdots \ge \level{u_x} \qquad \text{and} \qquad
\level{v_1} \ge \cdots \ge \level{v_y}.
\end{equation*}
and thus we have that the vertices in $V'$ to each side of $\thistlev$ satisfy the claim.

Now, let $u\in V'$ be the closest vertex to $\thistlev$ to its left in the arrangement; let $v\in V'$ be the closest vertex to $\thistlev$ to its right in the arrangement. In order to prove the claim it remains to show that $\level{u}\ge\level{v}$. Let $\{w_1,\dots,w_\alpha\}\subseteq \neighbors{\thistlev}$ be the set of vertices between $u$ and $\thistlev$ in $\arr$; similarly, let $\{z_1,\dots,z_\beta\}\subseteq \neighbors{\thistlev}$ be the set of vertices between $\thistlev$ and $v$ in $\arr$ (\cref{fig:bip_and_nonbip_arrangements:one_thistle_and_neighbors:proof}). $\alpha$ and $\beta$ denote the number of vertices between $u$ and $\thistlev$, and between $\thistlev$ and $v$ respectively.

\begin{figure}
	\centering
	\includegraphics{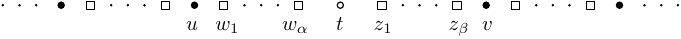}
	\caption{Proof of \cref{lemma:bip_and_nonbip_arrangements:maximal_nonbip:one_thistle_and_neighbors}. Squares denote vertices from $\neighbors{\thistlev}$ and filled dots denote vertices in $V'$.}
	\label{fig:bip_and_nonbip_arrangements:one_thistle_and_neighbors:proof}
\end{figure}

By way of contradiction, suppose that $\arr$ is a maximal non-bipartite arrangement with only one thistle vertex in which $\level{v}>\level{u}$. We have to show that by either moving $u$ to the right of $v$ or moving $v$ to the left of $u$ the cost of the arrangement increases and therefore $\arr$ is not maximal. This displacement is implemented as a sequence of swaps of consecutive vertices in the arrangement. For this, we apply \cref{lemma:properties_max_arrs:known:vertex_swap} (\cref{eq:properties_max_arrs:known:vertex_swap:cuts:middle}), in which the value of $\neighborsarrsymbol$ is always $0$.

\begin{itemize}
\item Consider the case in which $\level{u}>0$. Then, we have that $\level{v}>0$ as well, and $\alpha\ge0$. It is easy to see that $\beta=0$: if $\beta>0$ we would have that $\level{z_j}<0$ (for all $j\in[1,\beta]$) but also $\level{v}>0$ which would contradict \cref{propos:bip_and_nonbip_arrangements:maximal_nonbip:near_Nurse}. Notice that $uv\notin E$, otherwise $v$ would be a thistle, and $uw_i\notin E$, $vw_i\notin E$ for all $i\in[1,\alpha]$. Moving $v$ to the left of $u$ consists of swapping first $v$ with $\thistlev$, then swapping $v$ with $w_\alpha$, and so on. This changes the cost of the arrangement by (\cref{lemma:properties_max_arrs:known:vertex_swap} -- \cref{eq:properties_max_arrs:known:vertex_swap:cuts:middle})
\begin{align*}
\varphi_{\arr}(u\leftarrow v)
	&= \level{v} - \level{\thistlev}
	+ \sum_{i=1}^\alpha \left( \level{v} - \level{w_{\alpha - i + 1}} \right)
	+ \level{v} - \level{u} \\
	&= (\alpha + 2)\level{v} - \level{u} - \level{\thistlev} - \sum_{i=1}^\alpha \level{w_i}.
\end{align*}
Notice that since $vw_i\notin E$, the terms of the form $a_{vw_i}$ in \cref{eq:properties_max_arrs:known:vertex_swap:cuts:middle} do not contribute to the equation above. Similarly, moving $u$ to the right of $v$, changes the cost of the arrangement by
\begin{align*}
\varphi_{\arr}(u\rightarrow v)
	&= \sum_{i=1}^\alpha \left( \level{w_i} - \level{u} \right)
	+ \level{\thistlev} - \level{u}
	+ \level{v} - \level{u} \\
	&= \level{v} - (\alpha + 2)\level{u} + \level{\thistlev} + \sum_{i=1}^\alpha \level{w_i}.
\end{align*}
Recall that we assumed $\level{v}>\level{u}$: if one of the changes $\varphi_{\arr}(u\leftarrow v)$ or $\varphi_{\arr}(u\rightarrow v)$ is positive then the arrangement is not maximal and thus we cannot have $\level{v}>\level{u}$. Now we show that one of the changes is positive by showing that both cannot be negative. If both were indeed negative, then
\begin{align*}
\varphi_{\arr}(u\leftarrow v) + \varphi_{\arr}(u\rightarrow v) &< 0, \\
(\alpha + 3)\level{v} - (\alpha + 3)\level{u} &< 0, \\
\level{v} &< \level{u},
\end{align*}
which contradicts our initial assumption that $\level{v}>\level{u}$. Therefore, these two changes do not {\em both} decrease the cost of the arrangement, and thus {\em at least one of the two} increases it.

\item Consider now the case in which $\level{u}<0$. Here we distinguish two cases.
	\begin{itemize}
	\item If $\level{v}<0$ then we just have to mirror the arrangement and apply the arguments above to the result.
	
	\item In case $\level{v}>0$ then we have to notice that $\alpha=0$ and $\beta=0$ and $uv\notin E$ (otherwise $u$ and $v$ would be thistles). Applying arguments similar to the ones above we see that,
	\begin{align*}
	\varphi_{\arr}(u\rightarrow v)
		&= \level{\thistlev} + \level{v} - 2\level{u},\\
	\varphi_{\arr}(u\leftarrow v)
		&= 2\level{v} - \level{\thistlev} - \level{u}.
	\end{align*}
	Again, assuming both changes are negative we reach the conclusion that
	\begin{equation*}
	\varphi_{\arr}(u\rightarrow v) + \varphi_{\arr}(u\leftarrow v) = 3\level{v} - 3\level{u} < 0
	\longleftrightarrow \level{v} < \level{u}
	\end{equation*}
	which also contradicts our initial assumption.
	\end{itemize}
\end{itemize}
\qed
\end{proof}

Next we prove a result that highlights an interesting substructural organization of the vertices in maximal non-bipartite arrangements of exactly one thistle vertex.

\begin{corollary}
\label{cor:bip_and_nonbip_arrangements:components_maximal_bipartite}
Let $\arr$ be a maximal non-bipartite arrangement of a tree $\tree$ with exactly one thistle vertex, say $\thistlev$. By removing $\thistlev$ from $\tree$ we obtain $\degree{\thistlev}$ connected components $\tree_i$ each of which is arranged in a maximal bipartite arrangement in $\arr$. That is, the arrangements $\arr(\tree_i)$ are maximal bipartite arrangements.
\end{corollary}
\begin{proof}
Recall that $\tree$ is a connected bipartite graph and thus exists a unique proper two-coloring. To prove the statement in the corollary, we focus on a specific $\tree_i$; let $w_i$ be the only vertex of $\tree_i$ such that $w_i\in\neighbors{\thistlev}$. Let $\sgn$ be the sign function so that $\sgn(x)$ denotes the sign of $x\in\reals$. It is easy to see that
\begin{itemize}
\item $\sgn(\level{w_i})=\sgn(\level{u})$, for all $u\in\tree_i\setminus\{w_i\}$ of the same color as that of $w_i$,
\item $\sgn(\level{w_i})=-\sgn(\level{v})$ for all $v\in\tree_i\setminus\{w_i\}$ of the opposite color to that of $w_i$,
\end{itemize}
since $\tree_i$ is arranged in a maximal bipartite arrangement, the vertices of the same color lie on the same side of the arrangement. Now, all non-neighbors of $\thistlev$ are arranged by non-increasing degree (\cref{lemma:bip_and_nonbip_arrangements:maximal_nonbip:one_thistle_and_neighbors}), thus no vertex $v$ is found between two vertices $u$, and vice versa. Since all the vertices to each side of $\thistlev$ are arranged by non-increasing degree (\cref{propos:bip_and_nonbip_arrangements:maximal_nonbip:near_Nurse}), vertex $w_i$ is arranged in $\arr$ in the correct position as per a non-increasing degree order among all vertices to the same side of $\thistlev$ as $w_i$ and thus among all vertices of $\tree_i$ of the same color as $w_i$. Thus $\arr(\tree_i)$ is a maximal bipartite arrangement.
\qed
\end{proof}

This observation needs not hold for bipartite graphs with cycles, or more general graphs, since the removal of these vertices may not disconnect the graph. Notice that \cref{cor:bip_and_nonbip_arrangements:components_maximal_bipartite} does not involve maximal non-bipartite arrangements with more than one thistle vertex. Furthermore, although \cref{cor:bip_and_nonbip_arrangements:components_maximal_bipartite} applies only to {\em maximal} non-bipartite arrangements with one thistle vertex, it is easy to see that the same property also holds in {\em maximum} arrangements regardless of its number of thistle vertices. This last claim is illustrated in \cref{fig:bip_and_nonbip_arrangements:maximal_nonbip:combination_bipartite} with two maximum (non-bipartite) arrangements of one and two thistle vertices respectively.

\begin{figure}
	\centering
	\includegraphics[scale=0.925]{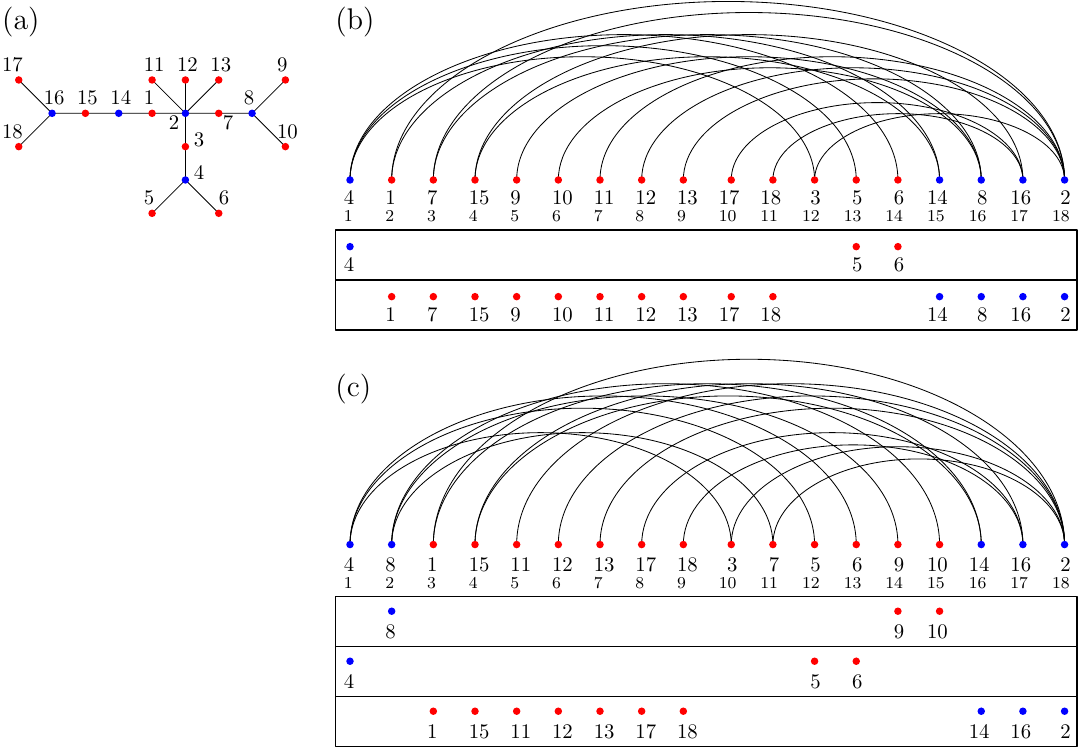}
	\caption{Maximum non-bipartite arrangements of the tree in (a) with cost $187$. Numbers indicate vertex labels. (b) A maximum non-bipartite arrangement with one thistle vertex ($3$). (c) A maximum non-bipartite arrangement with two thistle vertices ($3$ and $7$). In both (b) and (c), the arrangement of the vertices is split into the number of connected components that we obtain by removing the thistle vertices from the tree.} % L'arbre de la figura té "18 vèrtexos", "índex 80842" i "head vector: 0 1 2 3 4 4 2 7 8 8 2 2 2 1 14 15 16 16"
	\label{fig:bip_and_nonbip_arrangements:maximal_nonbip:combination_bipartite}
\end{figure}

Now follows a proof of correctness and cost analysis of \cref{algo:bip_and_nonbip_arrangements:NonBipartite_MaxLA_one_thistle}.

\begin{theorem}
\label{thm:bip_and_nonbip_arrangements:maximal_nonbip_one_thistle}
\cref{algo:bip_and_nonbip_arrangements:NonBipartite_MaxLA_one_thistle} constructs a maximal non-bipartite arrangement of any $n$-vertex tree $\tree=(V,E)$ with a single, given thistle vertex $\thistlev\in V$ in time $\bigO{nd2^d}$, where $d=\degree{\thistlev}$, and space $\bigO{n}$.
\end{theorem}
\begin{proof}
Let $\neighbors{\thistlev}=\{v_1,\dots,v_d\}$ be the neighbors of $\thistlev$ and let $\tree_1,\dots,\tree_d$ be the connected components of $\tree$ obtained after removing $\thistlev$ from $\tree$; let component $\tree_i$ contain vertex $v_i$. The central idea behind this algorithm is, given an assignment $\assignment \;:\; \neighbors{\thistlev} \rightarrow \{-1,+1\}$ of the side of the thistle to which each $v_i$ goes ($-1$ means {\tt left}, $+1$ means {\tt right}), to construct a maximal arrangement in which the neighbors of $\thistlev$ are placed according to $\assignment$. As explained in \cref{cor:bip_and_nonbip_arrangements:components_maximal_bipartite}, the remaining vertices are distributed so that $\tree_1,\dots,\tree_d$ are arranged in maximal bipartite arrangements. To construct this maximal non-bipartite arrangement with $\assignment$ one has to figure out the optimal placement of the vertices in $V'=V\setminus (\{\thistlev\} \cup \neighbors{\thistlev})$ with respect to $\thistlev$. The maximal non-bipartite arrangement of $\tree$ with $\thistlev$ as its only thistle vertex is the arrangement with the highest cost over all possible assignments $\assignment$.

Each assignment $\assignment$ determines the level value each vertex has to have; we use $\levelfunc \;:\; V \rightarrow \whole$ to denote such level values. The level of $\thistlev$ is, simply,
\begin{equation*}
\levelfunc(\thistlev) = \sum_{i=1}^{d} \assignment(v_i).
\end{equation*}
Since the best assignment of sides to the neighbors of $\thistlev$ is not known, the algorithm tries all $2^d$ possible assignments. However, some of these assignments can be filtered out, such as those for which $|\levelfunc(\thistlev)|=\degreesymbol$ since $\thistlev$ is not a thistle, and those for which $\levelfunc(\thistlev)<0$ since there exists a symmetric assignment $\assignment'$ (by simple mirroring) such that $\assignment'(u) = -\assignment(u)$ (for any $u\in V$) and thus $\levelfunc(\thistlev) = -\levelfunc[\assignment'](\thistlev)$. Notice that the vertices in $V\setminus \{\thistlev\}$ cannot be thistles by assumption and, as such, we must have that $|\levelfunc(u)| = \degree{u}$; we only have to determine the sign of the level value. For any $v_i\in\neighbors{\thistlev}$, let $\levelfunc(v_i) = - \assignment(v_i) \cdot \degree{v_i}$. The level value of the rest of the vertices switches sign at every level of a breadth first search traversal in the subtree $\subtree{v_i}[\thistlev]$, starting at $v_i$, compared to that in the previous level.

Therefore, the goal is, for a fixed thistle vertex $\thistlev$ and assignment $\assignment$, to construct an arrangement $\arr$ that is maximal among those arrangements such that
\begin{itemize}
\item The neighbors of $\thistlev$ are placed in the arrangement according to $\assignment$, and thus $\thistlev$ has level value $\levelfunc(\thistlev)$.

\item The rest of the vertices have level value as specified by $\levelfunc$ (otherwise, other thistles would be present).

\item \cref{propos:bip_and_nonbip_arrangements:maximal_nonbip:near_Nurse,lemma:bip_and_nonbip_arrangements:maximal_nonbip:one_thistle_and_neighbors} hold.
\end{itemize}

\begin{figure}
	\centering
	\includegraphics{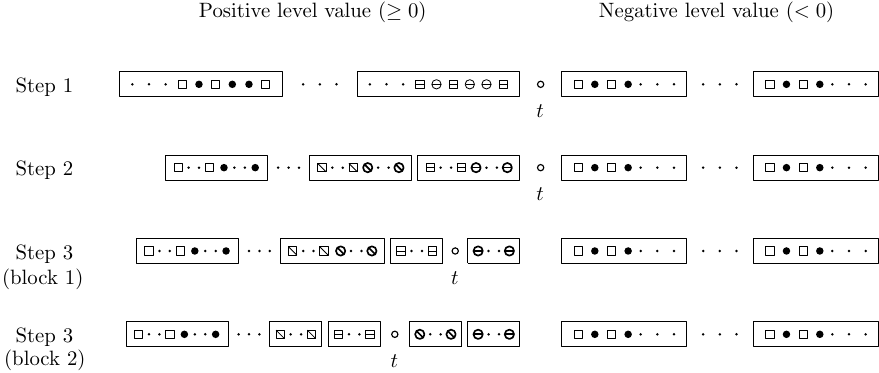}
	\caption{Steps to construct a maximal arrangement with one thistle vertex $\thistlev$. Squares (empty or stricken) denote vertices in $\neighbors{\thistlev}$, and circles (filled or stricken) denote vertices in $V\setminus(\{\thistlev\}\cup\neighbors{\thistlev})$. Vertices in the same rectangle are vertices of the same level value.}
	\label{fig:bip_and_nonbip_arrangements:maximal_nonbip_one_thistle:steps}
\end{figure}

The construction of the arrangement with $\thistlev$ as a thistle vertex and fixed assignment $\assignment$ consists of several steps (\cref{fig:bip_and_nonbip_arrangements:maximal_nonbip_one_thistle:steps}). W.l.o.g., we assume that $\levelfunc(\thistlev)\ge 0$.
\begin{enumerate}
\item \label{step:bip_and_nonbip_arrangements:maximal_nonbip_one_thistle:1} First, place the vertices in $V\setminus\{\thistlev\}$ with positive level value in non-increasing order. Then, place $\thistlev$ and, lastly, place the vertices in $V\setminus\{\thistlev\}$ with negative level value in non-increasing order. Let $\arr$ be the resulting arrangement. Notice that this initial arrangement satisfies the three conditions above. This step can be done in time $\bigO{n}$ using counting sort \parencite{Cormen2001a}.

\item \label{step:bip_and_nonbip_arrangements:maximal_nonbip_one_thistle:2} Permute all sequences of consecutive vertices of equal, and positive level value in $\arr$ so that we first find the neighbors of the thistle (if any) and then the non-neighbors (if any). Notice that $\thistlev$ can be in one of these sequences if the vertex immediately to its left in $\arr$ has the same level value as $\thistlev$. In this case, $\thistlev$ is to be placed between its neighbors and the non-neighbors of that level value. This keeps the relative order of the thistle and its neighbors and does not change the level signature of $\arr$. This can be seen as a normalization step that guarantees correctness of the next. This can be done in time $\bigO{n}$ using counting sort \parencite{Cormen2001a}, by assigning the value $0$ to the neighbors of the thistle, the value $1$ to the thistle, and the value $2$ to the other vertices.
\end{enumerate}

After Steps \ref{step:bip_and_nonbip_arrangements:maximal_nonbip_one_thistle:1} and \ref{step:bip_and_nonbip_arrangements:maximal_nonbip_one_thistle:2}, all vertices in $V'$ are arranged in a maximal bipartite arrangement with respect to their connected component, and vertex $\thistlev$ is placed so that its level value coincides with that given by $\assignment$, but $\thistlev$ may not be in its optimal position with respect to the vertices in $V'$. In order to find the maximal arrangement for fixed $\thistlev$ and $\assignment$, one has to figure out how to slide vertices in $V'$ in the arrangement while keeping their relative order among them as per \cref{lemma:bip_and_nonbip_arrangements:maximal_nonbip:one_thistle_and_neighbors}. To achieve this, it is sufficient to move vertices in $V'$ located at the left side of $\thistlev$ to the right of $\thistlev$, one vertex at a time, and starting at the rightmost vertex in $V'$ located at the left side of $\thistlev$ (Step \ref{step:bip_and_nonbip_arrangements:maximal_nonbip_one_thistle:3}). For the sake of clarity, we call this operation {\em relocation}.

Notice that there is no need to also try and relocate vertices in $V'$ from the right of $\thistlev$ to the left of $\thistlev$ since we either have that $\levelfunc(\thistlev)>0$ and this would decrease the cost of the arrangement, or $\levelfunc(\thistlev)=0$ and there exists another assignment $\assignment'$ for which all the neighbors of $\thistlev$ to its right in $\assignment$ are assigned to its left in $\assignment'$ and vice versa.

\begin{enumerate}
\setcounter{enumi}{2}
\item \label{step:bip_and_nonbip_arrangements:maximal_nonbip_one_thistle:3} Now, let $\arr$ be the result of the previous two steps, and of the form
\begin{equation*}
\arr = (
	\dots,
	\underbrace{u_1, \dots, u_{\delta}}_{\in\neighbors{\thistlev}},
	\overbrace{v_1, \dots, v_{\gamma}}^{\in V'},
	\underbrace{w_1, \dots, w_{\beta}}_{\in\neighbors{\thistlev}},
	\overbrace{y_1, \dots, y_{\alpha}}^{\in V'},
	\thistlev,
	\dots
),
\end{equation*}
where $v_i,y_i\in V'$, $u_i,w_i\in\neighbors{\thistlev}$, and
\begin{equation*}
\levelfunc(u_1)
	= \dots
	= \levelfunc(u_\delta)
	=
\levelfunc(v_1)
	= \dots
	= \levelfunc(v_\gamma)
>
\levelfunc(w_1)
	= \dots
	= \levelfunc(w_\beta)
= \levelfunc(y_1)
	= \dots
	= \levelfunc(y_\alpha).
\end{equation*}
For the sake of brevity, we use $\levelfunc(u)$ to denote $\levelfunc(u_k)$ for all $k\in[1,\delta]$,
$\levelfunc(v)$ to denote $\levelfunc(v_k)$ for all $k\in[1,\gamma]$, $\levelfunc(w)$ to denote $\levelfunc(w_k)$ for all $k\in[1,\beta]$ and $\levelfunc(y)$ to denote $\levelfunc(y_k)$ for all $k\in[1,\alpha]$.

The next step is to relocate vertices in $V'$, so as to increase the cost of the arrangement. That is, in $\arr$, find the vertex in $V'$ closest to $\thistlev$ to its left and relocate it. The first such vertex in $\arr$ is $y_\alpha$, the second is $y_{\alpha-1}$ and so on until reaching $y_1$. Then, the next such vertex before $y_1$ is $v_\gamma$, and so on. Notice that relocating a $v_j$ (or other vertex in $V'$ of higher level value than $\levelfunc(v)$) before all $y_k$ have been relocated would contradict \cref{lemma:bip_and_nonbip_arrangements:maximal_nonbip:one_thistle_and_neighbors} since $\levelfunc(v)>\levelfunc(y)$. During the process, keep track of the best arrangement after each vertex is relocated.

The cost of this step depends on the number of regions in the arrangement of vertices $(d_i, c_i)_{i=1}^{\kappa}$ of $d_i\ge1$ neighbors of $\thistlev$ followed by $c_i\ge1$ vertices from $V'$, with $\kappa \le \leftdeg{\thistlev}$ by construction; for the sake of simplicity, and only in the cost analysis of this step, we assume that these regions are determined by whether the vertices are neighbors of $\thistlev$ or not, rather than by level value. The cost of moving all vertices is
\begin{equation*}
c_1 + \sum_{i=2}^{\kappa} c_i \cdot \sum_{j=1}^{i - 1} d_j
	\le \sum_{i=1}^{\kappa} c_i \cdot \sum_{j=1}^{i} d_j
	\le d \sum_{i=1}^{\kappa} c_i
	\le d(n - d - 1).
\end{equation*}
The cost of this step is therefore $\bigO{nd}$.

In fact, there is no need to relocate all vertices in $V'$ located at the left of $\thistlev$. One can stop relocating vertices as soon as the cost decreases since relocating more vertices will further decrease the cost. To show this, first notice that relocating a vertex $y_k$ will increase the cost of the arrangement only when $\levelfunc(y_k)<\levelfunc(\thistlev)$ since the cost of swapping $y_k$ and $\thistlev$ is (\cref{eq:properties_max_arrs:known:vertex_swap:cuts:middle})
\begin{equation*}
\varphi_{\arr}(y_k \rightarrow \thistlev)
	= \levelfunc(\thistlev) - \levelfunc(y_k)
	= \levelfunc(\thistlev) - \levelfunc(y).
\end{equation*}
Obviously, if relocating any $y_k$ increases the cost of $\arr$, then so does relocating all $y_1,\dots,y_{k-1}$ since they all have the same level value. On the other hand, if relocating any $y_k$ decreases the cost, then so does relocating all $y_1,\dots,y_{k-1}$. Therefore, it is easy to see that one can relocate all vertices of the same level value in block. After relocating all $y_k$, we obtain (\cref{eq:properties_max_arrs:known:vertex_swap:cuts:middle})
\begin{equation*}
\arr_2 = (
	\dots,
	u_1,
	\dots,
	u_\delta,
	v_1,
	\dots,
	v_\gamma,
	w_1,
	\dots,
	w_\beta,
	\thistlev,
	y_1,
	\dots,
	y_\alpha,
	\dots
),
\end{equation*}
and the change in cost is
\begin{equation*}
\varphi_{\arr}(y_* \rightarrow \thistlev)
	= \sum_{k=1}^{\alpha} (\levelfunc(\thistlev) - \levelfunc(y_k))
	= \alpha \left( \levelfunc(\thistlev) - \levelfunc(y) \right).
\end{equation*}

We now prove the correctness of the stopping criterion. First, assume that relocating all $y_i$ (thus producing $\arr_2$) decreases the cost of the arrangement, that is, $\D[\arr][\tree]>\D[\arr_2][\tree]$. Under this assumption, it is easy to see that $\levelfunc(y) > \levelfunc(\thistlev)$ and thus $\levelfunc(v) > \levelfunc(y) > \levelfunc(\thistlev)$. We now show that further relocating any of $v_k$, thus producing $\arr_3$
\begin{equation*}
\arr_3 = (
	\dots,
	u_1,
	\dots,
	u_\delta,
	w_1,
	\dots,
	w_\beta,
	\thistlev,
	v_1,
	\dots,
	v_\gamma,
	y_1,
	\dots,
	y_\alpha,
	\dots
),
\end{equation*}
further decreases the cost, that is, $\D[\arr_2][\tree]>\D[\arr_3][\tree]$. The cost of relocating a single $v_k$ is
\begin{equation*}
\varphi_{\arr_2}(v_k \rightarrow \thistlev)
	= \sum_{i=1}^{\beta} (\levelfunc(w_i) - \levelfunc(v_k)) + \levelfunc(\thistlev) - \levelfunc(v_k)
	= \beta(\levelfunc(w) - \levelfunc(v_k)) + \levelfunc(\thistlev) - \levelfunc(v_k).
\end{equation*}
It is easy to see that $\varphi_{\arr_2}(v_j \rightarrow \thistlev)<0$ since
\begin{itemize}
\item $\levelfunc(v) > \levelfunc(w)$ by construction of $\arr$, and
\item $\levelfunc(v) > \levelfunc(\thistlev)$ since $\levelfunc(v)>\levelfunc(\thistlev)$ due to our assumption that $\varphi_{\arr}(y_* \rightarrow \thistlev)$.
\end{itemize}
Therefore, the cost of relocating all $v_j$,
\begin{equation*}
\varphi_{\arr_2}(v_* \rightarrow \thistlev)
	= \sum_{k=1}^{\gamma} \varphi_{\arr_2}(v_k \rightarrow \thistlev)
	= \gamma \left( \beta (\levelfunc(w) - \levelfunc(v)) + \levelfunc(\thistlev) - \levelfunc(v) \right).
\end{equation*}
Relocating more vertices further decreases the cost.

Finally, we need to prove the correctness of the criterion after performing all the relocations that increase the cost of the arrangement. Consider the result
\begin{equation*} 
\arr' = (
	\dots,
	\overbrace{p_1,\dots, p_\varepsilon}^{\in V'},
	\underbrace{q_1, \dots, q_\theta}_{\in\neighbors{\thistlev}},
	\overbrace{r_1, \dots, r_\mu}^{\in V'},
	\underbrace{s_1, \dots, s_\nu}_{\in\neighbors{\thistlev}},
	\thistlev,
	\dots
),
\end{equation*}
such that $\D[\arr'][\tree]>\D[\arr][\tree]$, $p_i, r_i \in V'$, $q_i, s_i\in\neighbors{\thistlev}$, where
\begin{equation*}
\levelfunc(p_1) = \dots = \levelfunc(p_\varepsilon)
>
\levelfunc(q_1) = \dots = \levelfunc(q_\theta) = \levelfunc(r_1) = \dots = \levelfunc(r_\mu)
>
\levelfunc(s_1) = \dots = \levelfunc(s_\nu).
\end{equation*}
and the vertices $r_k$ are the first vertices in $\arr$ that will decrease the cost of $\arr'$ if relocated. We prove now that if relocating all $r_k$ in $\arr'$, thus producing
\begin{equation*} 
\arr'' = (\dots, p_1,\dots, p_\varepsilon, q_1, \dots, q_\theta, s_1, \dots, s_\nu, \thistlev, r_1, \dots, r_\mu, \dots),
\end{equation*}
does not increase the cost of $\arr'$ ($\D[\arr'][\tree]\ge\D[\arr''][\tree]$) then further relocating $p_i$ in $\arr''$, thus producing
\begin{equation*} 
\arr''' = (\dots, q_1, \dots, q_\theta, s_1, \dots, s_\nu, \thistlev, p_1,\dots, p_\varepsilon, r_1, \dots, r_\mu, \dots),
\end{equation*}
will decrease the cost of $\arr''$ (thus $\D[\arr''][\tree]>\D[\arr'''][\tree]$). First (\cref{eq:properties_max_arrs:known:vertex_swap:cuts:middle}),
\begin{align*}
\varphi_{\arr'}(r_k \rightarrow \thistlev)
	&= \sum_{i=1}^\nu ( \levelfunc(s_i) - \levelfunc(r_k) ) + \levelfunc(\thistlev) - \levelfunc(r_k) \\
	&= -(\nu + 1)\levelfunc(r) + \nu\cdot\levelfunc(s) + \levelfunc(\thistlev), \\
\varphi_{\arr''}(p_k \rightarrow \thistlev)
	&= \sum_{i=1}^\theta (\levelfunc(q_i) - \levelfunc(p_k)) + \sum_{i=1}^\nu ( \levelfunc(s_i) - \levelfunc(p_k) ) + \levelfunc(\thistlev) - \levelfunc(p_k) \\
	&= -(\theta + \nu + 1)\levelfunc(p) + \theta\cdot\levelfunc(q) + \nu\cdot\levelfunc(s) + \levelfunc(\thistlev).
\end{align*}
We now show that if $\varphi_{\arr'}(r_k \rightarrow \thistlev)\le0$ then $\varphi_{\arr''}(p_k \rightarrow \thistlev)<0$ for a single $p_k$. Quite simply, notice that
\begin{equation*}
\varphi_{\arr''}(p_k \rightarrow \thistlev)
	= -(\nu + 1)\levelfunc(p) + \nu\cdot\levelfunc(s) + \levelfunc(\thistlev) + \theta(\levelfunc(q) - \levelfunc(p)),
\end{equation*}
for which it is easy to see that $\theta(\levelfunc(q) - \levelfunc(p)) < 0$ since $\levelfunc(p) > \levelfunc(q)$, and, easily, enough,
\begin{equation*}
-(\nu + 1)\levelfunc(p) + \nu\cdot\levelfunc(s) + \levelfunc(\thistlev)
<
-(\nu + 1)\levelfunc(r) + \nu\cdot\levelfunc(s) + \levelfunc(\thistlev)
=
\varphi_{\arr'}(r_k \rightarrow \thistlev)
\le
0
\end{equation*}
since $\levelfunc(p) > \levelfunc(r)$. Thus moving a single $p_k$ decreases the cost of the arrangement, and so does moving all $p_k$.
\end{enumerate}

Therefore, the exploration of all possible arrangements with $\thistlev$ as a thistle vertex with level value determined by $\assignment$ is complete and Steps \ref{step:bip_and_nonbip_arrangements:maximal_nonbip_one_thistle:1}, \ref{step:bip_and_nonbip_arrangements:maximal_nonbip_one_thistle:2}, \ref{step:bip_and_nonbip_arrangements:maximal_nonbip_one_thistle:3} find the maximal arrangement with $\thistlev$ and $\assignment$ fixed.
\qed
\end{proof}

Now we can easily tackle {\tt 1-thistle MaxLA} by applying \cref{algo:bip_and_nonbip_arrangements:NonBipartite_MaxLA_one_thistle} to potential thistles. For this, we define a set of vertices to test as thistle vertices. This set, denoted as $\epotthistlesT$, includes all potential thistles as defined above, and one internal vertex from antenna.

\begin{algorithm}
	\caption{{\tt 1-thistle MaxLA} in time $\bigO{\numepotthistlesT n\maxdegreesetT{\epotthistlessymbol}2^{\maxdegreesetT{\epotthistlessymbol}}}$ and space $\bigO{n}$.}
	\label{algo:bip_and_nonbip_arrangements:one_thistle_MaxLA}
	\DontPrintSemicolon
	
	\KwIn{$\tree=(V_1\cup V_2, E)$ a free tree.}
	\KwOut{A maximal non-bipartite arrangement of $\tree$ with exactly one thistle vertex.}
	
	\SetKwProg{Fn}{Function}{ is}{end}
	\Fn{\textsc{1-Thistle-MaxLA}$(\tree)$} {
		$\arr \gets \emptyset$ \tcp{empty arrangement of $n$-positions}
		\For {$\thistlev\in\epotthistlesT$} {
			$\arr_1\gets\textsc{1-Given-Thistle-MaxLA}(\tree, \thistlev)$ \tcp{\cref{algo:bip_and_nonbip_arrangements:NonBipartite_MaxLA_one_thistle}}
			$\arr\gets\max\{\arr, \arr_1\}$ \;
		}
		\Return $\arr$
	}
\end{algorithm}

\begin{theorem}
\label{thm:bip_and_nonbip_arrangements:1_thistle_MaxLA}
For a given tree $\tree\in\freetrees$, \cref{algo:bip_and_nonbip_arrangements:one_thistle_MaxLA} solves {\tt 1-thistle MaxLA} in space $\bigO{n}$ and time
\begin{itemize}
\item $\bigO{\numepotthistlesT n\maxdegreesetT{\epotthistlessymbol} 2^{\maxdegreesetT{\epotthistlessymbol}}}=\bigO{n^2\maxdegreeT 2^{\maxdegreeT}}$ on individual trees,
\item Typically $\bigO{n^3 \cdot \log{n}}$ over $\freetrees$,
\end{itemize}
where $\maxdegreesetT{\epotthistlessymbol}$ is the maximum degree over the vertices in $\epotthistlesT$, and $\maxdegreeT$ is the maximum degree of $\tree$.
\end{theorem}
\begin{proof}
Recall that $\epotthistlesT$ includes one degree-$2$ vertex from each antenna of the tree as vertices to inspect; the proof of \pathoptimizationO{} applies to {\em maximum} arrangements. We just need one internal vertex of each branchless path since, when it acts as a thistle, it can be exchanged with other internal vertices of the same branchless path by swapping vertices as in the proof of \pathoptimization. The correctness of the algorithm follows directly from \cref{thm:bip_and_nonbip_arrangements:maximal_nonbip_one_thistle}.

It is easy to see that it has time complexity
\begin{equation*}
\sum_{\thistlev\in\epotthistlesT} \bigO{n\degree{\thistlev}2^{\degree{\thistlev}}}
	= \bigO{\numepotthistlesT n \maxdegreesetT{\epotthistlessymbol} 2^{\maxdegreesetT{\epotthistlessymbol}}}.
\end{equation*}
In the worst case, the current definition of potential thistle vertices leads to $\maxdegreesetT{\epotthistlessymbol}=\maxdegreeT$ and $\numepotthistlesT=\bigO{n}$, a simpler upper bound of the cost is $\bigO{n^2\maxdegreeT2^{\maxdegreeT}}$.

Now, \textcite{Goh1994a} found that the maximum degree of an unlabeled free tree is typically $\bigTh{\log{n}}$ over $\freetrees$. Then the exponential cost of the algorithm will typically vanish and thus have typical complexity of $\bigO{\numepotthistlesT n^2 \cdot \log{n}}=\bigO{n^3 \cdot \log{n}}$.

%~ There is clearly a dependency between $\numpotthistlesT$ and the degrees of vertices in $\potthistlesT$ (obviously, $\numpotthistlesT \leq n - \maxdegreesetT{\potthistlessymbol}$). It is easy to see that the effective time complexity of the algorithm is at least cubic in $n$. By Jensen's inequality, we have that 
%~ \begin{equation}
%~ \frac{1}{n}\sum_{\thistlev\in\potthistlesT} 2^{\degree{\thistlev}}
	%~ \ge 2^{ (1/n) \sum_{\thistlev\in\potthistlesT}\degree{\thistlev}}
	%~ = 2^{2- 2/n}
	%~ = \bigTh{1}.
%~ \end{equation}
%~ Then the cost of the algorithm is 
%~ \begin{equation}
%~ n\sum_{\thistlev\in\potthistlesT} \frac{1}{n}\bigO{n2^{\degree{\thistlev}}}
	%~ \ge n^2\sum_{\thistlev\in\potthistlesT} \bigTh{1}
	%~ = \bigTh{n^2 \numpotthistlesT}
%~ \end{equation}
%~ Then we conclude that the cost of the algorithm is $\bigOm{n^2 \numpotthistlesT}$.
\qed
\end{proof}

\textcite{Nurse2018a,Nurse2019a} devised an algorithm to solve {\tt MaxLA} on directed trees, and then applied it to undirected trees; they showed that their algorithm has time complexity $\bigO{n^{4\maxdegreeT}}$. While their algorithm solves {\tt MaxLA} for all undirected trees, its cost is worse than exponential. Now, when considering that the maximum degree of an unlabeled free tree is typically $\bigTh{\log{n}}$ \parencite{Goh1994a} the complexity of the algorithm by Nurse and De Vos reduces to $\bigO{n^{4\log{n}}}$ while the complexity of \cref{algo:bip_and_nonbip_arrangements:one_thistle_MaxLA} reduces to $\bigO{n^3 \cdot \log{n}}$. Therefore, given the low proportion of trees with high maximum degree, we expect our algorithm to perform better (both in the worst and typical cases) than that by Nurse and De Vos at the expense of not solving {\tt MaxLA} for all trees.

%------------------------------------------------%
% </automatic inline of '4-3-non-bipartite.tex'> %
%------------------------------------------------%
%----------------------------------------------%
% </automatic inline of '4-0-bipartition.tex'> %
%----------------------------------------------%
%---------------------------------------------------%
% <automatic inline of '5-0-classes-of-graphs.tex'> %
%---------------------------------------------------%
\section{{\tt MaxLA} for specific classes}
\label{sec:max_for_classes}

In this section we identify two classes of graphs for which the solution to {\tt MaxLA} can be calculated in time and space $\bigO{n}$: $k$-regular graphs $\kregularclass$ and $k$-linear trees ($k\le2$).

%---------------------------------------------%
% <automatic inline of '5-1-0-k-regular.tex'> %
%---------------------------------------------%
\subsection{$k$-regular graphs for $k\le 2$}
\label{sec:max_for_classes:k_regular}

We provide results for $k\le2$: $\kregularclass[0]$ are graphs of $n$ vertices and no edges; $\kregularclass[1]$ are the $n$-vertex graphs ($n$ even) that are the disjoint union of $n/2$ pairs on vertices linked by and edge; $\kregularclass[2]$ are the $n$-vertex graphs whose connected components are all cycle graphs.

The cases for $k\le1$ are trivial. First, the solution to {\tt MaxLA} for any $\kregularclass[0]$ has cost $0$. Second, the solution to {\tt MaxLA} for any $G\in\kregularclass[1]$ is simple.

\begin{theorem}
Let $G\in\kregularclass[1]$ be a $1$-regular graph. We have that
\begin{equation*}
\DMax{G} = \frac{n^2}{4},
\end{equation*}
and the maximum arrangement is bipartite with level signature $(+1,\dots,+1,-1,\dots,-1)$.
\end{theorem}
\begin{proof}
Given any $G\in\kregularclass[1]$, the vertices of each connected component can only have level values $+1$ and $-1$ (because all vertices have degree $1$). Thus the maximum arrangement of $G$ is obtained by merging the maximum arrangements of each connected component following \cref{propos:properties_max_arrs:known:components_arranged_maximally,propos:properties_max_arrs:known:non_increasing_levsig,propos:properties_max_arrs:known:permutation_of_equal_level}. The result is a bipartite arrangement with $n/2$ vertices of level $+1$, followed by $n/2$ vertices of level $-1$. Constructing such arrangement of any $g\in\kregularclass[1]$ can be done in time $\bigO{n}$ and space $\bigO{1}$. Notice that we do not apply {\tt Bipartite MaxLA} (\cref{algo:bip_and_nonbip_arrangements:maximal_bip:algorithm:Bipartite_MaxLA}) directly for $G$ since it is defined for {\em connected} bipartite graphs. It is easy to see that the cost of the arrangement is
\begin{equation*}
\DMax{G} = (n-1) + (n-3) + \cdots + 1 = \frac{n^2}{4}.
\end{equation*}
\qed
\end{proof}

%-----------------------------------------------%
% <automatic inline of '5-1-1-two-regular.tex'> %
%-----------------------------------------------%
\subsubsection{$2$-regular}
\label{sec:max_for_classes:k_regular:2_regular}

Here we solve {\tt MaxLA} for the class of $2$-regular graphs $\kregularclass[2]$. We first derive a simple formula for the value of {\tt MaxLA} for cycle graphs $\cycle\in\cycleclass\subseteq\kregularclass[2]$. This result is a simple application of the characterization of maximum arrangements in \cref{sec:properties_max_arrs:known,sec:properties_max_arrs:new} and we use it to solve {\tt MaxLA} more easily for disconnected $2$-regular graphs. We also characterize the arrangements that yield the maximum cost for any $\cycle\in\cycleclass$.

\begin{lemma}
\label{thm:max_for_classes:k_regular:2_regular:cycle}
Let $\cycle\in\cycleclass$ be a cycle graph of $n\ge3$ vertices. We have that
\begin{equation*}
\DMax{\cycle} = \left\lfloor \frac{n^2}{2} \right\rfloor = \DMax{\pathgraph} + 1
\end{equation*}
for any $\pathgraph\in\pathgraphclass$ (\cref{eq:introduction:DMax_path_graphs}), and
\begin{enumerate}[(i)]
\item \label{thm:max_for_classes:k_regular:2_regular:cycle:even} For $n$ even, maximum arrangements are always bipartite, that is, $\maxarrset{\cycle}=\maxbiparrset{\cycle}$ and have level signature $(+2,\dots,+2,-2,\dots,-2)$.
\item \label{thm:max_for_classes:k_regular:2_regular:cycle:odd} For $n$ odd, maximum arrangements are always non-bipartite that is, $\maxarrset{\cycle}=\maxnonbiparrset{\cycle}$; these arrangements contain only one thistle vertex and have level signature $(+2,\dots,+2,0,-2,\dots,-2)$.
\end{enumerate}
\end{lemma}
\begin{proof}
We construct a maximum arrangement with $k$ thistle vertices. Since the sum of all level values is $0$ (\cref{eq:preliminaries:sum_levels_equals_0}) and the only level values possible are $+2$, $-2$ and $0$, the number of vertices with level value $+2$ must be the same as the number of vertices with level value $-2$. Let $x$ be said number and note that $2x + k = n$. Thus when $n$ is odd then $k$ is odd, and when $n$ is even then $k$ is even.

We shall see that the construction satisfies \Nurse. Due to \cref{propos:properties_max_arrs:known:non_increasing_levsig}, and irrespective of the parity of $n$, thistle vertices will be placed at the center of the arrangement with $(n-k)/2$ vertices to both sides, those to the left have level value $+2$ and those to the right have level value $-2$ (\cref{fig:max_for_classes:k_regular:2_regular:cycle:arrangement_for_cycle}), thus forming three groups of vertices. By \cref{propos:properties_max_arrs:known:permutation_of_equal_level}, permuting the vertices in each level group does not change the cost of the arrangement.

\begin{figure}
	\centering
	\includegraphics{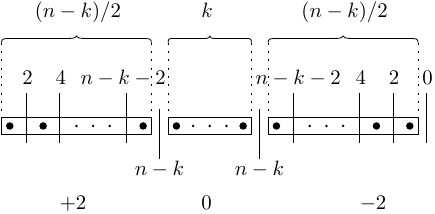}
	\caption{Proof of \cref{thm:max_for_classes:k_regular:2_regular:cycle}. An arrangement of a cycle graph $\cycle\in\cycleclass$ with $k$ thistle vertices in it. The vertices are distributed in a non-increasing order by level value, forming three groups of vertices `$+2$', `$0$', and `$-2$'. The number of vertices in each group is indicated atop the figure. Numbers directly above or below each vertical line indicates the cut width at that position of the arrangement.}
	\label{fig:max_for_classes:k_regular:2_regular:cycle:arrangement_for_cycle}
\end{figure}

To calculate the cost of an arrangement $\arr$ with $k$ thistle vertices (\cref{fig:max_for_classes:k_regular:2_regular:cycle:arrangement_for_cycle}) we only need to sum the cut widths (\cref{eq:preliminaries:arr_cost_as_sum_of_cuts}) which we know since these can be obtained using the level values (\cref{eq:preliminaries:cut_values:sum_of_levels}). Cut widths and level values are shown in \cref{fig:max_for_classes:k_regular:2_regular:cycle:arrangement_for_cycle}. The sum of all cut widths is
\begin{equation*}
\D[\arr][\cycle]
	= \sum_{i=1}^{(n - k)/2} 2i + (k - 1)(n - k) + \sum_{i=1}^{(n - k)/2} 2i = \frac{n^2 - k^2}{2}.
\end{equation*}
Now,
\begin{itemize}
\item[(\ref{thm:max_for_classes:k_regular:2_regular:cycle:even})] For $n$ even, the maximum cost is achieved at $k=0$, thus $\DMax{\cycle} = n^2/2$ and the arrangement is bipartite.
\item[(\ref{thm:max_for_classes:k_regular:2_regular:cycle:odd})] For $n$ odd, the maximum cost is achieved at $k=1$, thus $\DMax{\cycle} = (n^2 - 1)/2$ and the arrangement is necessarily non-bipartite since cycle graphs are not bipartite for odd $n$.
\end{itemize}
\qed
\end{proof}

Now we tackle the case the class of disconnected $2$-regular graphs $\kregularclass[2]\setminus\cycleclass$; to this aim, we use \cref{thm:max_for_classes:k_regular:2_regular:cycle} to obtain the maximum arrangement of each cycle in any such graph, and then merge these arrangements following \cref{propos:properties_max_arrs:known:components_arranged_maximally,propos:properties_max_arrs:known:non_increasing_levsig,propos:properties_max_arrs:known:permutation_of_equal_level} (\cref{algo:max_for_classes:k_regular:2_regular:disconnected_2_regular}).\footnote{Notice that this does not work in general. Consider the disjoint union of two path graphs of $2$ vertices each, say $a-b-c$ and $d-e-f$. The two possible arrangements constructed as in \cref{algo:max_for_classes:k_regular:2_regular:disconnected_2_regular} are $(b,d,f,a,c,e)$ with cost $14$, and $(b,d,e,f,a,c)$ with cost $12$.}

\begin{algorithm}
	\caption{{\tt MaxLA} for disconnected $2$-regular graphs in time and space $\bigO{n}$}
	\label{algo:max_for_classes:k_regular:2_regular:disconnected_2_regular}
	\DontPrintSemicolon
	
	\KwIn{$G$ a disconnected $2$-regular graph.}
	\KwOut{A maximum arrangement of $G$.}
	
	\SetKwProg{Fn}{Function}{ is}{end}
	\Fn{\textsc{MaxLA-Disconnected-$2$-Regular}$(G)$} {
		$\cycle_1, \dots, \cycle_p \gets$ All connected components of $G$. \tcp{Cost: $\bigO{n}$}
		
		$\maxarr_1, \dots, \maxarr_p \gets$ A maximum arrangement for each $\cycle_i$ (\cref{thm:max_for_classes:k_regular:2_regular:cycle}) \tcp{Cost: $\bigO{n}$}
		
		$\maxarr \gets$ Merge all arrangements $\maxarr_1, \dots, \maxarr_p$ as per \cref{propos:properties_max_arrs:known:components_arranged_maximally,propos:properties_max_arrs:known:non_increasing_levsig,propos:properties_max_arrs:known:permutation_of_equal_level} \tcp{Cost: $\bigO{n}$}
		\Return $\maxarr$
	}
\end{algorithm}

Next, we prove in \cref{thm:max_for_classes:k_regular:2_regular:disconnected_2_regular} the correctness of \cref{algo:max_for_classes:k_regular:2_regular:disconnected_2_regular}.
\begin{theorem}
\label{thm:max_for_classes:k_regular:2_regular:disconnected_2_regular}
Let $G\in\kregularclass[2]\setminus\cycleclass$ be any $n$-vertex disconnected $2$-regular graph. By definition, $G$ is the disjoint union of several cycle graphs $\cycle_1, \dots, \cycle_p$, with $p>1$.
\begin{enumerate}[(i)]
\item \label{thm:max_for_classes:k_regular:2_regular:disconnected_2_regular:level_signature} The set of maximum arrangements of $G$ all share the same level signature.
\item \label{thm:max_for_classes:k_regular:2_regular:disconnected_2_regular:construction} A maximum arrangement can be constructed in time and space $\bigO{n}$.
\item \label{thm:max_for_classes:k_regular:2_regular:disconnected_2_regular:types} When all $\cycle_i$ are bipartite, $\maxarrset{G}=\maxbiparrset{G}$. Conversely, when some of $\cycle_i$ is not bipartite, $\maxarrset{G}=\maxnonbiparrset{G}$; the number of thistle vertices equals the number of non-bipartite $\cycle_i$.
\end{enumerate}
\end{theorem}
\begin{proof}
~\\
\begin{enumerate}
\item[(\ref{thm:max_for_classes:k_regular:2_regular:disconnected_2_regular:level_signature})] Let $G$ be the disjoint union of several cycles. Notice that the level signature of the maximum linear arrangement of a cycle graph of even number of vertices is unique $(2,\dots,2,-2,\dots,-2)$ and so is of a cycle graph of odd number of vertices $(2,\dots,2,0,-2,\dots,-2)$. Thus if several maximum arrangements, each of a disjoint cycle graph, are combined following \cref{propos:properties_max_arrs:known:components_arranged_maximally,propos:properties_max_arrs:known:non_increasing_levsig,propos:properties_max_arrs:known:permutation_of_equal_level} there is also a unique non-decreasing level signature which give us necessarily the maximum linear arrangement.

\item[(\ref{thm:max_for_classes:k_regular:2_regular:disconnected_2_regular:construction})] The construction of maximum arrangement of $G$ can be done in time and space $\bigO{n}$ by (1) applying \cref{thm:max_for_classes:k_regular:2_regular:cycle} to each cycle individually (which leaves us with arrangements with vertex levels $+2$, $0$ and $-2$), and (2) combining all the resulting arrangements in a linear pass over all said arrangements.

\item[(\ref{thm:max_for_classes:k_regular:2_regular:disconnected_2_regular:types})] Obviously, if all $\cycle_i$ are bipartite, the maximum arrangements of the cycles will all be bipartite, and so will be their combination. Only one non-bipartite cycle is required to make the combination non-bipartite. The number of thistles in the maximum arrangement of $G$ equals the number of non-bipartite cycles since these are the only ones that produce cycles in their individual maximum arrangement, which is kept in their combination.
\end{enumerate}
\qed
\end{proof}
%------------------------------------------------%
% </automatic inline of '5-1-1-two-regular.tex'> %
%------------------------------------------------%

%----------------------------------------------%
% </automatic inline of '5-1-0-k-regular.tex'> %
%----------------------------------------------%
%--------------------------------------------%
% <automatic inline of '5-2-0-k-linear.tex'> %
%--------------------------------------------%
\subsection{$k$-linear trees for $k\le 2$}
\label{sec:max_for_classes:k_linear}

$k$-linear $\klinearclass$ trees are trees whose vertices of degree $\ge3$ all lie on a single induced path \parencite{Johnson2020a}. Concerning $k$-linear trees, we provide results for $k\le2$: path graphs $\pathgraphclass=\klinearclass[0]$, spider graphs $\spiderclass=\klinearclass[1]$, and $2$-linear trees $\twolinearclass=\klinearclass[2]$. A spider graph is a subdivision of a star graph (\cref{fig:max_for_classes:k_linear:spider_two_linear_drawing}(a)); when $n\ge3$, the only vertex of degree $3$ or higher is called the {\em hub}, denoted as $\h$. Moreover, a maximal path that starts at $\h$ and ends at a leaf is called a {\em leg} of the spider \parencite{Bennett2019a}. Spider graphs are also known as {\em generalized stars} and the legs are also known as {\em arms} \parencite{Johnson2018a}. $2$-linear trees are trees where there are exactly two vertices of degree $3$ or more \parencite{Johnson2020a} (\cref{fig:max_for_classes:k_linear:spider_two_linear_drawing}(b)). We call these vertices the hubs denoted, respectively, $\h$ and $\g$ where, w.l.o.g., we assume that $\degree{\h}\ge\degree{\g}$. We call the path of vertices between $\h$ and $\g$ the {\em bridge} of the $2$-linear tree; we also call each maximal path of vertices starting at a hub (excluding the bridge) a {\em leg} of that hub.\footnote{$2$-linear trees can be seen as two spider graphs joined by a path connecting both hubs.}

In spite of being vanishingly small among all $n$-vertex trees as $n$ tends to infinity, spider graphs and $2$-linear trees are interesting due to their number being exponential in $n$ \parencite{Johnson2020a,OEIS_kLinearTrees,OEIS_SpiderTrees}.

\begin{figure}
	\centering
	\includegraphics{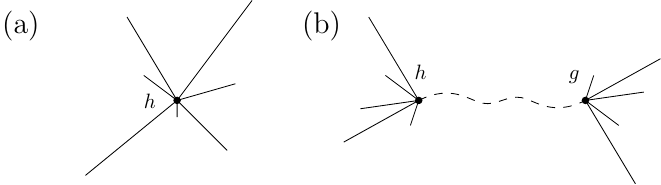}
	\caption{(a) A spider graph. Solid lines emanating from the hubs represent the legs of the spider. (b) A $2$-linear tree. Solid lines indicate the legs of the `spiders' whose hubs are $\h$ and $\g$, with $\degreeh,\degreeg\ge 3$. The dashed curve joining the two hubs $\h$ and $\g$ depicts the bridge of the $2$-linear tree, a branchless path where the endpoints have degree $\ge3$.}
	\label{fig:max_for_classes:k_linear:spider_two_linear_drawing}
\end{figure}

%-----------------------------------------------%
% <automatic inline of '5-2-1-path-graphs.tex'> %
%-----------------------------------------------%
\subsubsection{$0$-linear: path graphs}
\label{sec:max_for_classes:k_linear:path}

In \cref{sec:appendix:proof:classes:path_graphs}, we prove \cref{eq:introduction:DMax_path_graphs} as a warm up exercise for the next two sections. The proof uses arguments presented in this article, which ease and shorten the proof compared to previous results \parencite[Appendix A]{Ferrer2021a}.

\begin{corollary}
\label{cor:max_for_classes:k_linear:path}
Let $\pathgraph\in\pathgraphclass$ be an $n$-vertex path graph. $\pathgraph$ is maximizable only by bipartite arrangements, more formally, $\maxarrset{\pathgraph}=\maxbiparrset{\pathgraph}$, and \cref{eq:introduction:DMax_path_graphs} holds.
\end{corollary}
%------------------------------------------------%
% </automatic inline of '5-2-1-path-graphs.tex'> %
%------------------------------------------------%
%------------------------------------------------%
% <automatic inline of '5-2-2-spider-trees.tex'> %
%------------------------------------------------%
\subsubsection{$1$-linear: spider graphs}
\label{sec:max_for_classes:k_linear:spider}

In this section we prove that MaxLA is solvable in time $\bigO{n}$ for spider graphs. For this, we first prove that the hub of a spider graph cannot be a thistle in a maximum arrangement (\cref{thm:max_for_classes:k_linear:spider:no_nonbipartite_arrs}). The proof is given in \cref{sec:appendix:proof:classes:spiders}. Finally, we prove that maximum arrangements of spider graphs are always bipartite (\cref{cor:max_for_classes:k_linear:spider:linear_time_algo}).

% ------------------------------------------------------------------------------
\begin{theorem}
\label{thm:max_for_classes:k_linear:spider:no_nonbipartite_arrs}
Let $\spider\in\spiderclass$ be a spider graph and let $\maxarr$ be any of its maximum arrangements. The hub cannot be a thistle vertex in $\maxarr$.
\end{theorem}
% ------------------------------------------------------------------------------

Notice that $k$-quasi star graphs and quasi star graphs (introduced in \cref{sec:introduction}) are subclasses of $\spiderclass$. Finally, we prove the main claim of this section.

\begin{corollary}
\label{cor:max_for_classes:k_linear:spider:linear_time_algo}
All maximum arrangements of any $\spider\in\spiderclass$ are bipartite, therefore, {\tt MaxLA} is solvable in time and space $\bigO{n}$ for any $\spider\in\spiderclass$.
\end{corollary}
\begin{proof}

From \cref{thm:max_for_classes:k_linear:spider:no_nonbipartite_arrs} the hub cannot be a thistle in a maximum arrangement. From \pathoptimization{} we know that vertices of degree $2$ cannot be thistles in a maximum arrangement either; recall that leaves cannot be thistles. Therefore, we can conclude that no vertex of a $\spider\in\spiderclass$ can be a thistle in a maximum arrangement and thus $\maxarrset{\spider}=\maxbiparrset{\spider}$. That is, every maximum arrangement of all spider graphs is a maximal bipartite arrangement. From \cref{thm:bip_and_nonbip_arrangements:maximal_bip:algorithm} (and \cref{algo:bip_and_nonbip_arrangements:maximal_bip:algorithm:Bipartite_MaxLA}) we know that maximal bipartite arrangements can be constructed in time and space $\bigO{n}$ for all trees.
\qed
\end{proof}
%-------------------------------------------------%
% </automatic inline of '5-2-2-spider-trees.tex'> %
%-------------------------------------------------%
%----------------------------------------------------%
% <automatic inline of '5-2-3-two-linear-trees.tex'> %
%----------------------------------------------------%
\subsubsection{$2$-linear trees}
\label{sec:max_for_classes:k_linear:two_linear}

Now we prove that {\tt MaxLA} is polynomial-time solvable for all $2$-linear trees. We first prove that the hubs cannot be thistle vertices in maximum arrangements, which we use later to prove that {\tt MaxLA} is solvable in time $\bigO{n}$ for $2$-linear trees. To make the article more streamlined, the proof of \cref{thm:max_for_classes:k_linear:two_linear:single_thistle} is found in \cref{sec:appendix:proof:classes:two_linear_trees}.

% ------------------------------------------------------------------------------
\begin{theorem}
\label{thm:max_for_classes:k_linear:two_linear:single_thistle}
Let $\twolinear\in\twolinearclass$ be a $2$-linear tree and let $\maxarr$ be any of its maximum arrangements. There is at most one thistle vertex in $\maxarr$ and said vertex is located in the bridge of $\twolinear$.
\end{theorem}
% ------------------------------------------------------------------------------

It is easy to see now that {\tt MaxLA} is solvable for any $\twolinear\in\twolinearclass$ in time $\bigO{n}$ based on the characterization of its maximum arrangements given in \cref{thm:max_for_classes:k_linear:two_linear:single_thistle} and \pathoptimization. \cref{algo:max_for_classes:k_linear:two_linear:n_squared} simply takes the maximum between a maximal bipartite arrangement and a maximal non-bipartite arrangement where the thistle vertex is an arbitrary vertex in the bridge of $\twolinear$. This is formalized in \cref{cor:max_for_classes:k_linear:two_linear:MaxLA_solvable_n}.

\begin{corollary}
\label{cor:max_for_classes:k_linear:two_linear:MaxLA_solvable_n}
{\tt MaxLA} is solvable in time and space $\bigO{n}$ for any $\twolinear\in\twolinearclass$.
\end{corollary}
\begin{proof}
Let $\twolinear\in\twolinearclass$ be a $2$-linear tree. A maximum arrangement is the arrangement with the higher cost between a maximal bipartite arrangement and a maximal non-bipartite arrangement. Recall that maximal bipartite arrangements of $n$-vertex connected bipartite graphs can be constructed in time and space $\bigO{n}$ (\cref{thm:bip_and_nonbip_arrangements:maximal_bip:algorithm,algo:bip_and_nonbip_arrangements:maximal_bip:algorithm:Bipartite_MaxLA}).

It remains to show that the construction of a maximal non-bipartite arrangement of $\twolinear$ can be done in time $\bigO{n}$. \cref{thm:max_for_classes:k_linear:two_linear:single_thistle} states that the only potential thistle vertex of $\twolinear$ is located in its bridge, a branchless path. Recall that a maximal non-bipartite arrangement of a single thistle vertex of degree $2$ can be constructed in time and space $\bigO{n}$ (\cref{thm:bip_and_nonbip_arrangements:maximal_nonbip_one_thistle,algo:bip_and_nonbip_arrangements:NonBipartite_MaxLA_one_thistle}). Therefore, the time and space needed to construct a maximal non-bipartite arrangement of $\twolinear$ is $\bigO{n}$.
\qed
\end{proof}

\begin{algorithm}
	\caption{{\tt MaxLA} for every $2$-linear trees in time $\bigO{n}$ and space $\bigO{n}$.}
	\label{algo:max_for_classes:k_linear:two_linear:n_squared}
	\DontPrintSemicolon
	
	\KwIn{$\twolinear\in\twolinearclass$ a $2$-linear tree.}
	\KwOut{A maximum arrangement of $\twolinear$.}
	
	\SetKwProg{Fn}{Function}{ is}{end}
	\Fn{\textsc{MaximumArrangementTwoLinear}$(\twolinear)$} {
		\tcp{Maximal bipartite arrangement of $\twolinear$. \cref{algo:bip_and_nonbip_arrangements:maximal_bip:algorithm:Bipartite_MaxLA}.}
		$\maxbiparr \gets$ \textsc{BipartiteMaxLA}$(\twolinear)$ \tcp{Cost: $\bigO{n}$}
		
		$\thistlev\gets$ an arbitrary internal vertex from the bridge of $\twolinear$ \tcp{Cost: $\bigO{n}$}
		
		\tcp{Maximal non-bipartite arrangement of $\twolinear$ with exactly one thistle. \cref{algo:bip_and_nonbip_arrangements:NonBipartite_MaxLA_one_thistle}.}
		$\maxnonbiparr\gets$\textsc{1-Given-Thistle-MaxLA}$(\twolinear, \thistlev)$ \tcp{Cost: $\bigO{n}$}
		\tcp{Return the maximum arrangement, breaking ties arbitrarily.}
		
		\Return max$(\twolinear;\maxbiparr,\maxnonbiparr)$ \tcp{Cost: $\bigO{n}$}
	}
\end{algorithm}
%-----------------------------------------------------%
% </automatic inline of '5-2-3-two-linear-trees.tex'> %
%-----------------------------------------------------%
%---------------------------------------------%
% </automatic inline of '5-2-0-k-linear.tex'> %
%---------------------------------------------%
%----------------------------------------------------%
% </automatic inline of '5-0-classes-of-graphs.tex'> %
%----------------------------------------------------%
%-------------------------------------------%
% <automatic inline of '6-conclusions.tex'> %
%-------------------------------------------%
\section{Conclusions}
\label{sec:conclusions}

In this paper we have tackled {\tt MaxLA} using different approaches.

In \cref{sec:properties_max_arrs:new}, we have put forward new properties of maximum arrangements of graphs: we have shown that the vertices of antennas cannot be thistle vertices in maximum arrangements (\pathoptimizationO), and that at most one of the internal vertices of bridges can be a thistle vertex in a maximum arrangement (\pathoptimizationOO). In \alternation{} we have specified how the level values are distributed along the vertices of a branchless path and in \propagation{} we have shown that the level values of vertices in a branchless path can be expressed as a function of the level values of other vertices.
Solving {\tt MaxLA} for graphs containing sufficiently long paths is now easier.

In \cref{sec:bip_and_nonbip_arrangements}, we have presented a division of the $n!$ arrangements of any graph into {\em bipartite} and {\em non-bipartite} arrangements, which encompasses previous known solutions of {\tt MaxLA}. In \cref{sec:bip_and_nonbip_arrangements:maximal_bip}, we have devised an optimal algorithm to solve {\tt Bipartite MaxLA} for any connected bipartite graph in time $\bigO{n}$. In \cref{sec:bip_and_nonbip_arrangements:maximal_bip:relationship}, we have improved the $2$-approximation factor for {\tt bipartite MaxLA} by \textcite{Hassin2001a} down to $3/2$ in trees. In \cref{sec:bip_and_nonbip_arrangements:maximal_nonbip}, we have presented some properties of maximal non-bipartite arrangements (\cref{propos:bip_and_nonbip_arrangements:maximal_nonbip:near_Nurse,lemma:bip_and_nonbip_arrangements:maximal_nonbip:one_thistle_and_neighbors}); we have also introduced a new variant of {\tt MaxLA} called {\tt 1-thistle MaxLA} and devised an algorithm \cref{algo:bip_and_nonbip_arrangements:one_thistle_MaxLA,thm:bip_and_nonbip_arrangements:1_thistle_MaxLA}) to solve it for trees that runs in time $\bigO{n^2\maxdegreeT2^{\maxdegreeT}}$; typically, over $\freetrees$, this algorithm runs in time $\bigO{n^3 \cdot \log{n}}$.

Concerning results on specific classes of graphs, we applied, in \cref{sec:max_for_classes}, the new characterization presented in \cref{sec:properties_max_arrs:new} to solve {\tt MaxLA} in four classes of graphs. We have (1) proven that cycle graphs are maximizable only by bipartite (resp. non-bipartite) arrangements when $n$ is even (resp. odd) (\cref{thm:max_for_classes:k_regular:2_regular:cycle}) and used it to prove that {\tt MaxLA} can be solved in time and space $\bigO{n}$ for disconnected $2$-regular graphs (\cref{thm:max_for_classes:k_regular:2_regular:disconnected_2_regular}), (2) proven that $k$-linear trees with $k \in \{0, 1\}$, that is path graphs and spider graphs, are maximizable only by bipartite arrangements (\cref{cor:max_for_classes:k_linear:path}, and \cref{thm:max_for_classes:k_linear:spider:no_nonbipartite_arrs,cor:max_for_classes:k_linear:spider:linear_time_algo} respectively), (3) proven that {\tt MaxLA} can be solved in time $\bigO{n}$ for $2$-linear trees (\cref{thm:max_for_classes:k_linear:two_linear:single_thistle,cor:max_for_classes:k_linear:two_linear:MaxLA_solvable_n}).
%--------------------------------------------%
% </automatic inline of '6-conclusions.tex'> %
%--------------------------------------------%
%-------------------------------------------%
% <automatic inline of '7-future-work.tex'> %
%-------------------------------------------%
\section{Future work}
\label{sec:future_work}

In this paper we could not solve {\tt MaxLA} completely, and there is still more work to be done to find a polynomial time algorithm that solves {\tt MaxLA} for trees, should such an algorithm exist. Future research should strive to devise said algorithm. Other interesting lines of research are the generalization of \cref{algo:bip_and_nonbip_arrangements:maximal_bip:algorithm:Bipartite_MaxLA} to disconnected bipartite graphs which may help improve the complexity of {\tt 1-thistle MaxLA}. Furthermore, notice that {\tt bipartite MaxLA} and {\tt 1-thistle MaxLA} can be seen as particular cases of {\tt$k-$thistle MaxLA}, for $k=0$ and $k=1$, respectively, where the goal is to find a maximum arrangement among those with exactly $k$ thistles. Future work should include the investigation of {\tt $k-$thistle MaxLA} for $k\ge2$. Our preliminary evidence suggests that $k$ cannot be too high with respect to $n$ in order to offer a competitive approximation to unconstrained {\tt MaxLA} (notice that a large number of thistle vertices would jeopardize the ability to reach the unconstrainted optimum; it is easy to see that if many vertices are thistles, very long edges cannot be formed).

Although we improved the approximation factor by \textcite{Hassin2001a} in trees, future work should investigate tighter upper bounds on the approximation factor of {\tt bipartite MaxLA}. It should be also investigated whether the combination of {\tt bipartite MaxLA} and {\tt 1-thistle MaxLA} is advantageous empirically with respect to the algorithm by Nurse and De Vos for {\tt MaxLA} since the combination has lower running time and may solve {\tt MaxLA} frequently enough. Therefore, future work should characterize the trees that are maximizable by {\tt $k-$thistle MaxLA} with $k \in \{0,1\}$ and evaluate the practical impact of the present work, for instance, the frequency with which $k=0$ solves {\tt MaxLA} and the frequency with which $k \leq 1$ solves {\tt MaxLA}. Although we could not find a polynomial time algorithm for {\tt MaxLA} on trees, our work has paved the way for speeding up brute force exact algorithms (such as Constraint Satisfaction models) by applying the new properties presented in this article.
%--------------------------------------------%
% </automatic inline of '7-future-work.tex'> %
%--------------------------------------------%

\section*{Declarations}

\paragraph{Funding} The authors' research is supported by a recognition 2021SGR-Cat (01266 LQMC) from AGAUR (Generalitat de Catalunya). LAP is supported by Secretaria d'Universitats i Recerca de la Generalitat de Catalunya and the Social European Fund and a Ph.D. contract extension funded by Banco Santander. JLE is supported by grant PID2022-138506NB-C22 funded by Agencia Estatal de Investigación. LAP and RFC are supported by the grant PID2024-155946NB-I00 funded by Ministerio de Ciencia, Innovación y Universidades (MICIU), Agencia Estatal de Investigación (AEI/10.13039/501100011033) and the European Social Fund Plus (ESF+).

\paragraph{Data availability} There exists no data associated with this manuscript.

\paragraph{Conflict of interest} The authors declare that they have no conflict of interest.

\appendix
%-------------------------------------------------------%
% <automatic inline of '8A-proof-properties-known.tex'> %
%-------------------------------------------------------%
\section{Proofs of known properties}
\label{sec:appendix:proof:properties:known}

\begin{proof}[\cref{lemma:properties_max_arrs:known:vertex_swap}]
To prove \cref{lemma:properties_max_arrs:known:vertex_swap}, we use \cref{fig:proof_vertex_swap_lemma:illustration} as a visual aid to the proof of each item. As in its formulation, let $\graph=(V,E)$ be any graph where $V=\{u_1,\dots,u_n\}$, and $\arr$ an arrangement of its vertices
\begin{equation*}
\arr=(u_1,\dots,u_{i-1},v,u_{i+1},\dots,u_{j-1},w,u_{j+1},\dots,u_n)
\end{equation*}
where $v=u_i$ and $w=u_j$ for two fixed $i,j\in[n]$, and let
\begin{equation*}
\arr'=(u_1,\dots,u_{i-1},w,u_{i+1},\dots,u_{j-1},v,u_{j+1},\dots,u_n)
\end{equation*}
be the resulting arrangement of swapping vertices $w$ and $v$ and leaving the other vertices in their original positions.

\begin{figure}
	\centering
	\includegraphics{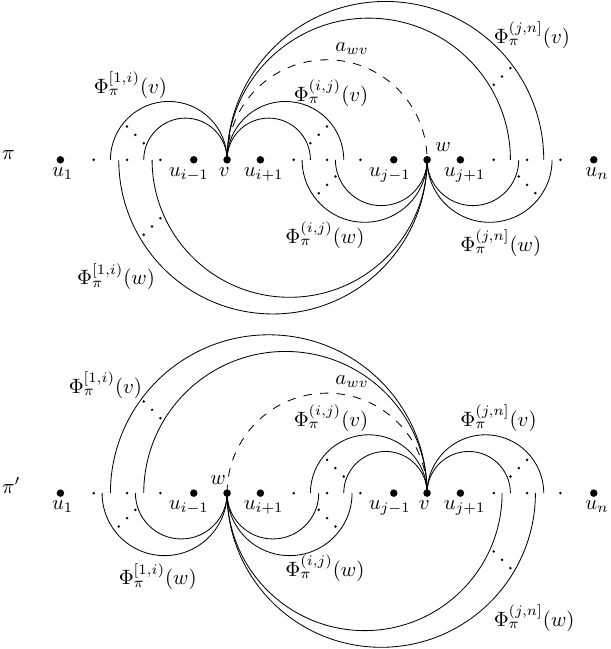}
	\caption{\cref{lemma:properties_max_arrs:known:vertex_swap}.}
	\label{fig:proof_vertex_swap_lemma:illustration}
\end{figure}

We first prove the equations related to level values.

\begin{itemize}
\item
	\begin{proof}[\cref{eq:properties_max_arrs:known:vertex_swap:levels:left,eq:properties_max_arrs:known:vertex_swap:levels:right}]
	It is easy to see that the levels of the vertices in $[1,i) \cup (j,n]$ do not change: the only vertices that change positions are $v$ and $w$ but these do not perturb the directional degree of any vertex in $[1,i) \cup (j,n]$.
	\qed
	\end{proof}

\item
	\begin{proof}[\cref{eq:properties_max_arrs:known:vertex_swap:levels:middle}]
	For this, simply consider any vertex $u_k$, for $k\in(i,j)$, and the four cases of existence of edge $vu_k$ and of edge $wu_k$:
	\begin{alignat*}{3}
	(1) & \quad &    a_{vu_k}=0, & \quad a_{wu_k}=0,    & \qquad \level{u_k}[\arr'] - \level{u_k} &= 0, \\
	(2) & \quad &    a_{vu_k}=0, & \quad a_{wu_k}=1,    & \qquad \level{u_k}[\arr'] - \level{u_k} &= -2, \\
	(3) & \quad &    a_{vu_k}=1, & \quad a_{wu_k}=0,    & \qquad \level{u_k}[\arr'] - \level{u_k} &= +2, \\
	(4) & \quad &    a_{vu_k}=1, & \quad a_{wu_k}=1,    & \qquad \level{u_k}[\arr'] - \level{u_k} &= 0.
	\end{alignat*}
	With these results, \cref{eq:properties_max_arrs:known:vertex_swap:levels:middle} follows easily.
	\qed
	\end{proof}

\item
	\begin{proof}[\cref{eq:properties_max_arrs:known:vertex_swap:levels:w}]
	First, notice that
	\begin{align}
	\label{eq:proof_vertex_swap_lemma:w:arr1}
	\level{w}        &= |\neighborsarr{w}{(j,n]}| - \left(a_{vw} + |\neighborsarr{w}{(i,j)}| + |\neighborsarr{w}{[1,i)}|\right) \\
	\label{eq:proof_vertex_swap_lemma:w:arr2}
	\level{w}[\arr'] &= |\neighborsarr{w}{(j,n]}| + a_{vw} + |\neighborsarr{w}{(i,j)}| - |\neighborsarr{w}{[1,i)}|.
	\end{align}
	Then, it is easy to see that \cref{eq:properties_max_arrs:known:vertex_swap:levels:w} follows from the difference of \cref{eq:proof_vertex_swap_lemma:w:arr2} and \cref{eq:proof_vertex_swap_lemma:w:arr1}.
	\qed
	\end{proof}

\item
	\begin{proof}[\cref{eq:properties_max_arrs:known:vertex_swap:levels:v}]
	As before, first notice that
	\begin{align}
	\label{eq:proof_vertex_swap_lemma:v:arr1}
	\level{v}        &= |\neighborsarr{v}{(j,n]}| + a_{vw} + |\neighborsarr{v}{(i,j)}| - |\neighborsarr{v}{[1,i)}| \\
	\label{eq:proof_vertex_swap_lemma:v:arr2}
	\level{v}[\arr'] &= |\neighborsarr{v}{(j,n]}| - \left(a_{vw} + |\neighborsarr{v}{(i,j)}| + |\neighborsarr{v}{[1,i)}| \right).
	\end{align}
	Then, it is easy to see that \cref{eq:properties_max_arrs:known:vertex_swap:levels:v} follows from the difference of \cref{eq:proof_vertex_swap_lemma:v:arr2} and \cref{eq:proof_vertex_swap_lemma:v:arr1}.
	\qed
	\end{proof}

Now we prove the equations related to cut widths.

\item
	\begin{proof}[\cref{eq:properties_max_arrs:known:vertex_swap:cuts:left,eq:properties_max_arrs:known:vertex_swap:cuts:right}]
	It is easy to see that the edges in the cuts in the positions in $[1,i)$ and in $[j,n)$ do not change after swapping vertices $v$ and $w$.
	\qed
	\end{proof}

\item
	\begin{proof}[\cref{eq:properties_max_arrs:known:vertex_swap:cuts:middle}]
	First, recall that (\cref{eq:preliminaries:cut_values:sum_of_levels}), for any $k\in[i,j)$,
	\begin{align*}
	\cut{k}        &= \level{u_1} +        \cdots + \level{u_{i-1}} +        \level{v} +        \level{u_{i+1}} +        \cdots + \level{u_k} \\
	\cut{k}[\arr'] &= \level{u_1}[\arr'] + \cdots + \level{u_{i-1}}[\arr'] + \level{w}[\arr'] + \level{u_{i+1}}[\arr'] + \cdots + \level{u_k}[\arr'].
	\end{align*}
	Applying \cref{eq:properties_max_arrs:known:vertex_swap:levels:left,eq:properties_max_arrs:known:vertex_swap:levels:middle} we get that
	\begin{equation*}
	\cut{k}[\arr'] - \cut{k}
		= \level{w}[\arr'] - \level{v} + 2\sum_{p=i+1}^k (a_{vu_p} - a_{wu_p}).
	\end{equation*}
	By further applying \cref{eq:properties_max_arrs:known:vertex_swap:levels:w} and regrouping terms inside the summation we get that
	\begin{equation*}
	\cut{k}[\arr'] - \cut{k}
		= \level{w} - \level{v} + 2\left(a_{vw} + |\neighborsarr{w}{(i,j)}| + |\neighborsarr{v}{(i,k]}| - |\neighborsarr{w}{(i,k]}|\right).
	\end{equation*}
	\qed
	\end{proof}

\end{itemize}
\qed
\end{proof}

\begin{proof}[\cref{propos:properties_max_arrs:known:components_arranged_maximally}]
By contrapositive. Let $\arr$ be an arrangement of $G$ where $\arr(H_i)$, the arrangement $\arr$ restricted to $H_i$, is not maximum. Then we can choose an arrangement of $H_i$ with higher cost than that of $\arr(H_i)$, and rearrange the vertices of $H_i$ within $\arr$ accordingly. The resulting arrangement has a higher cost than $\arr$ thus contradicting its maximality.
\qed
\end{proof}

The proofs of the remaining properties are all by contradiction.

\begin{proof}[\cref{propos:properties_max_arrs:known:non_increasing_levsig}]
Consider a maximum arrangement $\arr$ of $\graph$ where, by way of contradiction, there are two vertices $v$ and $w$ such that $\level{v} < \level{w}$ and $\arr(v) + 1 = \arr(w)$. Let $i=\arr(v)$ and $j=\arr(w)$. Swap vertices $v$ and $w$ to obtain $\arr'$. Due to \cref{lemma:properties_max_arrs:known:vertex_swap} we have that
\begin{alignat*}{2}
\cut{k}[\arr'] - \cut{k} &= 0,                               & \quad & \forall           k\in[1,i)\cup[j,n] \\
\cut{i}[\arr'] - \cut{i} &= \level{w} - \level{v} + 2a_{wv}. & \quad &
\end{alignat*}
Since $\level{v}<\level{w}$, we have that $\cut{i}[\arr'] - \cut{i}>0$. Therefore $\D[\arr'] > \D$ (\cref{eq:preliminaries:arr_cost_as_sum_of_cuts}).
\qed
\end{proof}

\begin{proof}[\cref{propos:properties_max_arrs:known:no_neighs_in_same_level}]
Consider a maximum arrangement $\arr$ of $\graph$,
\begin{equation*}
\arr = (u_1,\dots,u_{a-1},u_a,\dots,u_{i-1},v,u_{i+1}, \dots, u_{j-1},w,u_{j+1},\dots,u_b,u_{b+1},\dots,u_n),
\end{equation*}
in which $\level{u_{a-1}}>\level{u_a}=\level{u_b}>\level{u_{b+1}}$ for some $a,b$ such that $1\le a\le i<j\le b\le n$, and assume that $vw\in E(\graph)$. By \cref{propos:properties_max_arrs:known:non_increasing_levsig} we have that $\level{u_a}=\cdots=\level{u_b}$.

First, consider the case in which $v$ and $w$ have no neighbors in the interval $(i,j)$. Swapping $v$ and $w$ produces
\begin{equation*}
\arr' = (u_1,\dots,u_{a-1},u_a,\dots,u_{i-1},w,u_{i+1}, \dots, u_{j-1},v,u_{j+1},\dots,u_b,u_{b+1},\dots,u_n)
\end{equation*}
in which, by \cref{lemma:properties_max_arrs:known:vertex_swap},
\begin{equation*}
\cut{k}[\arr'] - \cut{k} = 0   \;\; \forall k\in[1,i), \quad
\cut{k}[\arr'] - \cut{k} = 2   \;\; \forall k\in[i,j), \quad \text{and} \quad
\cut{k}[\arr'] - \cut{k} = 0   \;\; \forall k\in[j,n].
\end{equation*}
Then $\cut{k}[\arr'] > \cut{k}$ for all $k\in[i,j)$ which implies that $\D[\arr'] > \D$ (\cref{eq:preliminaries:arr_cost_as_sum_of_cuts}).

In case either of $v$ or $w$ have neighbors in the interval $(i,j)$ then there exists another pair of vertices $u_p$ and $u_q$, such that $|p-q|<|i-j|$, both $p,q\in[i,j]$, $u_pu_q\in E(\graph)$ and vertices $u_p$ and $u_q$ have no neighbors in the interval $(p,q)$. We can apply the argument above to this pair of vertices to reach the same contradiction.
\qed
\end{proof}

\begin{proof}[\cref{propos:properties_max_arrs:known:permutation_of_equal_level}]
Consider a maximum arrangement $\maxarr$ of $\graph$,
\begin{equation*}
\arr = (u_1,\dots,u_{p-1},u_p,\dots,u_{i-1},v,u_{i+1}, \dots, u_{j-1},w,u_{j+1},\dots,u_q,u_{q+1},\dots,u_n),
\end{equation*}
in which $\level{u_{p-1}}>\level{u_p}=\level{u_q}>\level{u_{q+1}}$ for some $1\le p<q\le n$. By \cref{propos:properties_max_arrs:known:non_increasing_levsig}, the vertices $u_p,\dots,u_q$ have the same level value, that is, $\level{u_p}=\cdots=\level{u_q}$. Due to \cref{propos:properties_max_arrs:known:no_neighs_in_same_level}, $a_{vw}=0$, and also $v$ and $w$ do not have neighbors in the interval of positions $(i,j)$ in $\arr$. Let $\arr'$ be the result of swapping $v$ and $w$ in $\arr$. Due to \cref{lemma:properties_max_arrs:known:vertex_swap}, $\cutsignatureG[\arr']=\cutsignatureG$ and therefore, due to \cref{eq:preliminaries:arr_cost_as_sum_of_cuts}, $\D[\arr'] = \D$.
\qed
\end{proof}
%--------------------------------------------------------%
% </automatic inline of '8A-proof-properties-known.tex'> %
%--------------------------------------------------------%
%----------------------------------------------------%
% <automatic inline of '8B-proof-classes-paths.tex'> %
%----------------------------------------------------%
\section{Proof of \cref{eq:introduction:DMax_path_graphs}}
\label{sec:appendix:proof:classes:path_graphs}

\begin{proof}
In a maximum arrangement $\arr$ for $\pathgraph\in\pathgraphclass$, the non-leaf vertices must have level $\pm2$ (\pathoptimizationO); the leaves of $\pathgraph$ can only have level $\pm1$ (by definition). Also, recall that the level values must alternate sign (\alternationib). Therefore, we have that $|\level{v}|=\degree{v}$ for all vertices $v$, and thus every maximum arrangement is a maximal bipartite arrangement, hence $\maxarrset{\pathgraph}=\maxbiparrset{\pathgraph}$.

We derive an expression for the value of $\DMax{\pathgraph}$ in a similar way to how we did it for cycle graphs (\cref{thm:max_for_classes:k_regular:2_regular:cycle}). Due to \Nurse, the vertices have to be arranged in $\arr$ by non-increasing level value, breaking ties arbitrarily, and there are no two adjacent vertices in the tree with the same level value. For $n$ even,
\begin{equation*}
\DMax{\pathgraph}
	= \sum_{i=1}^{(n-2)/2} 2i + (n-1) + \sum_{i=1}^{(n-2)/2} 2i
	= \frac{n^2}{2} - 1,
\end{equation*}
and for $n$ odd,
\begin{equation*}
\DMax{\pathgraph}
	= \sum_{i=1}^{\lceil(n-2)/2\rceil} 2i + (n-2) + \sum_{i=1}^{\lfloor(n-2)/2\rfloor} 2i
	= \frac{n^2 - 1}{2} - 1,
\end{equation*}
hence \cref{eq:introduction:DMax_path_graphs}.
\qed
\end{proof}

\cref{eq:introduction:DMax_path_graphs} had been derived previously  \parencite[Appendix A]{Ferrer2021a}. The value was obtained by constructing a maximum arrangement using a result by \textcite{Chao1992a} which characterizes maximum and minimum permutations $\bm{\alpha}$ of $k$ numbers $A_k=\{\alpha_1,\dots,\alpha_k\}$ valued with a linear, convex or concave (strictly) increasing function $f$. The cost of a permutation was defined by \textcite{Chao1992a} as
\begin{equation*}
L_f(\bm{\alpha}) = \sum_{i=1}^{k-1} f(|\alpha_i - \alpha_{i+1}|).
\end{equation*}
The construction by \textcite[Appendix A]{Ferrer2021a} to calculate $\DMax{\pathgraph}$ consisted of setting $\{\alpha_1,\dots,\alpha_k\}$ to $\{1,\dots,n\}$ (where number $i$ denotes the $i$-th vertex of $\pathgraph$, and $\{i,i+1\}\in E(\pathgraph)$, $1\le i<n$) and setting $f$ to the identity function $g(x)=x$. This construction yielded a maximal bipartite arrangement.
%-----------------------------------------------------%
% </automatic inline of '8B-proof-classes-paths.tex'> %
%-----------------------------------------------------%
%------------------------------------------------------%
% <automatic inline of '8C-proof-classes-spiders.tex'> %
%------------------------------------------------------%
\section{Proof of \cref{thm:max_for_classes:k_linear:spider:no_nonbipartite_arrs}}
\label{sec:appendix:proof:classes:spiders}

Prior to proving that spider graphs are maximizable only by bipartite arrangements we define a simple, yet useful tool concerning the displacement of a single vertex.

\paragraph{Moving a single vertex} Here we want to study the rate of change of the sum of the lengths of the edges incident to a vertex $u$ when it is moved to the left or to the right in the arrangement (by consecutive swaps with its neighboring vertex), assuming that $u$ is not adjacent to any of the vertices it is swapped with. We denote this sum as
\begin{equation*}
\iD{u} = \sum_{uw\in E} \dl{uw}.
\end{equation*}
In any arrangement $\arr = (\dots, u, v, \dots)$ of a graph $\graph=(V,E)$, where $uv\notin E$, moving $u$ one position to the right (that is, swapping $u$ and $v$) produces an arrangement $\arr' = (\dots, v, u, \dots)$ such that
\begin{equation*}
\iD{u}[\arr'] - \iD{u} = -\level{u}.
\end{equation*}

% ------------------------------------------------------------------------------
We now show that maximum arrangements of spider graphs are all bipartite arrangements. Let $\spider\in\spiderclass$ be any spider graph. Let $\h$ be its hub, the only vertex such that $\degreeh=\degree{\h}\ge 3$, and let $\arr$ be a maximum arrangement of $\spider$. Recall that every maximum arrangement must satisfy \Nurse{} and that the vertices in the legs of the spider (which are antennas emanating from $\h$ of length $k-1$ in $\spider$) cannot be thistle vertices due to \pathoptimizationO. Therefore, in order to prove the claim it only remains to prove that $\h$ is not a thistle vertex in $\arr$. The possible values for the level of $\h$ are: $\pm\degreeh$, $\pm2$, $\pm1$, $0$. We prove the claim by dealing with each case separately. Due to symmetry, we deal only with the positive values.

\begin{enumerate}[(i)]
\item \label{thm:appendix:proof:classes:spider:case_level_Delta} $\level{\h}=\degreeh$. In this case, $\h$ is not a thistle vertex and $\arr$ is a bipartite arrangement.

\item \label{thm:appendix:proof:classes:spider:case_level_2} $\level{\h}=2$. Notice that $\arr$ cannot be maximum since there must be a vertex $u\in\neighbors{\h}$ such that $\level{u}=2$, which contradicts \cref{propos:properties_max_arrs:known:no_neighs_in_same_level}.

\item \label{thm:appendix:proof:classes:spider:case_level_0} $\level{\h}=0$. We show that such an arrangement is not maximum by obtaining one with higher cost. Now, due to \Nurse{} we assume w.l.o.g. that the starting distribution of the hub's neighbors is as indicated in \cref{fig:appendix:proof:classes:spider:hub_displacement:level_0}(O). In the whole of \cref{fig:appendix:proof:classes:spider:hub_displacement:level_0}, the values $\hptwo$, $\hpone$, $\hmone$ and $\hmtwo$ indicate both the set and number of neighbors of $\h$ of level $+2$, $+1$, $-1$, and $-2$ respectively; the values $\nptwo$, $\npone$, $\nmone$ and $\nmtwo$ indicate both the set and number of non-neighbor vertices of $\h$, with level $+2$, $+1$, $-1$, and $-2$ respectively. Therefore $\neighbors{\h}=\hptwo\cup\hpone\cup\hmone\cup\hmtwo$, and $V=\{\h\}\cup\hptwo\cup\hpone\cup\hmone\cup\hmtwo\cup\nptwo\cup\npone\cup\nmone\cup\nmtwo$. Since we assumed that $\level{\h}=0$, we have $\hptwo+\hpone=\hmone+\hmtwo$. Notice we cannot have $\hptwo=\hpone=0$ or $\hmtwo=\hmone=0$. Moreover, the values $\npone$ and $\nptwo$ denote the number of vertices of level $1$ and $2$, respectively, in $\arr$ not connected to $\h$ with an edge in $\spider$. \cref{fig:appendix:proof:classes:spider:hub_displacement:level_0} illustrates a sequence of steps (labeled from (A) to (E)) to move $\h$ from the `middle' of $\arr$ (O) to the leftmost end of $\arr$ (E).\footnote{Symmetric steps can be applied to move $\h$ to the rightmost end.} In \cref{fig:appendix:proof:classes:spider:hub_displacement:level_0}, every `local' increment is indicated next to each set of edges involved in the step. We denote the resulting arrangement (E) with $\arr'$. The sums of local increments at each step of \cref{fig:appendix:proof:classes:spider:hub_displacement:level_0} are the following.
\begin{flushleft}
\begin{tabular}{rll}
A)	& $A - O = \hpone^2 \ge 0$, 					& \\
B)	& $B - A = \npone(2\hpone - 1)$,				& \\
C)	& $C - B = \hptwo(\hptwo + 2\hpone - 1) \ge 0$,	& because we cannot have $\hptwo=\hpone=0$, \\
D)	& $D - C = \npone\hptwo \ge 0$,					& \\
E)	& $E - D = \nptwo(\degreeh - 2) \ge 0$,			& because $\degreeh\ge3.$
\end{tabular}
\end{flushleft}
Notice that in step (D), the value $\iD{v}$ for those $v\in\hptwo$ does not change. Now, it is easy to see that when $\hpone>0$ we have $A-O>0$ and $B-A\ge0$. Therefore, when $\hpone>0$, $\arr'$ always has a higher cost than $\arr$. In case $\hpone=0$ it easily follows that the total sum of the increments is $R=\npone(\hptwo - 1) + \hptwo(\hptwo - 1) + \nptwo(\degreeh - 2)$. Since $\hpone=0$ we must have $\hptwo>0$ and thus $R\ge0$. Whether $R=0$ or $R>0$, there will be at least one thistle vertex in $\arr'$, located in the legs, which arose from moving the vertices in $\hptwo$, therefore $\arr'$ cannot be maximum (\pathoptimizationO) and thus neither is $\arr$.\footnote{Notice that in order to prove that $R\ge0$ we needed the `readjustment' step in (D) where we move vertices in $\hptwo$ to the region in the arrangement corresponding to vertices of level $0$. Without it, we would have $R'=-\npone + \hptwo(\hptwo - 1) + \nptwo(\degreeh - 2)$ thus having to deal with the negative term $-\npone$ unfolding an unnecessarily complex analysis.}
\begin{figure}
	\centering
	\includegraphics{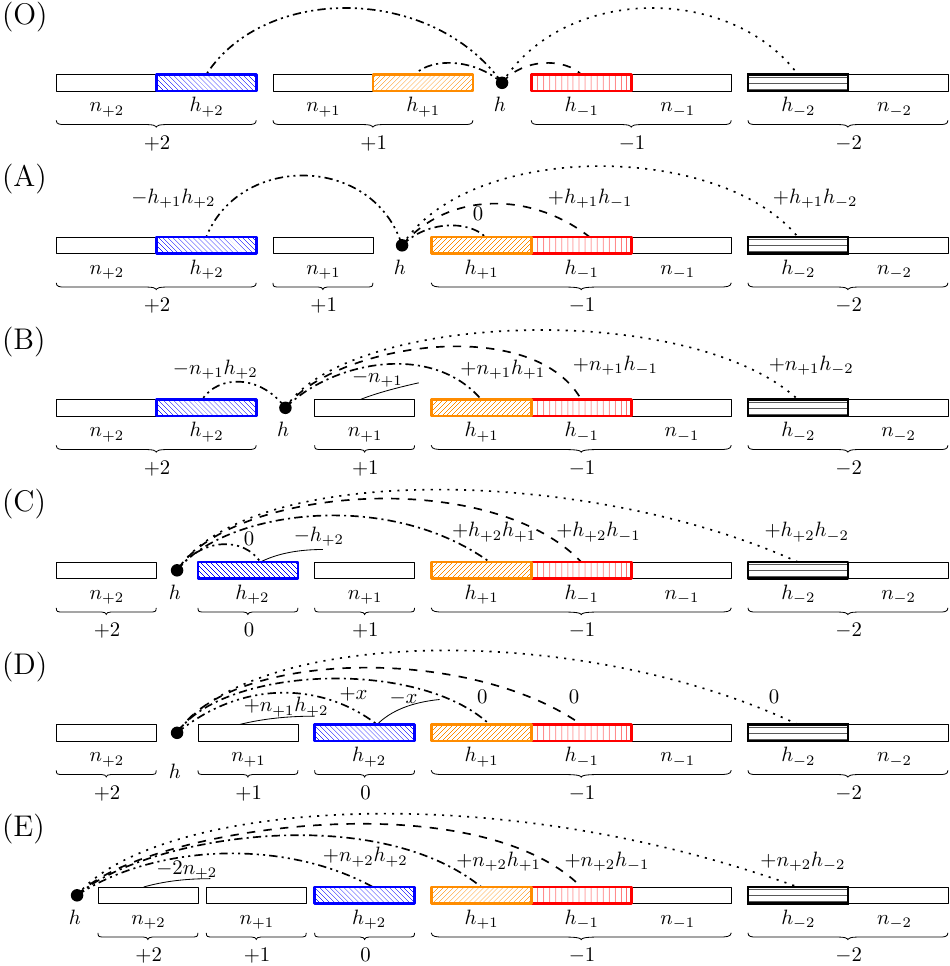}
	\caption{Proof of \cref{thm:max_for_classes:k_linear:spider:no_nonbipartite_arrs}. Steps to increase the cost of an arrangement $\arr$ of a spider graph $\spider\in\spiderclass$ with $\level{\h}=0$. At every step we indicate the increase (or decrease) in length of the edges with respect the previous step. The values $\hptwo$, $\hpone$, $\hmone$ and $\hmtwo$ indicate both the set and number of neighbors of $\h$ of level $+2$, $+1$, $-1$, and $-2$ respectively. The values $\nptwo$, $\npone$, $\nmone$ and $\nmtwo$ indicate both the set and number of vertices in the tree non-neighbors of $\h$ with level $+2$, $+1$, $-1$, and $-2$ respectively.}
	\label{fig:appendix:proof:classes:spider:hub_displacement:level_0}
\end{figure}

\item \label{thm:appendix:proof:classes:spider:case_level_1} $\level{\h}=1$. This case is similar to case (\ref{thm:appendix:proof:classes:spider:case_level_0}) and thus we only show the initial and final arrangements in \cref{fig:appendix:proof:classes:spiders:hub_displacement:level_1}. The format of \cref{fig:appendix:proof:classes:spiders:hub_displacement:level_1} is the same as that of \cref{fig:appendix:proof:classes:spider:hub_displacement:level_0}, except that now $\neighbors{\h}=\hptwo\cup\hmone\cup\hmtwo$ and $\hptwo+1=\hmone+\hmtwo$ since we assumed that $\level{\h}=1$. The increments we obtain for every step in \cref{fig:appendix:proof:classes:spiders:hub_displacement:level_1} are the following
\begin{flushleft}
\begin{tabular}{rll}
A)	& $A - O = \hptwo^2 \ge 0$, 			& \\
B)	& $B - A = \npone\hptwo \ge 0$,			& \\
C)	& $C - B = \hptwo(\degreeh - 2) \ge 0$,	& because $\degreeh\ge3$. \\
\end{tabular}
\end{flushleft}
The analysis is straightforward. We only need to show that in this case $\hptwo>0$. This follows from the fact that if $\hptwo=0$ then $1=\hmone+\hmtwo$. But that implies $\hmone=1$ and $\hmtwo=0$, or the other way around, which cannot satisfy $\degreeh=\hmone+\hmtwo\ge3$. Thus it is sufficient to see that the first and last increments are strictly positive and thus $\arr$ cannot be maximum.
\begin{figure}
	\centering
	\includegraphics{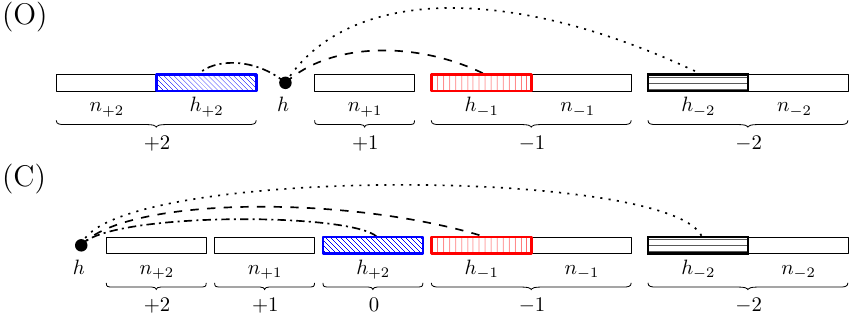}
	\caption{Proof of \cref{thm:max_for_classes:k_linear:spider:no_nonbipartite_arrs}. Initial and final arrangements of case (\ref{thm:appendix:proof:classes:spider:case_level_1}). Intermediate steps A and B are not shown. The format of the figure is the same as that of \cref{fig:appendix:proof:classes:spider:hub_displacement:level_0}.}
	\label{fig:appendix:proof:classes:spiders:hub_displacement:level_1}
\end{figure}
\end{enumerate}

The case-by-case analysis above show that a maximum arrangement of $\spider$ cannot have thistle vertices in a maximum arrangement of its vertices.
% ------------------------------------------------------------------------------
%-------------------------------------------------------%
% </automatic inline of '8C-proof-classes-spiders.tex'> %
%-------------------------------------------------------%
%---------------------------------------------------------%
% <automatic inline of '8D-proof-classes-two-linear.tex'> %
%---------------------------------------------------------%
\section{Proof of \cref{thm:max_for_classes:k_linear:two_linear:single_thistle}}
\label{sec:appendix:proof:classes:two_linear_trees}

Here we prove \cref{thm:max_for_classes:k_linear:two_linear:single_thistle}. This is simply a proof that the hubs cannot be thistle vertices in maximum arrangements, and thus only a vertex of the bridge can. In order to do this, we simply take one of the two hubs, assume a specific level value, and then move it from its pre-assumed location in the arrangement to an end of the arrangement while showing that the cost increases. Since there are many combinations of pairs of level values for the two hubs, we first classify all the cases, and reduce them to just 3 classes of cases, and solve each class independently.

Let $\twolinear\in\twolinearclass$ be a $2$-linear tree, let $\h$ and $\g$ be its two hubs, let $\degreeh=\degree{\h}\ge3$ and $\degreeg=\degree{\g}\ge3$ be their degrees, respectively, and let $\arr$ be any of its maximum arrangements. Let $P$ be the bridge of $\twolinear$, the only branchless path that is not an antenna; let $k$ be its length in edges. Notice that in order to prove the claim we only need to prove that $\h$ nor $\g$ are not thistle vertices in $\arr$ since (\pathoptimizationO) the vertices located in the legs of $\twolinear$ cannot be thistle vertices in $\arr$ and (\pathoptimizationOO) at most one internal vertex of $P$ can be a thistle vertex.

The proof is very similar to the proof of \cref{thm:max_for_classes:k_linear:spider:no_nonbipartite_arrs} (\cref{sec:appendix:proof:classes:spiders}), but we apply it to the hubs $\h$ and $\g$. Hub $\h$ can only have level values $\pm\degreeh$, $\pm2$, $\pm1$, $0$, and hub $\g$ can only have level values $\pm\degreeg$, $\pm2$, $\pm1$, $0$. We do not consider the cases when $\h$ and $\g$ have level $\pm\degreeh$ and $\pm\degreeg$ at the same time. Moreover, recall that the cases $\pm2$ are not possible in maximum arrangements for the same reasons explained in part (\ref{thm:appendix:proof:classes:spider:case_level_2}) of the proof of \cref{thm:max_for_classes:k_linear:spider:no_nonbipartite_arrs} (in \cref{sec:appendix:proof:classes:spiders}). We now list the cases $\proofcase_i=(a,b)$ to be solved based on the level values $a$ and $b$ of the leftmost hub and the rightmost hub.
\begin{center}
\begin{tabular}{rcr|rcr|rcr|rcr}
$\proofcase_1$ & $(\degreeg,1)$  & \proofcasepone &     $\proofcase_4$ & $(1,1)$         & \proofcasemone &                  &                 &           &                  &                  &           \\
$\proofcase_2$ & $(\degreeg,0)$  & \proofcasezero &     $\proofcase_5$ & $(1,0)$         & \proofcasezero &     $\proofcase_8$    & $(0,0)$         & \proofcasezero &                  &                  &           \\
$\proofcase_3$ & $(\degreeg,-1)$ & \proofcasemone &     $\proofcase_6$ & $(1,-1)$        & \proofcasemone &     $\proofcase_9$    & $(0,-1)$        & \proofcasemone &     $\proofcase_{11}$ & $(-1,-1)$        & \proofcasemone \\
          &                 &           &     $\proofcase_7$ & $(1,-\degreeh)$ & \proofcasemone &     $\proofcase_{10}$ & $(0,-\degreeh)$ & \proofcasezero &     $\proofcase_{12}$ & $(-1,-\degreeh)$ & \proofcasepone \\
\end{tabular}
\end{center}
In the list of cases above we assume that $\g$ is the leftmost hub and $\h$ is the rightmost hub. The circled letters \proofcasemone, \proofcasezero, \proofcasepone{} indicate the group of cases to which each particular case belongs. The cases are grouped based on how they are solved. In order to solve a case $\proofcase$, that is, prove that the cost of an arrangement $\arr$ of $\rho$ can be improved, $\arr$ is first mirrored ($m$), if needed, so that the level of the rightmost hub has the same level value as the rightmost hub of the other cases in the same group.
\begin{itemize}
\item[\proofcasemone] The rightmost hub has level $-1$ (\cref{fig:appendix:proof:classes:two_linear:single_thistle:cases_to_solve}(a)),
\item[\proofcasezero] The rightmost hub has level $0$ (\cref{fig:appendix:proof:classes:two_linear:single_thistle:cases_to_solve}(b)),
\item[\proofcasepone] The rightmost hub has level $+1$ (\cref{fig:appendix:proof:classes:two_linear:single_thistle:cases_to_solve}(c)).
\end{itemize}
Second, the resulting rightmost hub is moved to the desired end of the arrangement: either to the leftmost end of the arrangement ($\leftarrow$) or to the rightmost end of the arrangement ($\rightarrow$). Now follows, for each case, the operations to be performed in order to solve them.
\begin{center}
\begin{tabular}{rcl|rcl|rcl}
\multicolumn{3}{c|}{\proofcasemone}                    & \multicolumn{3}{c|}{\proofcasezero}                      & \multicolumn{3}{c}{\proofcasepone}                       \\
           &                 &                    &              &                 &                    &              &                  &                   \\
 $\proofcase_3$ & $(\degreeg,-1)$ & $\rightarrow$      & $\proofcase_2$    & $(\degreeg,0)$  & $\rightarrow$      & $\proofcase_1$    & $(\degreeg,1)$   & $\leftarrow$      \\
 $\proofcase_4$ & $(1,1)$         & $m$, $\rightarrow$ & $\proofcase_5$    & $(1,0)$         & $\rightarrow$      & $\proofcase_{12}$ & $(-1,-\degreeh)$ & $m$, $\leftarrow$ \\
 $\proofcase_6$ & $(1,-1)$        & $\rightarrow$      & $\proofcase_8$    & $(0,0)$         & $\rightarrow$      &              &                  &                   \\
 $\proofcase_7$ & $(1,-\degreeh)$ & $m$, $\rightarrow$ & $\proofcase_{10}$ & $(0,-\degreeh)$ & $m$, $\rightarrow$ &              &                  &                   \\
 $\proofcase_9$ & $(0,-1)$        & $\rightarrow$      &              &                 &                    &              &                  &                   \\
 $\proofcase_{11}$ & $(-1,-1)$    & $\rightarrow$      &              &                 &                    &              &                  &                   \\
\end{tabular}
\end{center}
All groups are shown in \cref{fig:appendix:proof:classes:two_linear:single_thistle:cases_to_solve} which is organized in pairs of initial and final arrangements. In that figure we use $\g$ and $\h$ to label the leftmost and rightmost hubs, respectively. Moreover, the values $\hptwo$, $\hpone$, $\hmone$ and $\hmtwo$ indicate both the set and number of neighbors of $\h$ of level $+2$, $+1$, $-1$, and $-2$ respectively; none of these vertices is in $P$. The values $\nptwo$, $\npone$, $\nmone$ and $\nmtwo$ indicate both the set and number of vertices non-neighbors of $\h$ with level $+2$, $+1$, $-1$, and $-2$ respectively. $\nzero$ denotes both the set and number of vertices of level $0$ not including $\h$. Notice that $\nzero\subseteq\{\g,\w\}$, where $\w$ is one of the internal vertices of $P$ which, due to \pathoptimizationOO, can have level $0$ in a maximum arrangement. If $k>2$ then we can assume, w.l.o.g. thanks to \pathoptimizationOO, that $\w$ is a neighbor of $\g$ and thus $\nzero\cap\neighbors{\h}=\emptyset$. We deal with the cases $k=1$ and $k=2$ in the analysis of each group.

Now we show that for each group of arrangements, the cost can always be improved. As in the proof of \cref{thm:max_for_classes:k_linear:spider:no_nonbipartite_arrs}, for each group we move a hub to a desired location where it is no longer a thistle vertex and show that either the result has a higher cost, or the cost is the same but the result is not maximum; while moving the hub, we rearrange the necessary vertices to occupy the appropriate positions in the arrangement according to their level value. In each group, we denote the initial arrangement as $\arr$, and the final arrangement as $\arr'$. In the following analyses, we use $a_{xy}\in\{0,1\}$ to denote whether or not $xy\in E$ for any two vertices $x,y$.
\begin{itemize}
\item[\proofcasemone] $\level{\h}=-1$; it is illustrated in \cref{fig:appendix:proof:classes:two_linear:single_thistle:cases_to_solve}(a) and has three steps ($A$ to $C$). Here, $\hptwo\cup\hpone\cup\hmtwo\subseteq\neighbors{\h}$ and $\hptwo+\hpone=1+\hmtwo$. If $\level{\g}=-1$, then $\g\in\nmone$ and the position of $\g$ is affected when the vertices in $\hmtwo$ are readjusted. However, $\g$ is moved to the right, whereby we deduce that $\iD{\g}[\arr']>\iD{\g}$. Now, if $k=1$ and $\level{\g}=0$ then $\nzero=\{\g\}$ and $\hptwo\cup\hpone\cup\hmtwo\cup\nzero=\neighbors{\h}$; if $k=1$ and $\level{\g}\neq0$ then $\nzero=\emptyset$ and $\hptwo\cup\hpone\cup\hmtwo=\neighbors{\h}$; if $k=2$ then $\w\in\nzero$ and $\awh=1$; plus, $\dl{\w\h}[\arr']-\dl{\w\h}\ge 0$. Disregarding the increment of $\iD{\g}[]$ (when $\level{\g}=-1$) and the (strictly positive) increment of the length of $\w\h$ (except in the last step), now follow lower bounds of the increments in every step.
\begin{flushleft}
\begin{tabular}{rll}
A)	& $A - O \ge \agh\hmtwo + \hmtwo^2 \ge 0$, 			& \\
B)	& $B - A \ge \nmone\hmtwo \ge 0$,					& \\
C)	& $C - B \ge \agh\nmtwo + \nmtwo(\degreeh - 2) \ge 0$,& because $\degreeh\ge3$. \\
\end{tabular}
\end{flushleft}
Notice that the vertices in $\hmtwo$ have now level $0$ in $\arr'$. Since we must have $\hmtwo>0$, $\arr'$ cannot be maximum due to \pathoptimizationO. Therefore, since the cost of $\arr'$ is at least the cost of $\arr$, then $\arr$ cannot be maximum either.

\item[\proofcasezero] $\level{\h}=0$; it is illustrated in \cref{fig:appendix:proof:classes:two_linear:single_thistle:cases_to_solve}(b) and has six steps ($A$ to $F$). Here, $\neighbors{\h}=\hptwo\cup\hpone\cup\hmone\cup\hmtwo$ and $\hptwo+\hpone=\hmone+\hmtwo$. If $\level{\g}=0$, the position of $\g$ is affected when the neighbors $\hmone$ are readjusted. However, when $\level{\g}=0$, $\iD{\g}[\arr']=\iD{\g}$. If $k=1$ then $\nzero=\emptyset$ (\cref{propos:properties_max_arrs:known:no_neighs_in_same_level}) and thus $\level{\g}\ge1$ and the position of $\g$ is never affected. In case $k=2$ then $\w\in\nzero$ and $\awh=1$. Notice that the length of edge $\w\h$ would only increase in case it existed. Disregarding the (strictly positive) increment of the length of edge $\w\h$ (except in the last step), now follow lower bounds of the increments in every step.
\begin{flushleft}
\begin{tabular}{rll}
A)	& $A - O \ge \agh\hmone + \hmone^2 \ge 0$, 					& \\
B)	& $B - A \ge \nzero\hmtwo\ge0$,								& \\
C)	& $C - B \ge \agh\nmone + \nmone(2\hmone - 1)$,				& \\
D)	& $D - C \ge \agh\hmtwo + \hmtwo(\hmtwo + 2\hmone - 1)\ge0$,& because we cannot have $\hmtwo=\hmone=0$, \\
E)	& $E - D \ge \nmone\hmtwo\ge0$,								& \\
F)	& $F - E \ge \agh\nmtwo + \nmtwo(\degreeh - 2)\ge0$,		& because $\degreeh\ge3$. \\
\end{tabular}
\end{flushleft}
If $\hmone>0$ then $A-O>0$ and $C-B>0$. If $\hmone=0$ then $\hmtwo>0$ and the total sum of increments is $R=\nmone(\hmtwo - 1 + \agh) + \hmtwo(\hmtwo - 1) + \nmtwo(\degreeh - 2) + \agh\hmtwo + \agh\nmtwo$. It is easy to see that $R\ge0$. Since $\arr'$ is not maximum (due to the vertices in $\hmtwo$ having level 0 in $\arr'$ and \pathoptimizationO) then $\arr$ is not maximum.

\item[\proofcasepone] $\level{\h}=1$; it is illustrated in \cref{fig:appendix:proof:classes:two_linear:single_thistle:cases_to_solve}(c) and has four steps ($A$ to $D$). Here, $\level{\g}=\degreeg$,  and $\hptwo+1=\hmone+\hmtwo$. The position of $\g$ is always affected since we assume that for any $2$-linear tree, $\degreeh\ge\degreeg$. Here, if $k=1$ then $\nzero=\emptyset$ and $\neighbors{\h}=\hptwo\cup\hpone\cup\hmone\cup\hmtwo$. If $k=2$ then $\h$ may have a neighbor $w\in\nzero$, and, as before, $\dl{\w\h}[\arr']-\dl{\w\h}\ge 0$. Disregarding this edge (except in the second-to-last step), the increments of the four steps we identified are the following.
\begin{flushleft}
\begin{tabular}{rll}
A)	& $A - O \ge -\agh\hptwo + \hptwo^2 \ge 0$, 			& \\
B)	& $B - A \ge \npone\hptwo^2\ge0$,						& \\
C)	& $C - B \ge -\agh\nptwo + \nptwo(\degreeh - 2)\ge0$,	& because $\degreeh\ge3$, \\
D)	& $D - C \ge \degreeh - \degreeg\ge0$,					& because $\degreeh\ge\degreeg$. \\
\end{tabular}
\end{flushleft}
As before, since $\arr'$ is not maximum (due to the vertices in $\hmtwo$ having level 0 in $\arr'$ and \pathoptimizationO) and the cost of $\arr'$ is at least the cost of $\arr$, then $\arr$ is not maximum.
\end{itemize}

\begin{figure}
	\centering
	\includegraphics[scale=0.975]{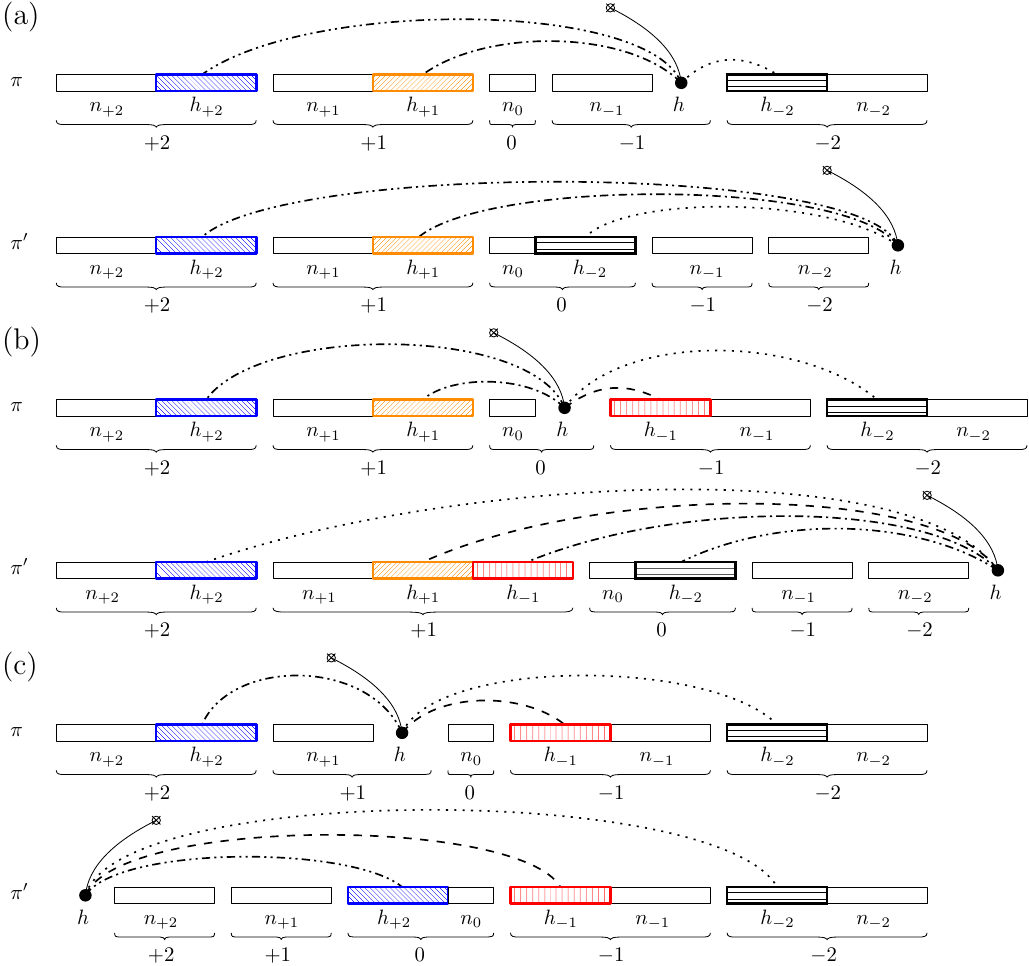}
	\caption{Proof of \cref{thm:max_for_classes:k_linear:two_linear:single_thistle}. Pairs of arrangement before moving $\h$, and final arrangement after moving $\h$. In the figure, hub $\g$ is not shown, but is assumed to be such that $\arr(\g)<\arr(\h)$; the edge connecting $\g$ and $\h$ (when it exists in the tree) is depicted with a solid line emanating from $\h$ and ending at a crossed circle. The values $\hptwo$, $\hpone$, $\hmone$ and $\hmtwo$ indicate both the set and number of neighbors of $\h$ of level $+2$, $+1$, $-1$, and $-2$ respectively. The values $\nptwo$, $\npone$, $\nmone$ and $\nmtwo$ indicate both the set and number of vertices non-neighbors of $\h$ with level $+2$, $+1$, $0$, $-1$, and $-2$ respectively. (a) $\level{\h}=-1$. (b) $\level{\h}=0$. (c) $\level{\h}=
	1$.}
	\label{fig:appendix:proof:classes:two_linear:single_thistle:cases_to_solve}
\end{figure}

The case-by-case analysis above shows that an arrangement $\arr$ cannot be maximum if the hubs of $\twolinear$ are thistle vertices in $\arr$; thus the hubs cannot be thistle vertices.
%----------------------------------------------------------%
% </automatic inline of '8D-proof-classes-two-linear.tex'> %
%----------------------------------------------------------%

%%%%%%%%%%%%%%%%%%%%%%%%%%%%%%%%%%%%%%%%%%%%%%%%%%%%%%%%%%%%%%%%%%%%%%%%%%%%%%%%

% ----------------------------------------------------------------------
% Bibliography
\printbibliography

\end{document}